\documentclass[11pt, a4paper, reqno]{amsart}
\usepackage{amsthm, amsmath, amsfonts, amssymb, appendix, dsfont, latexsym}
\usepackage{amsfonts, mathrsfs}
\usepackage{tikz}
\usetikzlibrary{arrows}

\makeatletter
\usepackage{hyperref} 
\usepackage{delarray, a4, color}
\usepackage{euscript}
\usepackage[latin1]{inputenc}

 \usepackage[hmarginratio=1:1,textwidth=6truein]{geometry}
\theoremstyle{plain}
\newtheorem{theorem}{Theorem}[section]
\newtheorem{lemma}[theorem]{Lemma}

\newtheorem{proposition}[theorem]{Proposition}

\theoremstyle{remark}
\newtheorem{remark}[theorem]{Remark}

\numberwithin{equation}{section}

\def\eps{\varepsilon}

\newcommand\1{{\ensuremath {\mathds 1} }}
\def\NN {\mathbb{N}}
\def\R {\mathbb{R}}
\def\C {\mathbb{C}}
\def\Q {\mathbb{Q}}
\def\Z {\mathbb{Z}}
\def\S {\mathbb{S}}
\def\N {\EuScript{N}}

\newcommand{\bkappa}{{\boldsymbol{\kappa}}}
\newcommand{\ba}{\mathbf{a}}

\newcommand{\br}{\mathbf{r}}
\newcommand{\bx}{\mathbf{x}}
\newcommand{\by}{\mathbf{y}}
\newcommand{\bz}{\mathbf{z}}
\newcommand{\bA}{\mathbf{A}}
\newcommand{\bJ}{\mathbf{J}}
\newcommand{\bX}{\mathbf{X}}
\newcommand{\cD}{\mathscr{D}}

\newcommand{\sx}{\textup{x}}
\newcommand{\sz}{\textup{z}}

\newcommand{\domD}[1]{\cD^N_{#1}} 
\newcommand{\dist}{\mathrm{dist}}

\newcommand{\nablap}{\nabla^{\perp}}
\newcommand{\eSR}{e_{\mathrm{SR}}}
\newcommand{\eLR}{e_{\mathrm{LR}}}
\DeclareMathOperator{\curl}{\mathrm{curl}}
\DeclareMathOperator{\supp}{supp}

\newcommand{\sym}{\mathrm{sym}}
\newcommand{\asym}{\mathrm{asym}}
\newcommand{\loc}{\mathrm{loc}}

\newcommand{\vphi}{\varphi}

\renewcommand{\phi}{\varphi}

\newcommand{\bDelta}{{\mbox{$\triangle$}\hspace{-8.0pt}\scalebox{0.8}{$\triangle$}}}

\newcommand{\limplus}{{\mathchoice{\vcenter{\hbox{$\scriptstyle +$}}}
  {\vcenter{\hbox{$\scriptstyle +$}}}
  {\vcenter{\hbox{$\scriptscriptstyle +$}}}
  {\vcenter{\hbox{$\scriptscriptstyle +$}}}
}}
\newcommand{\limminus}{{\mathchoice{\vcenter{\hbox{$\scriptstyle -$}}}
  {\vcenter{\hbox{$\scriptstyle -$}}}
  {\vcenter{\hbox{$\scriptscriptstyle -$}}}
  {\vcenter{\hbox{$\scriptscriptstyle -$}}}
}}
\newcommand{\limpm}{{\mathchoice{\vcenter{\hbox{$\scriptstyle \pm$}}}
  {\vcenter{\hbox{$\scriptstyle \pm$}}}
  {\vcenter{\hbox{$\scriptscriptstyle \pm$}}}
  {\vcenter{\hbox{$\scriptscriptstyle \pm$}}}
}}

\usepackage{mathtools}
\mathtoolsset{showonlyrefs=true}

\usepackage{placeins}

\title[Exclusion bounds for extended anyons]{Exclusion bounds for extended anyons} 

\author[S. Larson]{Simon LARSON}

\author[D. Lundholm]{Douglas LUNDHOLM}
\address{Department of Mathematics, KTH Royal Institute of Technology, SE-100 44 Stockholm, Sweden}

\email{simla@math.kth.se, dogge@math.kth.se}


\begin{document}

\begin{abstract}
We introduce a rigorous approach to the many-body spectral theory 
of extended anyons, i.e.\ quantum particles confined to 
two dimensions that interact via attached magnetic fluxes 
of finite extent. Our main results are many-body magnetic Hardy 
inequalities and local exclusion principles for these particles, 
leading to estimates for the ground-state energy of the anyon gas 
over the full range of the parameters. This brings out further 
non-trivial aspects in the dependence on the anyonic statistics 
parameter, and also gives improvements in the ideal (non-extended) case.
\end{abstract}

\maketitle
\setcounter{tocdepth}{2}
\tableofcontents

\section{Introduction}\label{sec:intro}

In many-body quantum mechanics, the notion of particle 
indistinguishability and statistics plays a fundamental role. 
Namely, particles of the same kind are typically logically identical 
and fall into two classes: bosons or fermions, giving rise to such diverse
phenomena as Bose--Einstein condensation and coherent propagation of light 
in the former case, and the Fermi sea with its implications for conduction bands, 
atomic structure, etc., in the latter. However, while these are the only two 
options for fundamental particles that propagate in three-dimensional space, 
for quantum systems confined to lower dimensions there is a possibility
for effective particles (quasiparticles) escaping the usual boson/fermion dichotomy.
We shall here consider the two-dimensional case where the quantum state of a system of 
$N$ particles at positions $\bx_j \in \R^2$ may be described by a 
square-integrable, normalized, complex wave function\footnote{Here we restrict to the simplest
case of $\C$-valued wave functions corresponding to \emph{abelian} anyons, 
while $\C^n$-valued, possibly \emph{non-abelian}, anyons are also possible 
\cite{Froehlich-90,Nayak-etal-08}.}
$\Psi\colon \R^{2N} \to \C$, where  
$|\Psi(\sx)|^2$ is interpreted as
the probability density of 
finding the particles at positions $\sx = (\bx_1, \ldots, \bx_N)$.
If the particles are indistinguishable the density needs to be symmetric
under permutations of the particle labels:
\begin{equation}\label{eq:particleExchangeDensity}
	|\Psi(\bx_1, \ldots, \bx_j, \ldots, \bx_k, \ldots, \bx_N)|^2 
	= |\Psi(\bx_1, \ldots, \bx_k, \ldots, \bx_j, \ldots, \bx_N)|^2, 
	\quad j \neq k.
\end{equation}
However, the exact phase of $\Psi$ is not an observable quantity and therefore 
\eqref{eq:particleExchangeDensity} leaves room for an exchange phase:
\begin{equation}\label{eq:particleExchangePhase}
	\Psi(\bx_1, \ldots, \bx_j, \ldots, \bx_k, \ldots, \bx_N)
	= e^{i\alpha\pi} \Psi(\bx_1, \ldots, \bx_k, \ldots, \bx_j, \ldots, \bx_N), 
	\quad j \neq k, 
\end{equation}
where $\alpha \in \R$ ($2\Z$-periodic) is called the statistics parameter.
If $\alpha=0$ the particles are called bosons (symmetric $\Psi$), 
and if $\alpha=1$ they are fermions (antisymmetric $\Psi$). 
Because of the antisymmetry,
fermions obey Pauli's exclusion principle~\cite{Pauli-47} leading to Fermi--Dirac statistics,
while bosons do not, leading to Bose--Einstein statistics.
These are indeed the familiar possibilities found in introductory quantum mechanics textbooks,
however, upon investigating the argument more carefully
one realizes that one needs to be more precise with what is meant with the exchange
$j \leftrightarrow k$ in~\eqref{eq:particleExchangeDensity}-\eqref{eq:particleExchangePhase}.
Namely, the exchange should in fact be viewed as a continuous loop in the 
manifold of positions $\sx$ of $N$
identical particles, and then topology plays a crucial role. 
Thus we \emph{define}~\eqref{eq:particleExchangePhase}
to mean a continuous simple exchange of a single pair of particles
(in two dimensions counterclockwise and with no other particles enclosed; 
furthermore the exchange phase 
can be shown to be independent of which pair of particles is considered).
In three dimensions and higher, the direction of the exchange does not matter and
a double exchange is topologically the same as no exchange; therefore
the group of continuous exchanges reduces to the group of permutations 
and one ends up with the usual bosons or fermions.
In two spatial dimensions, on the other hand, the exchange group is the braid group and 
it then turns out that
\emph{any} phase $e^{i\alpha\pi} \in \mathrm{U}(1)$ in~\eqref{eq:particleExchangePhase} 
is allowed 
\cite{LeiMyr-77, GolMenSha-81, Wilczek-82a, Wilczek-82b, Wu-84} 
(see also~\cite[p.~386]{Souriau-70}).
The corresponding particles are therefore called anyons~\cite{Wilczek-82b}.
We refer to~\cite{Froehlich-90, Khare-05, Lerda-92, Myrheim-99, Nayak-etal-08, Ouvry-07, Stern-08, Wilczek-90} 
for extensive reviews on this topic.

The relative change of phase of the wave function $\Psi$ with respect to changes of the 
coordinates may be
geometrically understood as due to the curvature of a corresponding 
complex line bundle of which $\Psi$ is a section.
This is naturally described by a magnetic field, 
and in the case of anyons one may indeed
model the above statistics phase as induced by a magnetic field of Aharonov--Bohm type.
Namely, one could start with $\Psi \in L^2_{\sym}((\R^2)^N)$
(or $\Psi \in L^2_{\asym}((\R^2)^N)$) being bosonic (fermionic)
and then attach magnetic fluxes to the particles so that their winding
around each other gives rise to the correct phase \eqref{eq:particleExchangePhase}.
This is commonly called the magnetic gauge picture for anyons, 
and it is actually in this form that they may most realistically arise in a real 
physical system. The most promising such realization is in the context of
the fractional quantum Hall effect (FQHE)~\cite{Girvin-04, Goerbig-09, Jain-07, Laughlin-99, StoTsuGos-99}, 
a strongly correlated planar electron
(or bosonic atom~\cite{BloDalZwe-08, Cooper-08, MorFed-07, RonRizDal-11, Viefers-08})
system in a strong transverse external magnetic field, 
where particles have the freedom to bind magnetic flux and thereby become 
anyons~\cite{AroSchWil-84, Laughlin-99, LunRou-16}. 
However, in this scenario the flux typically has some extent determined
by the experimental conditions, and one therefore talks about \emph{extended} 
anyons~\cite{ChoLeeLee-92, LunRou-15, Mashkevich-96, Trugenberger-92b}
as opposed to the purely theoretical (but conceptually attractive) ideal\footnote{By 
ideal in this context we mean that the only interaction is statistical and 
independent of any energy, momentum or length scale (cf.~\cite[p.\ 146]{Khare-05}).}
anyons which are purely pointlike.
Denoting the size of the flux, say its radius if disk-shaped, 
by $R \ge 0$ we can thus talk about $R$-extended anyons, and one may also
introduce a dimensionless parameter $\bar\gamma := R\bar\varrho^{1/2}$
to describe the state of the system, 
where $\bar\varrho$ denotes the average density of the particles.
The parameter $\bar\gamma$ 
is the ratio of the magnetic dimension to the average interparticle distance and
has therefore been called the magnetic filling ratio
in~\cite{Trugenberger-92b, Trugenberger-92}.

Our interest in this paper is to study a free gas of such extended anyons, 
i.e.\ ignoring any additional interactions as a simplifying first step, and
focusing on the most basic aspect: its ground-state energy.
We consider this in the thermodynamic limit 
(cf.~\cite{CatBriLio-98,LieSeiSolYng-05}), 
that is the limit as both the number
of particles $N$ and the volume (area) of the system $V$
tends to infinity while keeping the density $\bar\varrho = N/V$ fixed.
In the ideal non-interacting case, 
the quantum gas consisting of a large number of bosons or fermions 
in a large volume at fixed density
has been completely understood since the early days of 
quantum mechanics and is nowadays often 
given as a textbook exercise,
as it only amounts to adding up eigenvalues of a one-body operator.
However, the purely anyonic case $\alpha \in (0, 1)$ 
still remains an unsolved problem
after almost four decades, owing to the fact that  
the statistical many-body interaction cannot be completely removed
in favor of a one-body description as for bosons and fermions.
The simplest case of two anyons can be solved 
exactly~\cite{LeiMyr-77, Wilczek-82b, AroSchWilZee-85}, 
that of three and four anyons has been studied 
numerically~\cite{SpoVerZah-91, MurLawBraBha-91, SpoVerZah-92}, 
and beyond that various approximative descriptions have been 
proposed~\cite{ChenWilWitHal-89, ChiSen-92, FetHanLau-89, IenLec-92, SenChi-92, Trugenberger-92b, Trugenberger-92, WenZee-90, Westerberg-93, Wilczek-90}.
One of these is called average-field theory (cf.\ mean-field theory~\cite{Wilczek-90, LunRou-15})
whereby the magnetic flux of the anyons is seen as sufficiently spread out
(in other words $\bar\gamma$ should be sufficiently large) 
so that the particles are effectively moving in a
(locally) uniform magnetic field, say $B(\bx) \sim 2\pi\alpha \varrho(\bx)$
where $2\pi\alpha$ is the flux of each anyon and $\varrho(\bx)$ the local density, 
and therefore have a definite magnetic ground-state energy given by 
that of the lowest Landau level, 
hence proportional to 
$|B| \sim 2\pi|\alpha|\varrho$.
In other words the energy per particle in this approximation 
is given by
\begin{equation} \label{eq:average-field}
	2\pi|\alpha| \bar\varrho
\end{equation}
in the case of the homogeneous gas. 
Another approximation has been to assume that the gas is so dilute that
only two-particle interactions are relevant~\cite{AroSchWilZee-85, Minor-93}.

Except for a small number of results
concerning the mathematical formulation of the many-anyon 
problem~\cite{BakCanMulSun-93, DelFigTet-97, DoeGroHen-99, LigMin-95}, 
there has not been much progress on the rigorous mathematical 
side until recently. 
In~\cite{LunSol-13a} the case of ideal anyons was considered using a local approach
involving a relative magnetic Hardy inequality and a local exclusion principle, 
leading to a first set of non-trivial rigorous bounds for the ground-state energy
of the ideal anyon gas. These bounds, which will be outlined below, 
have an interesting non-trivial dependence
on the statistics parameter $\alpha$ in that they depend, 
in the many-body limit, solely on the quantity
\begin{equation} \label{eq:C_alpha}
	\alpha_* := \inf\limits_{p, q \in \Z} \bigl|(2p+1)\alpha - 2q\bigr|, 
\end{equation}
which is zero unless $\alpha$ is an odd-numerator fraction 
$\alpha = \mu/\nu \in \Q$ (reduced, with $\nu \ge 1$) 
and in which case $\alpha_* = 1/\nu$.
In~\cite{LunSol-14} a fundamental question concerning operator domains 
for ideal anyons was settled and applications of the local energy bounds 
to interacting systems were considered.
Also, the validity of an average-field approximation for the case of 
almost-bosonic ($\alpha \to 0$) $R$-extended anyons was proved in~\cite{LunRou-15} 
(see also~\cite{CorLunRou-16}).

Here we shall consider the homogeneous $R$-extended anyon gas 
in the thermodynamic limit and build on 
the local approach of~\cite{LunSol-13a} to prove a lower bound 
for the ground-state energy per particle 
with statistics parameter $\alpha \in \R\setminus\{0\}$ and magnetic filling ratio 
$\bar\gamma = R\bar\varrho^{1/2} \ge 0$
of the form
$$
	C e(\alpha, \bar\gamma) \bar\varrho, 
$$
where $C>0$ is a universal constant and
(see Figure~\ref{fig:EnergyVsGammaPlots} below for intermediate values)
$$
	e(\alpha, \gamma) \sim \left\{ \begin{array}{ll}
		\frac{2\pi}{|{\ln \gamma}|} + \pi(j_{\alpha_*}')^2 \ge 2\pi \alpha_*, &\quad \gamma \to 0, \\[6pt]
		2\pi|\alpha|, &\quad \gamma \gtrsim 1. 
		\end{array}\right.
$$
Here $j_\nu'$ denotes the first positive zero of the derivative of the 
Bessel function $J_\nu$ (and~$j_0':=0$).
This bound effectively interpolates between a dilute regime involving~
\eqref{eq:C_alpha} and a high-density regime with a dependence on $\alpha$
matching that of average-field theory~\eqref{eq:average-field}.
Also in the case of even-numerator $\alpha$, where $\alpha_*=0$, the bound
is strictly positive but vanishes in the dilute limit in a way similar to 
that of a dilute Bose gas in two dimensions~\cite{Schick-71, LieYng-01}.
This may however not be so surprising in the case that $\alpha \in 2\Z$
(composite bosons; cf.~\cite{Jain-07}), 
considering the periodicity in the statistics parameter for ideal anyons.

\subsection{The extended anyons model}\label{sec:intro-prelims}

In order to state our results precisely we need to introduce some 
notation that will be used throughout the paper.

We take as our concrete model for $R$-extended anyons a set of $N$ identical 
bosons, to each of which has been attached a magnetic field in the
shape of a disk with radius $R$ and total flux $2\pi\alpha$, 
and which is felt by all the other particles 
(cf.~\cite{ChoLeeLee-92, LunRou-15, LunRou-16, Mashkevich-96, Trugenberger-92b}).
Such flux centered at the origin can be given explicitly  
by the magnetic vector potential $\alpha\bA_0$ with
$$
	\bA_0(\bx) := \frac{(\bx-\, \cdot\, )^\perp}{|\bx-\, \cdot\, |^2} * \frac{\1_{B_R(0)}}{\pi R^2}
	= \frac{\bx^\perp}{|\bx|_R^2}, 
	\qquad
	\curl \bA_0(\bx) = 2\pi\frac{\1_{B_R(0)}}{\pi R^2}(\bx).
$$
Here $(x, y)^\perp := (-y, x)$, i.e.\ a $\pi/2$ counterclockwise rotation, 
$B_R(\bx)$ denotes the open ball/disk of radius $R$ centered at $\bx \in \R^2$, and
$$
	|\bx|_R := \max\{|\bx|, R\},
$$
which can be interpreted as a regularized distance.
Starting from a conventional magnetic Hamiltonian formulation, 
the (non-relativistic) free kinetic energy operator is then
\begin{equation} \label{eq:kineticEnergyOp}
	\hat{T}_\alpha := \sum_{j=1}^N D_j^2, 
\end{equation}
where we have normalized physical units so that $\hbar^2/(2m) = 1$ and
the magnetically coupled momentum operator for each particle $j$ is given by
$$
	D_j := -i\nabla_{\bx_j} + \alpha\bA_j(\bx_j), 
$$
where
\begin{equation} \label{eq:many-anyon vec pot}
	\bA_j(\bx) 
	:= \frac{(\bx-\, \cdot\, )^\perp}{|\bx-\, \cdot\, |^2} * \sum_{k \neq j} \frac{\1_{B_R(\bx_k)}}{\pi R^2}
	= \sum_{k \neq j} \frac{(\bx-\bx_k)^\perp}{|\bx-\bx_k|_R^2}, 
\end{equation}
corresponding to the total magnetic field felt by the particle $\bx_j$
\begin{equation} \label{eq:many-anyon mag field}
	\curl \alpha\bA_j 
	= 2\pi\alpha \sum_{k \neq j} \frac{\1_{B_R(\bx_k)}}{\pi R^2}
	\ \ \stackrel{R \to 0}{\longrightarrow} \ \ 
	2\pi\alpha \sum_{k \neq j} \delta_{\bx_k}.
\end{equation}
We note that this form for the magnetic interaction is not only convenient 
but also realistic from the perspective of the FQHE~\cite{LunRou-16}.
Also note that we allow for any $\alpha \in \R$ here.

The operator \eqref{eq:kineticEnergyOp} acts on the bosonic Hilbert space
$L^2_\sym(\R^{2N})$ as an unbounded operator.
Let us denote by $\domD{\alpha, R}$ the natural 
(minimal as well as maximal~\cite[Theorem~5]{LunSol-14}) domain of 
the magnetic gradient
$$
	\begin{array}{rcl}
		D\colon L^2_\sym(\R^{2N};\C) &\to& L^2(\R^{2N};\C^{N}) \\
			\Psi &\mapsto& D\Psi 
				= (-i\nabla + \alpha\bA)\Psi 
				= \bigl( (-i\nabla_j + \alpha\bA_j)\Psi \bigr)_{j=1}^N, 
	\end{array}
$$
then this is also the natural form domain of \eqref{eq:kineticEnergyOp}, 
and $\hat{T}_\alpha := D^*D$.
In the case $R > 0$ (as well as for $\alpha=0$)
we have $\domD{\alpha, R} = H^1_{\sym}(\R^{2N})$, 
since $\bA$ is then a bounded perturbation of $-i\nabla$.
On the other hand, if $R = 0$ 
then $\bA$ is singular and
these spaces are typically different
(see~\cite[Section~2.2]{LunSol-14}).
For $R = 0$ and $\alpha \in 2\Z$ (respectively $\alpha \in 2\Z+1$), however, 
$\domD{\alpha, 0}$ is gauge-equivalent to 
$\domD{0, 0} = H^1_{\sym}(\R^{2N})$ 
(respectively $\domD{1, 0} = U^{-1} H^1_{\asym}(\R^{2N})$):
\begin{equation} \label{eq:gauge-equivalence}
	D_{(\alpha+2n)} = U^{-2n} D_{(\alpha)} U^{2n}, 
	\qquad
	\domD{\alpha+2n, 0} = U^{-2n} \domD{\alpha, 0}, 
	\qquad n \in \Z, 
\end{equation}
where $U$ is the isometry (singular gauge transformation)
$$
	U\colon L^2_{\sym/\asym} \to L^2_{\asym/\sym}, 
	\qquad (U\Psi)(\sx) := \prod_{1 \le j<k \le N} \frac{z_j-z_k}{|z_j-z_k|} \Psi(\sx),
$$
with $z_j$ the complex coordinate representatives of $\bx_j$ given by identifying $\R^2$ with $\C$.
In other words, for ideal anyons the spectrum of the operator $\hat{T}_\alpha$ 
is 2-periodic in $\alpha$, 
however we will find that this is not the case for extended anyons.

We define the one-body density associated with any normalized  
state $\Psi \in L^2(\R^{2N})$ by
$$
	\varrho_\Psi(\bx) := \sum_{j=1}^N \int_{\R^{2(N-1)}} |\Psi(\bx_1, \ldots, 
	\bx_{j-1}, \bx, \bx_{j+1}, \ldots, \bx_N)|^2 \prod_{k \neq j}d\bx_k, 
$$
with $\int_\Omega \varrho_\Psi$ the expected number of particles to be found on $\Omega \subseteq \R^2$,
while $\bar\varrho := N/|Q_0|$ denotes the average density if confined to a 
domain (typically a square) $Q_0 \subseteq \R^2$, 
i.e.\ for states $\Psi$ with $\supp \Psi \subseteq Q_0^N$. 
Furthermore, with
\begin{equation} \label{eq:diagonals}
	\bDelta := 
	\bigl\{ \sx \in (\R^2)^N : \text{$\exists\ j \neq k$ s.t. $\bx_j = \bx_k$} \bigr\}
\end{equation}
the fat diagonal of the configuration space $(\R^2)^N$, 
we note that we may use the density of
$C^\infty_c(\R^{2N} \setminus \bDelta) \cap L^2_\sym(\R^{2N})$ 
in the domain $\domD{\alpha, R}$ (again, see~\cite[Theorem~5]{LunSol-14}).

\subsection{Main bounds}\label{sec:intro-bounds}

We are now ready to state our main results for $R$-extended anyons.
For the reader's convenience we outline and compare to the previously studied ideal case,
which is also improved in several aspects in this work. 

Our study of the homogeneous anyon gas relies on two key insights which were
brought together in~\cite{LunSol-13a} for ideal anyons.
On the one hand, 
we follow an idea originally used 
by Dyson and Lenard 
in their proof of the stability of matter for fermionic Coulomb systems~\cite{DysLen-67}
(see also~\cite{Dyson-68, Lenard-73}).
They realized that the Pauli exclusion principle 
is strong enough (for many purposes, including the stability of matter)
acting only between pairs or small numbers of particles. 
It is in fact sufficient that the local kinetic energy is strictly positive for two particles
and that it grows at least linearly with the number of particles,
in contrast to the true ground-state energy for fermions which grows with $N$ according to the
Weyl asymptotics for the sum of Laplacian eigenvalues, i.e.\ as $N^{1+2/d}$ 
in dimension $d$.
We refer to such a bound as a local exclusion principle, and
the method has recently been generalized to interacting 
bosonic gases with the Pauli principle replaced by repulsive 
interactions~\cite{LunSol-13a, LunSol-13b, LunSol-14, LunPorSol-15, LunNamPor-16},
and to point-interacting fermionic gases~\cite{FraSei-12}. 
Essentially the idea is based on splitting the full domain to which the gas is 
confined into subdomains whose size is chosen so that the expected number of 
particles in each domain is not too large or, for that matter, too small. 
By estimating the local contribution to the energy from each subdomain one can 
obtain bounds for the total energy of the gas which are of the correct order. 

The second key idea that we will use is based on the observation that a pair of fermions, 
due to their relative antisymmetry,
experience an effective repulsion. This may be concretized in the following 
many-particle Hardy inequality for fermions~\cite[Theorem~2.8]{HofLapTid-08}:
\begin{equation} \label{eq:fermionic-Hardy}
	\sum_{j=1}^N \int_{\R^{dN}} |\nabla_j \Psi|^2 \,d\sx 
	\ \ge \ \frac{d^2}{N} \sum_{1 \le j<k \le N} \int_{\R^{dN}} 
		\frac{|\Psi(\sx)|^2}{|\bx_j - \bx_k|^2} \,d\sx,
\end{equation}
valid for any $N$-body state $\Psi \in H^1_\asym(\R^{dN})$ in any dimension $d \ge 1$.
Antisymmetry is in fact crucial here, as the inequality is not valid for bosons 
(the corresponding optimal Hardy constant vanishes in two dimensions).
A local version of \eqref{eq:fermionic-Hardy}, given below,
was obtained in~\cite{LunSol-13a}
for ideal anyons, i.e.\ with $d=2$ and with the right-hand side remaining linear in $N$, 
thus providing a local exclusion principle for anyons.
It was shown that this inequality may be combined with the Dyson--Lenard approach
to yield global bounds for the energy of the gas 
depending on the statistics parameter.

We start with an observation which is only helpful in the sufficiently extended case.
Namely,
for ideal anyons the singular magnetic potential $\bA$ effectively excludes the
diagonals $\bDelta$ from the configuration space, 
much like a strong repulsive point interaction. 
For $R$-extended anyons we have instead the following effective 
repulsive short-range interaction of soft-disk type:

\begin{lemma}[\textbf{Short-range magnetic interaction}]\label{lem:short-range}\mbox{}\\
	For any $\alpha \in \R$, $R > 0$, $N \ge 1$, 
	and 
	$\Psi \in \domD{\alpha, R} = H^1_{\sym}(\R^{2N})$
	we have that
	\begin{equation} \label{eq:short-range}
		\sum_{j=1}^N \int_{\R^{2N}} |D_j\Psi|^2\, d\sx
		\ \ge \ 2\pi|\alpha| \sum_{j \neq k} \int_{\R^{2N}} 
			\frac{\1_{B_R(0)}}{\pi R^2}(\bx_j - \bx_k)
			\, |\Psi|^2 \, d\sx.
	\end{equation}
\end{lemma}
Note that this repulsion is not at all as powerful as \eqref{eq:fermionic-Hardy}
upon taking the limit $R \to 0$ (or equivalently $\bar\varrho \to 0$),
because functions in $H^1(\R^{2N})$ may be approximated by 
smooth functions supported away from diagonals as $R \to 0$~\cite[Lemma~3]{LunSol-14}, 
such that the right-hand side of \eqref{eq:short-range} vanishes identically.
However the inequality will be useful in the case that
$R \sim \bar\varrho^{-1/2}$, 
i.e.\ $\bar\gamma \sim 1$.

Now, defining (denoted $C_{\alpha, N}$ in~\cite{LunSol-13a})
\begin{equation} \label{eq:C_alpha_N}
	\alpha_{N} 
	:= \min\limits_{p \in \{0, 1, \ldots, N-2\}} \min\limits_{q \in \Z} |(2p+1)\alpha - 2q|, 
	\qquad
	\alpha_* := \inf_{N \ge 2} \alpha_{N} = \lim_{N \to \infty} \alpha_{N}, 
\end{equation}
we may state  
the following local many-particle magnetic Hardy inequality for 
ideal anyons which was given in~\cite[Theorem~4]{LunSol-13a}:
\begin{theorem}\label{thm:ideal-Hardy}
	Let $\alpha \in \R$, $R=0$, $N \ge 1$ and $\Omega \subseteq \R^2$
	be open and convex. Then, for any $\Psi \in \domD{\alpha, 0}$, 
	$$
		\sum_{j=1}^N \int_{\Omega^N} |D_j\Psi|^2\, d\sx
		\ \ge \ \frac{\alpha_{N}^2}{N} \sum_{j < k} \int_{\Omega^N} 
			\frac{|\Psi|^2}{r_{jk}^2}
			\1_{\Omega \circ \Omega}(\bx_j, \bx_k) \, d\sx, 
	$$
	with 
	the reduced support
	$\1_{\Omega \circ \Omega}(\bx_j, \bx_k) := \1_{B_{\delta(\bX_{jk})}(0)}(\br_{jk})$, 
	and
	$$
		\br_{jk} := (\bx_j-\bx_k)/2, \qquad
		\bX_{jk} := (\bx_j+\bx_k)/2, \qquad
		r_{jk} := |\br_{jk}|, \qquad
		\delta(\bx) := \dist(\bx, \partial\Omega)
	$$
	pairwise coordinates and distances.
\end{theorem}
For fermions, with $\alpha=1$ and $\alpha_N = \alpha_* = 1$,
considered on the full two-dimensional plane $\Omega = \R^2$,
this is exactly the inequality \eqref{eq:fermionic-Hardy}.
For anyons the dependence on the statistics parameter $\alpha$ comes in
via the expressions \eqref{eq:C_alpha_N} as will be explained below.

Our first main result is the following improvement and extension of 
Theorem~\ref{thm:ideal-Hardy}  
to $R$-extended anyons, thereby providing us with a concrete 
(and indeed useful) measure of the
long-range effect of the statistical magnetic interaction:

\begin{theorem}[\textbf{Long-range magnetic interaction}]\label{thm:long-range}\mbox{}\\
	Let $\alpha \in \R$, $R \ge 0$, $N \ge n \ge 1$ and 
	$\Omega \subseteq \R^2$ be open and convex. Then, for any 
	$\Psi \in \domD{\alpha, R}$ and $\kappa \in [0, 1)$, 
	\begin{align}
		\sum_{j=1}^n \int_{\Omega^n} |D_j\Psi|^2 \, d\sx
		&\ge \frac{1}{n} \int_{\Omega^n} \biggl| \sum_{j=1}^n D_j \Psi \biggr|^2 \, d\sx \\
		&\qquad + \frac{1}{n} \sum_{j < k} \int_{\Omega^n} \biggl(
			(1-\kappa)\bigl| \partial_{r_{jk}} |\Psi| \bigr|^2 + 
			c(\kappa)^2 \frac{\alpha_N^2 }{r_{jk}^2}
			\1_A(\bx_j, \bx_k) 
			\, |\Psi|^2 \biggr) d\sx \\
		&\ge 4\pi(1-\kappa) \frac{1}{n} \sum_{j < k} \int_{\Omega^n} 
			g\biggl( \frac{c(\kappa)\alpha_N}{\sqrt{1-\kappa}}, \frac{3R/\delta(\bX_{jk})}{1-3R/\delta(\bX_{jk})} \biggr)^2
			\, \frac{ \1_A(\bx_j, \bx_k) }{ 4\pi\delta(\bX_{jk})^2 }
			\, |\Psi|^2 \, d\sx, 
	\end{align}
	where $D_j$ may depend on the positions of all $N$ particles 
	$\sx \in \R^{2N}$, 
	the support
	$$
		\1_A(\bx_j, \bx_k) := \1_{B_{\delta(\bX_{jk})-3R}(0) \setminus B_{3R}(0)}(\br_{jk})
	$$
	describes a maximal annulus contained in $\Omega$
	(with some $R$-dependent margins)
	in terms of the relative coordinate, and
	$g(\nu, \gamma)$ for $\nu\in \R_\limplus$ and $0 \le \gamma < 1$  
	is the square root of the smallest positive solution $\lambda$ 
	associated with the Bessel equation
	$-u''-u'/r+\nu^2 u/r^2 = \lambda u$
	on the interval $[\gamma, 1]$ with Neumann boundary conditions, 
	while  
	$g(\nu, \gamma):=\nu$ for $\gamma \ge 1$. 

	In the ideal 
	case $R=0$ the
	inequality is valid with
	$c(\kappa) \equiv 1$ (hence take $\kappa=0$), 
	while for any $R\ge 0$ it holds at least for 
	$c(\kappa) = 4.7\cdot 10^{-4}\kappa/(1+2\kappa)$.
	
	Moreover, the function $g$ 
	has the following properties:
	$$
		\nu \le g(\nu, \gamma) \le j_{\nu}', 
		\qquad
		g(\nu, \gamma)  
		\sim \left\{ \begin{array}{ll}
			j_{\nu}' \ge \sqrt{2\nu}, 	& \gamma \to 0, \\
			\nu, 						& \gamma \to 1, 
		\end{array}\right.
	$$
	where $j_\nu'$ denotes the first positive zero of the derivative 
	of the Bessel function $J_\nu$ (and $j_0':=0$).
\end{theorem}

Theorem~\ref{thm:long-range} will be applied to study the energy of the homogeneous 
anyon gas according to the local strategy outlined above. 
In such a setting $\Omega$ is typically not the domain to which our gas is 
confined, but rather a subdomain thereof, and $n$ is the number of particles 
present in $\Omega$ while $N$ is the total number of particles in the gas.
This more complicated division of particles is needed in the statement of the theorem 
because the magnetic derivatives depend on \emph{all} particles, 
not just those in $\Omega$,
which is even more relevant in the extended case.

We note that the above inequality may in some sense 
be viewed as a refinement (with respect to the 
angular dependence in pairwise relative coordinates)
of the usual (pointwise) diamagnetic inequality:

\begin{lemma}[\textbf{Diamagnetic inequality}]\label{lem:diamagnetic}\mbox{}\\
	For any $\alpha \in \R$, $R \ge 0$, $N \ge 1$ 
	and $\Psi \in \domD{\alpha, R}$
	we have that
	$$
		\sum_{j=1}^N \int_{\R^{2N}} |D_j\Psi|^2 \, d\sx
		\ \ge \ \sum_{j=1}^N \int_{\R^{2N}} \bigl| \nabla_j|\Psi| \bigr|^2 \, d\sx.
	$$
\end{lemma}

For $R>0$, the vector potential satisfies 
$\bA \in L^\infty(\R^{2N}) \subseteq L^2_{\loc}(\R^{2N})$ 
and hence it is covered by standard theorems; see e.g.~\cite[Theorem~7.21]{LieLos-01}.
For $R=0$ it is not, but the above diamagnetic inequality still holds
in this case, as was proved in~\cite[Lemma~4]{LunSol-14}
(and actually our understanding of the form domain $\domD{\alpha, 0}$ 
alluded to above depends on this general formulation of the inequality).

Note that $|\Psi| \in L^2_\sym(\R^{2N})$. Therefore
the diamagnetic inequality of Lemma~\ref{lem:diamagnetic} says that the kinetic
energy for anyons is always higher than that for bosons, 
while the short-range inequality of Lemma~\ref{lem:short-range} tells us that
the anyons also feel an effective repulsion proportional to $|\alpha|$ 
whenever they overlap.
Taking a combination of these two bounds would then correspond to a 
two-dimensional soft-disk repulsive Bose gas, whose energy in the dilute limit 
tends to zero logarithmically with the density 
(here the magnetic filling ratio 
$\bar\gamma := R\bar\varrho^{1/2} \to 0$)~\cite{LieYng-01}.
On the other hand, Theorem~\ref{thm:long-range} provides a local bound for
the energy in the form of a long-range inverse-square repulsion
similar to \eqref{eq:fermionic-Hardy}, and whose
strength depends on the fractionality of $\alpha$ via $\alpha_N \to \alpha_*$.
While this `statistical repulsion'
does not change the above repulsive picture much in the regime of high densities
($\bar\gamma \gtrsim 1$) where the anyons already feel each other's 
magnetic fields by (partially) overlapping, 
it makes a significant difference in the dilute limit, actually resulting in a uniform
bound for the energy from below in terms of $(j_{\alpha_*}')^2 \ge 2\alpha_*$.

As discussed in~\cite{LunSol-13a}, and further in~\cite{Lundholm-16},
the reason for the dependence on $\alpha_N \ge \alpha_*$ 
and not directly $\alpha$ in  
the bounds of Theorems~\ref{thm:ideal-Hardy} and~\ref{thm:long-range}
is the local gauge invariance 
of the pairwise relative magnetic potential.
Namely, in an exchange of a pair of particles additional flux may also be enclosed.
Apart from the flux corresponding to the simple exchange~\eqref{eq:particleExchangePhase},
enclosing $p$ other particles in such an exchange loop
contributes an additional $2p$ multiples of the exchange flux, 
yielding the factor $2p+1$ in~\eqref{eq:C_alpha_N}.
At the same time, any even multiple of a unit flux may be com\-pen\-sated for 
(gauged away) by an 
opposite and equally large orbital angular momentum of that same particle pair,
thus explaining the subtraction of 
an arbitrary even integer $2q$ in~\eqref{eq:C_alpha_N}.
However, for odd-numerator rational $\alpha$ there can never be a complete 
cancellation of this type, and therefore there is always some long-range 
pair repulsion, $\alpha_*>0$~\cite[Proposition~5]{LunSol-13a}.

All these effects are summarized in the following theorem concerning 
the $R$-extended anyon gas, which is our second main result (see Figure~\ref{fig:EnergyVsGammaPlots} for an illustration):

\begin{theorem}[\textbf{Universal bounds for the homogeneous anyon gas}]\label{thm:intro-gas}\mbox{}\\
	Let $e(\alpha, \bar\gamma)$, where $\bar\gamma=R\bar\varrho^{1/2}$, denote the ground-state energy 
	per particle and unit density of the extended anyon gas
	in the thermodynamic limit at fixed 
	$\alpha \in \R$, $R \ge 0$ and density $\bar\varrho > 0$ 
	where Dirichlet boundary conditions have been imposed, that is
	$$
		e(\alpha, \bar\gamma) 
		:= \liminf_{\substack{N, \, |Q_0| \to \infty \\ \! N/|Q_0| = \bar\varrho}} 
			\Biggl(
			\frac{1}{\bar\varrho N}
			\inf_{\substack{\Psi \in \domD{\alpha, R} \cap C^\infty_c(Q_0^N) \\ \|\Psi\|_2 = 1}} 
			\langle\Psi, \hat{T}_\alpha \Psi\rangle
			\Biggr).
	$$
	Then
	\begin{multline} \label{eq:intro-homo-universal}
		e(\alpha, \bar\gamma) \ge C\biggl(
			2\pi\frac{|\alpha| 
				\min\bigl\{2(1-\bar\gamma^2/4)^{-1}, K_\alpha\bigr\} }{ 
				K_\alpha + 2|\alpha|\ln(2/\bar\gamma) } 
				\1_{\bar\gamma < 2}
			+ 2\pi|\alpha| \1_{\bar\gamma \ge 2} \\
			+ \pi g(c\alpha_*, 12\bar\gamma/{\sqrt{2}})^2 (1-12\bar\gamma/{\sqrt{2}})_+^3
			\biggr), 
	\end{multline}
	for some universal constants $C, c > 0$, $K_\alpha \ge 2$ (is defined in Lemma~\ref{lem:LocalExclusionSR}), 
	and $g$ as in Theorem~\ref{thm:long-range}.
	Furthermore, for any $\alpha \in \R$
	we have for the ideal anyon gas that
	\begin{equation} \label{eq:intro-homo-ideal}
		e(\alpha, 0)  
		\ge \frac{1}{2} 2\pi \alpha_* \bigl(1 - O(\alpha_*^{1/3})\bigr).
	\end{equation}
\end{theorem}

\begin{figure}[ht]
	\centering
	\begin{tikzpicture}
		\node [above right] at (0,0) {\includegraphics[scale=0.82, clip, trim=10pt 0pt 0pt 0pt]{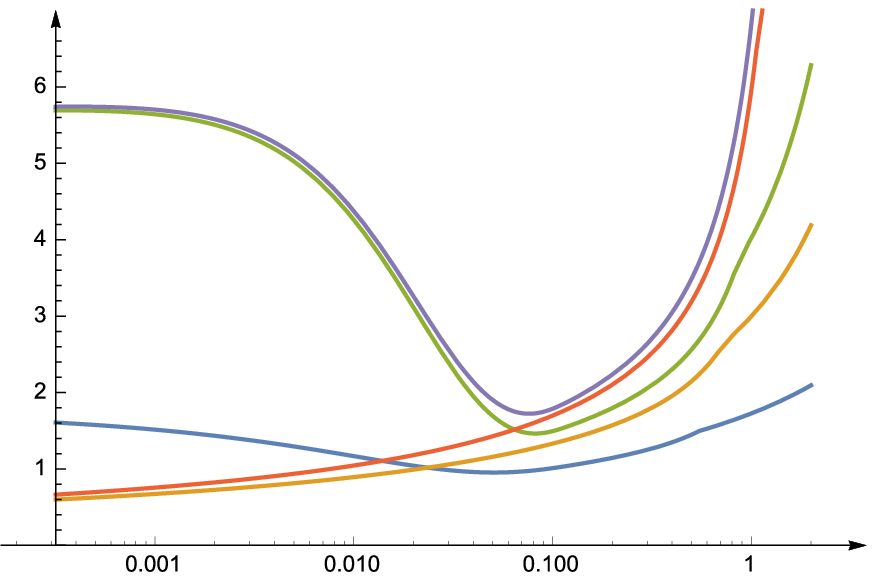}};
		\node [above right] at (7.0,.2) {\scalebox{0.8}{$\bar\gamma$\hspace{-5pt}}};
		\node [above right] at (0.3,0.8) {\scalebox{0.8}{$\alpha_*=0$\hspace{-5pt}}};
		\node [above right] at (0.3,1.35) {\scalebox{0.8}{$\alpha_*=1/3$\hspace{-5pt}}};
		\node [above right] at (0.3,4.04) {\scalebox{0.8}{$\alpha_*=1$\hspace{-5pt}}};
	\end{tikzpicture}
	\begin{tikzpicture}
		\node [above right] at (0,0) {\includegraphics[scale=0.82]{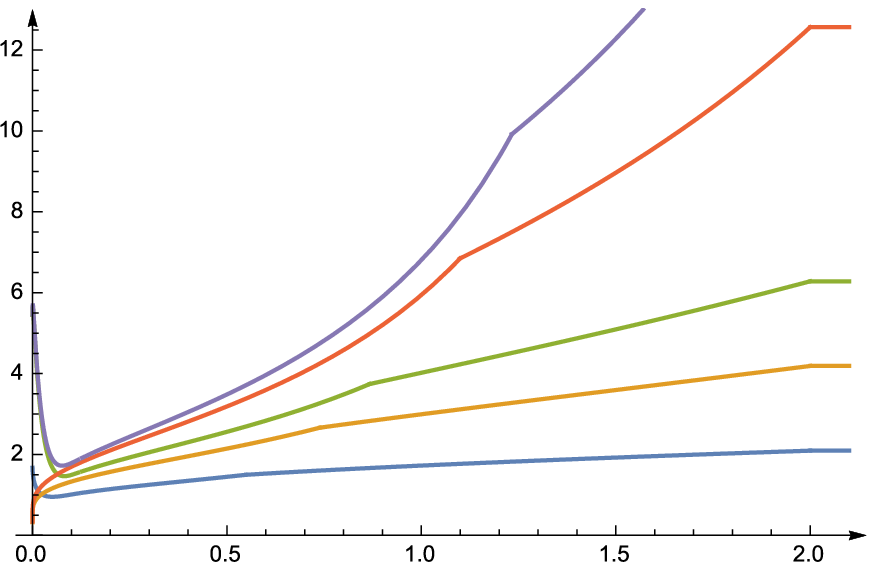}};
		\node [above right] at (7.3,.2) {\scalebox{0.8}{$\bar\gamma$\hspace{-5pt}}};
		\node [above right] at (6,1.04)   {\scalebox{0.8}{$\alpha=1/3$\hspace{-5pt}}};
		\node [above right] at (6,1.74) {\scalebox{0.8}{$\alpha=2/3$\hspace{-5pt}}};
		\node [above right] at (6.35,2.52) {\scalebox{0.8}{$\alpha=1$\hspace{-5pt}}};
		\node [above right] at (6.35,3.9) {\scalebox{0.8}{$\alpha=2$\hspace{-5pt}}};
		\node [above right] at (5.3,4.35) {\scalebox{0.8}{$\alpha=3$\hspace{-5pt}}};
	\end{tikzpicture}
	\caption{The universal lower bound for $e(\alpha, \bar\gamma)$ plotted as a 
	function of $\bar\gamma$ for 
	some fixed values of $\alpha$,
	in the hypothetical 
	case $C=1$, $c=1/{\sqrt{3}}$ for illustrative purposes.
	The figure to the right shows the general behavior over the full range of 
	$\bar\gamma$, while that on the left shows the behavior in the dilute 
	limit plotted in logarithmic scale where the long-range dependence on 
	$\alpha_*$ becomes relevant.}
	\label{fig:EnergyVsGammaPlots}
\end{figure}

As mentioned,
our approach to obtain the above theorem is to first formulate the effects
of the short- and long-range interactions in the form of local exclusion 
principles, an approach that goes back to Dyson and Lenard's original proof of 
the stability of matter for fermionic Coulomb systems~\cite{DysLen-67}.
This method was further developed 
in~\cite{LunSol-13a, LunSol-14, LunPorSol-15, LunNamPor-16},
not only to treat homogeneous gases but
also to prove Lieb--Thirring inequalities 
(i.e.\ uniform kinetic energy bounds in accord with the Thomas--Fermi approximation
for the inhomogeneous Fermi gas; cf.~\cite{LieThi-75, LieThi-76, LieSei-09})
with the usual Pauli exclusion principle for fermions 
replaced by more general repulsive interactions for bosons.
The reason for the factor $1/2$ in~\eqref{eq:intro-homo-ideal} compared to
the expected value $2\pi\alpha_*$ (at least if comparing to the Fermi gas at 
$\alpha=\alpha_*=1$ and assuming a linear interpolation
to small $\alpha$ such that $\alpha=\alpha_*$)
is that the long-range exclusion principle, 
which is applied locally on boxes of a tunable size, 
only increases linearly with the number of particles and
is strongest on a scale where about two particles fit in each box.
We provide further bounds for $e(\alpha, \bar\gamma)$ in various parameter regimes in Theorem~\ref{thm:gas}.

It should be remarked that our local exclusion principles also can be used to prove
Lieb--Thirring inequalities. 
We postpone the extended case to future work but note that the
ideal case is directly improved by the present results, 
namely replacing~\cite[Lemma~8]{LunSol-13a} with 
the local exclusion principle of
Lemma~\ref{lem:LocalExclusionLR} 
below yields the following bounds for ideal anyons, 
where the constant $(j_{\alpha_N}')^2 \ge 2\alpha_N \ge \alpha_N^2$ 
improves the one in~\cite[Theorems~1~and~11]{LunSol-13a}:

\begin{theorem}[\textbf{Lieb--Thirring inequality for ideal anyons}]\label{thm:intro-idealLT}\mbox{}\\
	With $\alpha \in \R$, $R=0$, $N \ge 1$ and $\Psi \in \domD{\alpha, 0}$
	we have that
	\begin{equation} \label{eq:intro-LT-kinetic}
		\langle \Psi, \hat T_\alpha \Psi\rangle
		\ge C (j_{\alpha_N}')^2 \int_{\R^2} \varrho_\Psi(\bx)^2 \, d\bx, 
	\end{equation}
	and if $V\colon \R^2 \to \R$ is an external one-body potential, 
	acting by $\hat V(\sx) := \sum_{j=1}^N V(\bx_j)$, then
	\begin{equation} \label{eq:intro-LT-potential}
		\langle \Psi, (\hat T_\alpha + \hat V)\Psi\rangle
		\ge -C' (j_{\alpha_N}')^{-2} \int_{\R^2} V_-(\bx)^2 \, d\bx, 
	\end{equation}
	for some positive universal constants $C$ and $C' = (4C)^{-1}$, 
	and $V_{\pm} := \max\{\pm V, 0\}$.
\end{theorem}

The question concerning optimality of the above bounds with respect to their dependence
on $\alpha$ in the dilute limit is a very difficult one, and will be 
discussed elsewhere~\cite{Lundholm-16}. 
However, we would like to point out that it was suggested 
in~\cite{LunSol-13b} (see also~\cite{Lundholm-13}) that a class of
FQHE-inspired trial states with a clustering behavior
could minimize the energy for certain fractions, 
and here we find additional support for this claim; 
cf.\ Figure~\ref{fig:rho/r_plot_clusters} below.
Furthermore, there was in~\cite{LunSol-13b}, 
then based on the weaker bounds of~\cite{LunSol-13a}, 
a slight inconsistency in the 
behavior with respect to odd-numerator $\alpha$ which is remedied by the improved
bounds presented here.

The structure of the paper is as follows. We lay the foundations in
Sections~\ref{sec:short-range} and \ref{sec:relative}
by proving the short-range
bound of Lemma~\ref{lem:short-range}, and the basis for the long-range
bound in the form of a relative magnetic Hardy inequality with symmetry.
Then the main body of the paper, Section~\ref{sec:long-range}, 
is concerned with the application of this Hardy
inequality to prove the long-range bound of 
Theorem~\ref{thm:long-range}.
This turns out to become surprisingly challenging 
in the extended case due to the oscillatory
nature of an effective potential,
and in fact takes up the largest part of the proofs section. 
In Section~\ref{sec:exclusion} the long- and short-range bounds are applied 
to prove local exclusion principles for anyons, 
and finally in Section~\ref{sec:gas} we discuss the homogeneous anyon
gas in the thermodynamic limit.

\medskip\noindent\textbf{Acknowledgments.} 
The authors would like to thank Michele Correggi, Ari Laptev, Fabian Portmann, 
Nicolas Rougerie, Robert Seiringer, Jan Philip Solovej and Aron Wennman for discussions. 
Thanks are also due to the editor and the anonymous referee whose comments have
improved the manuscript.
Financial support from the Swedish Research Council, 
grant nos.\ 2012-3864 (S.\ L.) and {2013-4734} (D.\ L.), 
is gratefully acknowledged.


\section{Short-range interaction}\label{sec:short-range}

The short-range interaction given by Lemma~\ref{lem:short-range} 
comes as a simple consequence of the well-known magnetic inequality
(see e.g.~\cite[Lemma~1.4.1]{FouHel-10} or~\cite[p.\ 171]{BalEvLew-15})
\begin{equation} \label{eq:magnetic inequality}
	\int_{\Omega} | (\nabla + i\bA )u |^2 \ge \pm \int_{\Omega} \curl \bA \, |u|^2, 
	\qquad u \in H^1_0(\Omega), \quad \Omega \subseteq \R^2.
\end{equation}
This inequality also follows directly from integrating the 
straightforward identity
$$
	|(\nabla + i\bA)u|^2 
	= |( (\partial_1+iA_1) \pm i(\partial_2+iA_2) )u|^2 
		\pm \curl \bJ[u] \pm \bA \cdot \nablap |u|^2, 
$$
with $\bJ[u] := \frac{i}{2}( u\nabla\bar{u} - \bar{u}\nabla u )$.

\begin{proof}[Proof of Lemma~\ref{lem:short-range}]
	Splitting the coordinates according to 
	$\sx = (\bx_j;\sx')$ for each particle $j$, 
	we write for the left-hand side of \eqref{eq:short-range}
	\begin{multline*}
		\sum_{j=1}^N \int_{\R^{2(N-1)}} \int_{\R^2} 
			|( \nabla_j + i\alpha\bA_j(\bx_j) )\Psi(\bx_j;\sx')|^2 \, d\bx_j d\sx' \\
		\ge \sum_{j=1}^N \int_{\R^{2(N-1)}} \int_{\R^2} 
			2\pi|\alpha| \sum_{k \neq j} \frac{\1_{B_R(\bx_k)}}{\pi R^2}(\bx_j)
			|\Psi(\bx_j;\sx')|^2 \, d\bx_j d\sx', 
	\end{multline*}
	where we used the expression~\eqref{eq:many-anyon mag field}
	for $\curl \alpha\bA_j(\bx_j)$
	in~\eqref{eq:magnetic inequality}.
	We have thus obtained the right-hand side of~\eqref{eq:short-range}.
\end{proof}

We note that the Dirichlet boundary conditions on $\Psi$ respectively $u$ are in fact 
necessary here since the bound \eqref{eq:magnetic inequality} 
is otherwise invalid, 
as can be seen by taking 
$\bA = \beta \bA_0$, $\beta \to 0$, 
and the trial state $u=1$.
Similarly, had we considered the inequality \eqref{eq:short-range}
locally on a small enough domain (compared to $R$) we would have 
found a contradiction as $\alpha \to 0$, 
unless Dirichlet boundary conditions are enforced.


\section{Relative magnetic inequality}\label{sec:relative}

For the long-range statistical interaction between anyons
we take the same starting point as in~\cite{LunSol-13a}, 
namely, the core observation is the validity of a relative magnetic
Hardy inequality which respects the symmetry of the anyon problem.
Non-symmetric versions of this inequality were introduced and studied
in~\cite{LapWei-98} (one-particle version) 
and in~\cite[Theorem~2.7]{HofLapTid-08} (many-particle version); 
see also~\cite{MelOuhRoz-04},~\cite[Chapter~5.5]{BalEvLew-15} and references therein.
However, as was pointed out in~\cite{LunSol-13a}, 
symmetry is crucial in order to obtain 
non-trivial bounds in the many-particle limit.
We formulate the following version of the inequality quite generally.

Initially, consider a magnetic field 
$b\colon B_{R}(0) \to \R$ defined on a disk of radius $R>0$, 
and assumed to be determined by a suitable continuous vector potential
$\ba\colon B_{R}(0) \to \R^2$ as $b = \curl \ba$.
Then the normalized flux inside a smaller disk of radius 
$r \in [0, R)$ is given by
\begin{equation} \label{eq:NormalizedFluxDef}
	\hat\Phi(r) := \frac{1}{2\pi} \int_{B_r(0)} b 
	= \frac{1}{2\pi} \int_{\partial B_r(0)} \ba \cdot d\br'.
\end{equation}
Note that if we were only given $\ba\colon \Omega \to \R^2$ 
on some annulus
$\Omega = B_R(0) \setminus \bar B_{R'}(0)$, with $0<R'<R$, 
i.e.\ if we only knew $b$ on $\Omega$
(so that only the right-hand side of \eqref{eq:NormalizedFluxDef} 
makes sense for $r \in (R', R)$), 
then $b$ can nevertheless be extended (non-uniquely) to the full interior
$B_{R'}(0)$, for example by taking
$$
	b|_{B_{R'}(0)} = \frac{2\pi\hat\Phi(R')}{\pi (R')^2}
	\qquad \text{or} \qquad b|_{B_{R'}(0)} = 2\pi\hat\Phi(R') \delta_0, 
$$
with $\Phi(R')$ here defined in terms of $\ba$ as in~\eqref{eq:NormalizedFluxDef} 
(note that we are not considering extending~$\ba$).
Then both expressions for $\hat\Phi(r)$ in \eqref{eq:NormalizedFluxDef} 
are well defined and agree for all $r \in (R', R)$.
We also note that if the magnetic field is antipodal-symmetric on $\Omega$, 
i.e.\ $b(-\br) = b(\br)$ for all $\br \in \Omega$, 
then the corresponding potential must (if gauge-normalized correctly)
be antipodal-\emph{anti}symmetric, 
$\ba(-\br) = -\ba(\br)$, $\br \in \Omega$, 
and vice versa.

\begin{lemma}[Magnetic Hardy inequality with symmetry] 
	\label{lem:magnetic-Hardy-symm}
	Let $\Omega = B_{R_2}(0) \setminus \bar{B}_{R_1}(0)$, with 
	$R_2 > R_1 \ge 0$, 
	be an annular domain in $\R^2$. 
	Let $\ba\colon \Omega \to \R^2$
	be a continuous vector potential corresponding to a magnetic field $b$, 
	$b|_\Omega = \curl \ba$, 
	that is defined	
	on the entire disk $B_{R_2}(0)$
	such that the normalized flux $\hat\Phi(r)$ given by \eqref{eq:NormalizedFluxDef}
	is finite for all $r \in (R_1, R_2)$.
	Furthermore, assume that $\ba$ is antipodal-\textbf{antisymmetric}
	resp.\ $b$ is antipodal-\textbf{symmetric} on $\Omega$, i.e.\
	$\ba(-\br) = -\ba(\br)$ resp.\ $b(-\br) = b(\br)$ for $\br \in \Omega$.
	
	Then, for any antipodal-\textbf{symmetric} $u\in C^\infty(\Omega)$, i.e.\ with $u(-\br) = u(\br)$ for all $\br \in \Omega$, 
	\begin{equation} \label{eq:magnetic-Hardy-boson}
		\int_{\Omega} 
		|(-i\nabla + \ba)u|^2 \, d\br
		\ \ge \ \int_{\Omega} \biggl(
			\bigl|\partial_r |u|\bigr|^2 +
			\inf_{k \in \Z} \bigl| \hat\Phi(r) - 2k \bigr|^2 
			\frac{|u|^2}{r^2} 
			\biggr) d\br.
	\end{equation}

	Alternatively, if instead $u$ is antipodal-\textbf{antisymmetric}, 
	$u(-\br) = -u(\br)$ for all $\br \in \Omega$, then
	\begin{equation} \label{magnetic-Hardy-fermion}
		\int_{\Omega} 
		|(-i\nabla + \ba)u|^2 \, d\br
		\ \ge \ \int_{\Omega} \biggl(
			\bigl|\partial_r |u|\bigr|^2 +
			\inf_{k \in \Z} \bigl| \hat\Phi(r) - (2k+1) \bigr|^2 
			\frac{|u|^2}{r^2} 
			\biggr) d\br.
	\end{equation}
\end{lemma}

\begin{proof}
	We apply the techniques from~\cite{LapWei-98}, 
	with symmetry taken into account as in~\cite[Lemma~2]{LunSol-13a}.
	We start by letting $h[\ba]$ denote the magnetic quadratic form on $\Omega$, 
	$$
		h[\ba](u) := \int_\Omega |(-i\nabla+\ba)u|^2\, d\br
		= \int_{R_1}^{R_2}\!\!\int_0^{2\pi} \bigl(|(-i\partial_r + a_r)u|^2 
			+ r^{-2}|(-i \partial_\vphi + r a_\vphi) u|^2\bigr) r\, d\vphi dr, 
	$$
	where $a_r := r^{-1} \br \cdot \ba$ and 
	$a_\vphi := r^{-1} \br^\perp \cdot \ba$.
	For the first term above we use the diamagnetic inequality 
	$|(\partial_r + ia_r)u| \ge \bigl|\partial_r |u|\bigr|$, 
	while for the second we can for each $r \in (R_1, R_2)$
	explicitly diagonalize the self-adjoint operator
	$K_\vphi(r) := -i\partial_\vphi + r a_\vphi(r, \vphi)$ acting on $L^2(\S^1)$. 
	The corresponding eigenvalues and normalized eigenfunctions of this 
	operator are given by:
	\begin{align}
		\lambda_k(r) &= -k + (2\pi)^{-1} r \int_0^{2\pi}a_\vphi(r, \vphi) \, d\vphi = -k + \hat\Phi(r), \\
		\psi_k(r, \vphi) &= (2\pi)^{-1/2}e^{i\left( \vphi \lambda_k(r)-r\int_0^\vphi a_\vphi(r, \eta)\, d\eta \right)}, 
	\end{align}
	for $k\in \Z$. 
	Because of the antipodal-antisymmetry of $\ba$, 
	implying antipodal-symmetry of $a_\vphi$, 
	i.e.\ $a_\vphi(r, \vphi)=a_\vphi(r, \vphi+\pi)$, 
	we have that
	\begin{align*}
		\psi_k(r, \vphi+\pi) 
		&= (-1)^k \psi_k(r, \vphi).
	\end{align*}
	Therefore, only the even/odd terms will contribute 
	upon expanding $u \in L^2_{\sym/\asym}(\Omega)$ as 
	\begin{equation}
		u(r, \vphi) 
		= \sum_{k\in \Z} u_k(r)\psi_k(r, \vphi)
		= \sum_{k\in \Z_{e/o}} u_k(r)\psi_k(r, \vphi), 
	\end{equation}
	with $\Z_e := 2\Z$ and $\Z_o := 2\Z+1$.
	
	By the above remarks and Parseval's identity we find that 
	\begin{align}
		h[\mathbf{a}](u) &= \int_{R_1}^{R_2}\!\! \int_0^{2\pi} |(\partial_r + ia_r)u|^2 \, r\, d\vphi dr 
			+ \int_{R_1}^{R_2} \sum_{k\in \Z_{e/o}} |\lambda_k(r)|^2 |u_k(r)|^2 \, r^{-1}\, dr \\
		&\ge \int_{R_1}^{R_2}\!\! \int_0^{2\pi} \bigl|\partial_r |u|\bigr|^2 \, r\, d\vphi dr 
			+ \int_{R_1}^{R_2} \inf_{k\in \Z_{e/o}} |\lambda_k(r)|^2 \sum_{k\in \Z_{e/o}} |u_k(r)|^2 \, r^{-1}\, dr \\
		&= \int_{R_1}^{R_2}\!\! \int_0^{2\pi} \bigl(\bigl|\partial_r |u|\bigr|^2 
			+ r^{-2}\inf_{k\in \Z_{e/o}} |\lambda_k(r)|^2 |u|^2\bigr) r\, d\vphi dr.
	\end{align}
	Thus the estimate we are left with is
	\begin{equation}
		h[\mathbf{a}](u)
		\ge \int_{R_1}^{R_2}\!\! \int_0^{2\pi} \bigl(\bigl|\partial_r |u|\bigr|^2 
			+ r^{-2}\inf_{k\in \Z_{e/o}}\bigl|\hat\Phi(r)-k\bigr|^2 |u|^2\bigr) r\, d\vphi dr, 
	\end{equation}	
	which proves the lemma.
\end{proof}

The above lemma not only extends the inequality of~\cite[Lemma~2]{LunSol-13a}
to more general (extended) magnetic fields, but also improves it by
keeping the radial derivative. This turns out to be crucial in order to
obtain an improved dependence on $\alpha$ in the dilute limit.
We note that in~\cite{HofLapTid-08} the radial derivatives
were effectively discarded in two dimensions.


\section{Analysis of the long-range interaction}\label{sec:long-range}

We set out to prove Theorem~\ref{thm:long-range}, 
and first note that by the remarks in Section~\ref{sec:intro-prelims}
we may assume without loss of generality that 
$\Psi \in C^\infty_c(\R^{2N} \setminus \bDelta)$.
Proceeding as was done in~\cite{LunSol-13a} for the non-extended case $R=0$, 
we start from the expression for the kinetic energy on a domain
$\Omega \subseteq \R^2$, 
$$
	\int_{\Omega^n} \sum_{j=1}^n |D_j \Psi|^2\, d\sx,  
	\qquad
	\text{where}
	\qquad
	D_j = -i \nabla_{\bx_j} 
		+ \alpha \sum_{\substack{k=1\\k \neq j}}^N \frac{(\bx_j-\bx_k)^\perp}{|\bx_j-\bx_k|_R^2}, 
$$
and we are considering the first $n$ particles $\bx_{j=1, \ldots, n} \in \Omega$
while the remaining $N-n$ ones may reside anywhere in $\R^2$.
Using that, for any $\sz = (\bz_j)_j \in \C^{n}$, 
\begin{equation} \label{eq:abs-rel-identity}
	\sum_{j=1}^n |\bz_j|^2 
		= \frac{1}{n} \sum_{1 \le j<k \le n} |\bz_j - \bz_k|^2
		+ \frac{1}{n} \biggl| \sum_{j=1}^n \bz_j \biggr|^2, 
\end{equation}
we have that
\begin{align}
	\int_{\Omega^n} \sum_{j=1}^n |D_j \Psi|^2\, d\sx 
	=& \frac{1}{n} \sum_{1\le j<k \le n} \int_{\Omega^{n-2}}\int_{\Omega^2} |(D_j-D_k)\Psi|^2 \, d\bx_j d\bx_k \prod_{l\neq j, k}d\bx_l \\
	&+ \frac{1}{n} \int_{\Omega^n} \biggl| \sum_{j=1}^n D_j\Psi \biggr|^2 d\sx, 
	\label{eq:DerivativeSplit}
\end{align}
where we also note that the magnetic field present in the 
last (total momentum) term simplifies to 
$$
	\sum_{j=1}^n D_j = -i \sum_{j=1}^n \nabla_{\bx_j} 
		+ \alpha \sum_{j=1}^n \sum_{k=n+1}^N \frac{(\bx_j-\bx_k)^\perp}{|\bx_j-\bx_k|_R^2}, 
$$
by the antisymmetry of the vector potential, and thus vanishes if $n=N$.

We now study the inner integral in \eqref{eq:DerivativeSplit} 
for the $j<k$ particle pair, and introduce relative coordinates:
$$
		\br_{jk} := (\bx_j-\bx_k)/2, \qquad
		\bX_{jk} := (\bx_j+\bx_k)/2, \qquad
		r_{jk} := |\br_{jk}|, 
$$
giving
\begin{multline}\label{eq:RelativeMagneticEnergy}
	\int_{\Omega^2} |(D_j-D_k)\Psi|^2 \, d\bx_jd\bx_k \\
	\begin{aligned}
	&= \int_{\Omega^2}\Bigl|\Bigl(-i(\nabla_{\bx_j}-\nabla_{\bx_k})
	+\alpha \sum_{l\neq j} \frac{(\bx_j-\bx_l)^\perp}{|\bx_j-\bx_l|_R^2}
	-\alpha\sum_{l\neq k} \frac{(\bx_k-\bx_l)^\perp}{|\bx_k-\bx_l|_R^2}\Bigr)\Psi\Bigr|^2 d\bx_jd\bx_k\\
	&= \int_{\Omega^2}\Bigl|\Bigl(-i\nabla_{\br_{jk}}
	+\alpha \ba_{0}(\br_{jk}) + \alpha \sum_{l\neq j, k}(\ba_l(\bX_{jk}, \br_{jk})-\ba_l(\bX_{jk}, -\br_{jk}))
		\Bigr)\Psi\Bigr|^2 d\bx_jd\bx_k, 
	\end{aligned}
\end{multline}
where the relative vector potentials are given by
\begin{align}\label{eq:EffectivMagVecPotential}
	\ba_0(\br):=\frac{4 \br^\perp}{|2\br|^2_R} = \frac{\br^\perp}{|\br|^2_{R/2}} \quad \textrm{and} \quad 
	\ba_l(\bX, \br):=\frac{(\bX+\br-\bx_l)^\perp}{|\bX+\br-\bx_l|_R^2}.
\end{align}
Hence, for any positions
$\sx' = (\bx_1, \ldots, \bx\!\!\!\!\diagup\!_j, \ldots, \bx\!\!\!\!\diagup\!_k, \ldots, \bx_N) \in \R^{2(N-2)}$
of the other particles and for each center-of-mass coordinate
$\bX = \bX_{jk} \in \Omega$ of the particle pair, 
we observe that the resulting magnetic vector potential 
$$
	\ba(\br) := \alpha \ba_{0}(\br) + \alpha \sum_{l\neq j, k}(\ba_l(\bX, \br)-\ba_l(\bX, -\br))
$$
is antipodal-antisymmetric on the relative disk 
$$
	\Omega_\bX := B_{\delta(\bX)}(0), \qquad
	\delta(\bx) := \dist(\bx, \partial\Omega), 
$$
with a corresponding antipodal-symmetric magnetic field
\begin{equation}\label{eq:Rel_MagField}
	b := \curl\ba = 2\pi\alpha \biggl(
		\frac{\1_{B_{R/2}(0)}}{\pi(R/2)^2} + \sum_{l\neq j, k} 
		\biggl( \frac{\1_{B_R(\bx_l-\bX)}}{\pi R^2} + \frac{\1_{B_R(-(\bx_l-\bX))}}{\pi R^2} \biggr)
		\biggr)
\end{equation}
(given here for $R>0$).
Also, the smooth function defined relative to $\bX$ and $\sx'$ by
$$
	u(\br) := \Psi(\bx_1, \bx_2, \ldots, \bx_j=\bX+\br, \ldots, \bx_k=\bX-\br, \ldots, \bx_n, \ldots, \bx_N)
$$
is antipodal-symmetric on $\Omega_\bX$.
Hence, we may apply the relative magnetic Hardy inequality of
Lemma~\ref{lem:magnetic-Hardy-symm} 
(for $R=0$ we split into concentric annuli avoiding the $\bx_l$ as in~\cite[Theorem~4]{LunSol-13a}) 
to obtain that
\begin{align}
	\int_{\Omega^2} |(D_j-D_k)\Psi|^2 \, d\bx_jd\bx_k
	&\ge \int_{\Omega} \int_{\Omega_\bX} |(-i\nabla + \ba)u|^2 \, 4\, d\br d\bX\\
	&\ge \int_{\Omega} \int_{\Omega_\bX} \biggl( 
		\bigl|\partial_r |u|\bigr|^2 + \frac{\rho(r)}{r^2}|u|^2
		\biggr) 4\, d\br d\bX, 
	\label{eq:RelativeScalarEnergy}
\end{align}
where
\begin{equation}\label{eq:def_rho}
	\rho(r) := \inf_{q \in \Z} \bigl| \hat\Phi(r) - 2q \bigr|^2, 
\end{equation}
and $\hat\Phi(r)$ here, and in what follows, denotes the flux through the disk $B_r(\bX)$ of the magnetic field~\eqref{eq:Rel_MagField}:
\begin{equation}
	\hat\Phi(r) = \frac{1}{2\pi} \int_{\partial B_r(0)} \ba \cdot d\br'
	= \frac{1}{2\pi} \int_{B_r(0)} b.
\end{equation} 
Note that the magnetic field is induced by the particle configuration 
$(\sx'; \bx_j, \bx_k)$, and the only dependence which remains after fixing 
$\sx'$ in~\eqref{eq:DerivativeSplit} and $\bX = \bX_{jk}$ 
in~\eqref{eq:RelativeScalarEnergy} 
is that of the relative coordinate $\br = \br_{jk}$. 
With the remaining particle positions expressed relative to the coordinate $\bX$, 
$\by_l:= \bx_l-\bX$, 
we can write the normalized flux $\hat \Phi(r)$ as:
\begin{equation}\label{eq:def_anyonPhi}
	\hat\Phi(r) 
	= \alpha \biggl( \int_{B_r(0)} \frac{\1_{B_{R/2}(0)}}{\pi(R/2)^2}
		+ 2\sum_{l\neq j, k} \int_{B_r(0)} \frac{\1_{B_R(\by_l)}}{\pi R^2} \biggr). 
\end{equation}
Hence $\rho(r)$ depends only on the arbitrary but fixed configuration $(\by_l)_l\in \R^{2(N-2)}$.
 
By the above discussion, the problem of bounding the kinetic energy~\eqref{eq:DerivativeSplit} 
has been reduced to studying the radial Schr\"odinger operator 
in~\eqref{eq:RelativeScalarEnergy} with explicit
\emph{scalar} interaction potential $\rho(r)/r^2$. 
This potential is essentially an inverse-square repulsion, modulated 
with a coupling strength $\rho(r)$ which measures how well the 
normalized flux $\hat\Phi(r)$ stabilizes 
away from the even integers.
In the dilute situation the flux and hence also $\rho$ would for the most part be constant, 
however we could have significant oscillations of $\rho(r)$ between one and zero 
whenever many particles
are enclosed over short differences in the radial variable $r$
(see Figure~\ref{fig:rho/r_plot}).
Controlling these oscillations turns out to be a significant challenge, 
and the entire remainder of this section shall be concerned with proving 
the following theorem, 
from which Theorem~\ref{thm:long-range} follows.

\begin{theorem}\label{thm:MainRadialBound}
	For any $0 \le R \le L/6$, $\kappa \in [0, 1]$, 
	$u \in W^{1, 2}([R, L], rdr)$,
	and $\rho$ defined in~\eqref{eq:def_rho}-\eqref{eq:def_anyonPhi}  
	with $(\by_l)_l \in \R^{2(N-2)}$ arbitrary, 
	we have that
	$$
		\int_{R}^L \biggl( |u'|^2 + \frac{\rho(r)}{r^2}|u|^2 \biggr) r\, dr
		\ge \int_{R}^{L} \biggl( (1-\kappa)|u'|^2 + c(\kappa)^2\frac{\alpha_N^2}{r^2}\1_{[3R, L-3R]}|u|^2 \biggr) r\, dr, 
	$$
	with $c(\kappa) = 4.7\cdot 10^{-4}\kappa/(1+2\kappa)$.
	In the case $R=0$ we may take $c(\kappa) \equiv 1$.
\end{theorem}

\begin{remark} \label{rem:MainRadialBound}
	The margins which appear here as a cut-off for the potential
	are not optimal and could be improved with more work, 
	to the cost of an even weaker constant. 
	The main reason for the weakness of the constant $c(\kappa)$ is the fact that we
	have chosen to control the above form by means of filling the gaps around the zeros of the 
	potential by smearing it
	over longer (but not too long) intervals, 
	and that in the worst possible situation there are very large regions of 
	intense oscillation and many such zeros.

	By considering the special case $\alpha=\alpha_N=1$ and
	densely packed, overlapping particles (i.e.\ when $\bar \gamma$ is large) 
	distributed so that the effective magnetic field is 
	approximately constant, we find that $c(\kappa)$ cannot 
	be greater than $1/{\sqrt{3}}$, which is what the 
	corresponding constant would be if one applied the same 
	argument to the case of a homogeneous magnetic field (see below).
	However, for $\alpha_* \le 1/2$ (or small enough so that $\rho$ is larger 
	than $\alpha_*^2$ for a sufficiently large set of radii), 
	we expect that the ground-state energy of the 
	left-hand side (though difficult to compute in general)
	should in almost all situations be bounded by
	that with $\rho(r)$ replaced by $\alpha_N^2$ (compare Figure~\ref{fig:rho/r_plot}). 
	We discuss further possible improvements to the constant 
	$c(\kappa)$ at the end of Section~\ref{sec:ProofLongRange}.
\end{remark}

\begin{figure}[ht]
	\centering
	\begin{tikzpicture}
		\node [above left] at (-0.3,0) {\includegraphics[scale=0.7, trim=0.2cm -0.4cm 0.5cm 0cm]{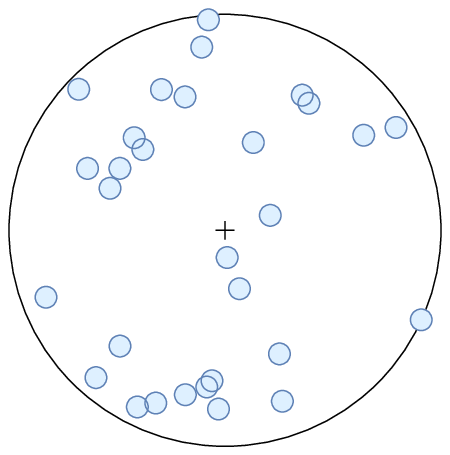}};
		\node [above right] at (0,0) {\includegraphics[scale=0.63]{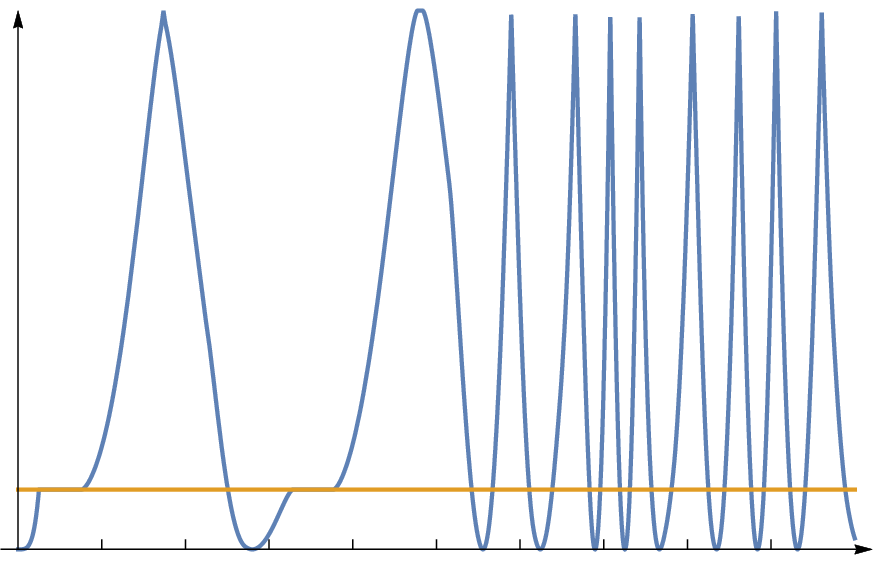}};
		\node [above right] at (5.65,0) {\scalebox{0.8}{$r$\hspace{-5pt}}};
		\node [above right] at (0.05,3.6) {\scalebox{0.8}{$\rho$}};
		\node [above right] at (-0.3,0.3) {\scalebox{0.8}{$\alpha_*^2$}};
		\node [above right] at (-0.2,3.3) {\scalebox{0.8}{$1$}};
		\node [above right] at (6+0,0) {\includegraphics[scale=0.63]{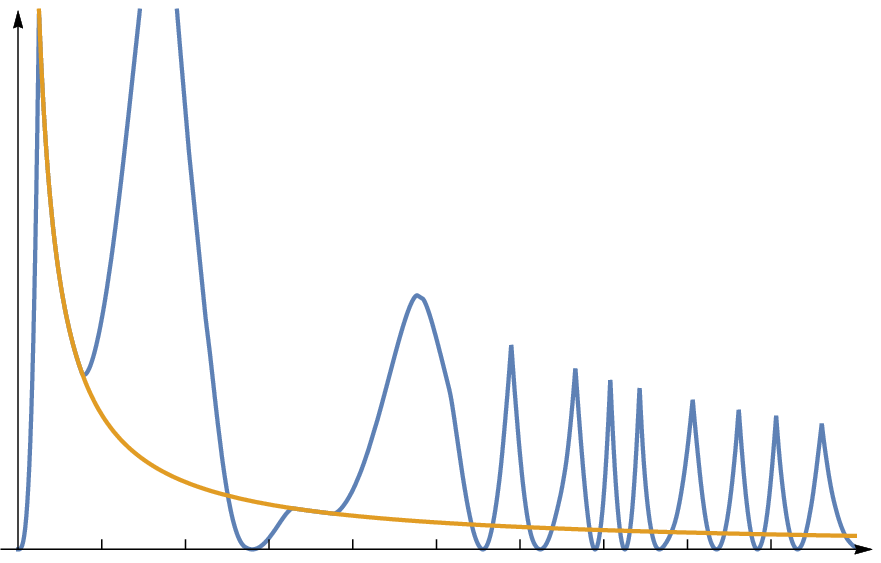}};
		\node [above right] at (6+5.65,0) {\scalebox{0.8}{$r$\hspace{-5pt}}};
		\node [above right] at (6+0,3.6) {\scalebox{0.8}{$\rho/r$}};
	\end{tikzpicture}
	\caption{The function $\rho(r)$ and the effective potential 
		$\rho(r)/r$ for a random (uniformly distributed)
		configuration of $30$ particles in a disk of radius $L=20 R$ with 
		$\alpha=\alpha_*=1/3$ plotted from $r=0$ to $r=L$, 
		where $\alpha_*^2$ resp.\ $\alpha_*^2/r$ are shown for comparison. 
		As one can see, the effective potential is generally quite a lot larger 
		than $\alpha_*^2/r$.}
	\label{fig:rho/r_plot}
\vspace{15pt}
	\begin{tikzpicture}
		\node [above left] at (-0.3,0) {\includegraphics[scale=0.7, trim=0.2cm -0.4cm 0.5cm 0cm]{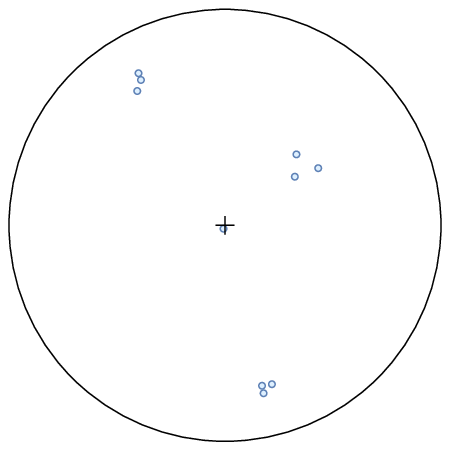}};
		\node [above right] at (0,0) {\includegraphics[scale=0.63]{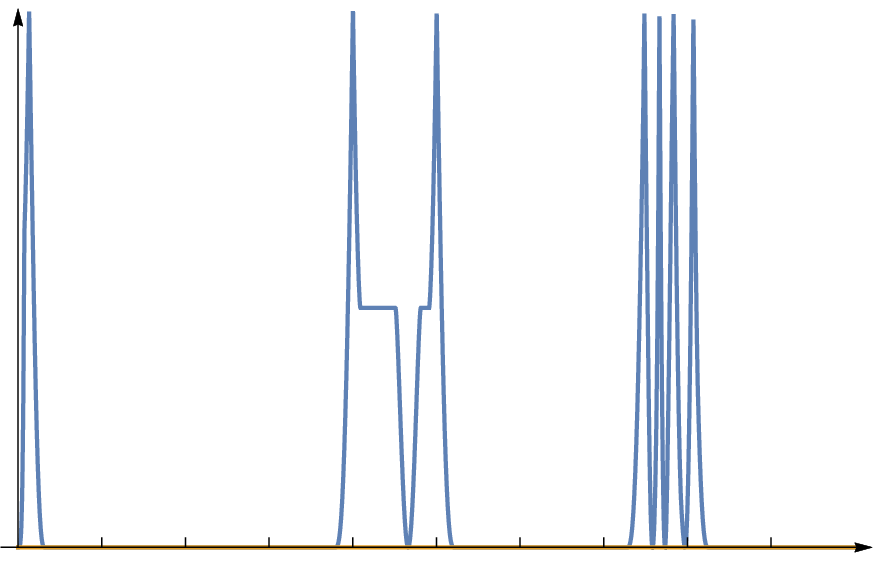}};
		\node [above right] at (5.65,0) {\scalebox{0.8}{$r$\hspace{-5pt}}};
		\node [above right] at (0.05,3.6) {\scalebox{0.8}{$\rho$}};
		\node [above right] at (-0.3,-0.05) {\scalebox{0.8}{$\alpha_*^2$}};
		\node [above right] at (-0.2,3.3) {\scalebox{0.8}{$1$}};
		\node [above right] at (6+0,0) {\includegraphics[scale=0.63]{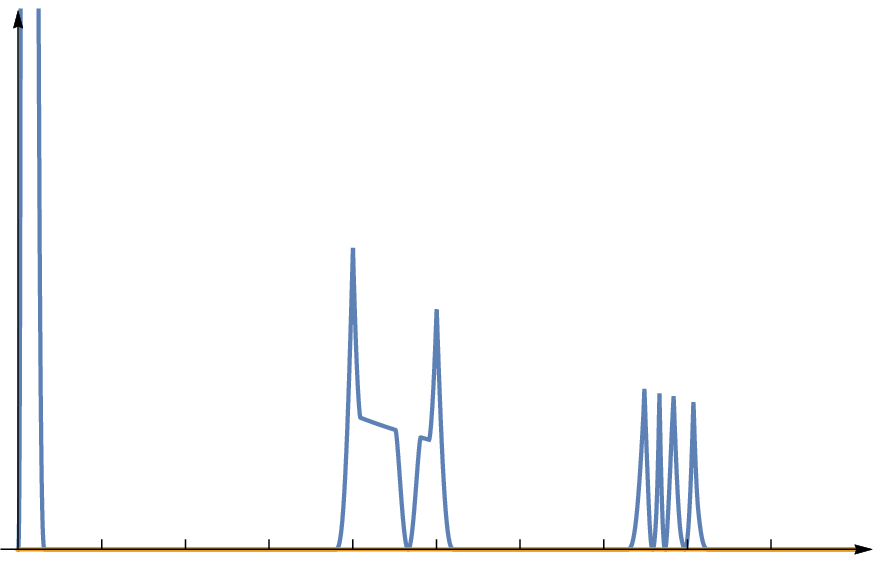}};
		\node [above right] at (6+5.65,0) {\scalebox{0.8}{$r$\hspace{-5pt}}};
		\node [above right] at (6+0,3.6) {\scalebox{0.8}{$\rho/r$}};
	\end{tikzpicture}
	\caption{The same as in Figure~\ref{fig:rho/r_plot}, now
    for $10$ particles in a disk of radius $L=60 R$ with 
    $\alpha=2/3$ ($\alpha_*=0$), and with a single particle close to 
    our center of mass and the remaining nine in clusters of three. 
    Note that in this case the effective potential can become identically 
    zero on long intervals.}
	\label{fig:rho/r_plot_clusters}
\end{figure}

\begin{proof}[Proof of Theorem~\ref{thm:long-range}]
	Inserting the bound of Theorem~\ref{thm:MainRadialBound} 
	with $L = \delta(\bX)$ into the expressions
	\eqref{eq:RelativeScalarEnergy}, \eqref{eq:RelativeMagneticEnergy}, 
	and \eqref{eq:DerivativeSplit}, 
	we obtain the first bound of the theorem.
	Furthermore, by rescaling $v(r) := u((L-3R)r)$ and
	considering the minimizer $v$ which is the solution of the Bessel equation
	$$
		-v''(r) -v'(r)/r + \nu^2 v(r)/r^2 = \lambda v(r), 
		\qquad  v'(\gamma)=0, \ v'(1) = 0, 
	$$
	with the minimal eigenvalue
	$\lambda = g\Bigl( \nu = \frac{c(\kappa)\alpha_N}{\sqrt{1-\kappa}}, 
		\gamma = \frac{3R/L}{1-3R/L} \Bigr)^2 \ge 0$, 
	one obtains
	\begin{multline*}
		\int_{3R}^{L-3R} \biggl( (1-\kappa)|u'|^2 + c(\kappa)^2\frac{\alpha_N^2}{r^2}|u|^2 \biggr) r\, dr
		= (1-\kappa) \int_{\gamma}^{1} \biggl( |v'|^2 
			+ \frac{c(\kappa)^2}{1-\kappa} \frac{\alpha_N^2}{r^2}|v|^2 \biggr) r\, dr \\
		\ge (1-\kappa) 
			g\Bigl(\frac{c(\kappa)\alpha_N}{\sqrt{1-\kappa}}, \gamma\Bigr)^2 
			\int_{\gamma}^{1} |v|^2 r\, dr
		= (1-\kappa)\frac{ 
			g\Bigl(\frac{c(\kappa)\alpha_N}{\sqrt{1-\kappa}}, \gamma\Bigr)^2 
			}{L^2(1-3R/L)^2} \int_{3R}^{L-3R} |u|^2 r\, dr, 
	\end{multline*}
	and therefore, 
	after the simplifying 
	estimate $(1-3R/L)^{-2} \ge 1$, 
	the second bound of Theorem~\ref{thm:long-range}. 
	The properties of $g$ described in the theorem are direct consequences of Propositions~\ref{prop:BesselPrimeZeroBounds} and~\ref{prop:BesselEigenvalueBounds}.
\end{proof}

Before continuing with the proof of Theorem~\ref{thm:MainRadialBound}
we note that, 
although this method involving the magnetic Hardy inequality turns out
to be sufficient and indeed well-suited for our purposes, 
it does not deal well with strong magnetic fields
(hence also the presence of a large external field), 
as the following example shows.
The strong magnetic fields arising from a large overlap between
the particles will instead be handled by the short-range part of the 
interaction, Lemma~\ref{lem:short-range}.

\begin{proposition}[Constant magnetic field on a disk]\label{prop:constant-field}
	The ground-state energy $\lambda_1(\beta)$ for the 
	Neumann form (with no symmetry imposed) with 
	a constant magnetic field $b(\br)=\beta \ge 0$ on the unit disk, 
	$$
		\lambda_1(\beta) := \inf_{\|u\|_2 = 1} \int_{B_1(0)} 
			\bigl|(-i\nabla + \beta\br^\perp/2)u\bigr|^2 \, d\br, 
	$$
	satisfies
	$$
		\lambda_1(\beta) \sim \Theta_0 \beta \quad \text{as} \quad \beta \to \infty, 
		\quad \text{where} \quad \Theta_0 \approx 0.59.
	$$
	However, the ground-state energy for the corresponding lower bound
	obtained from the Hardy inequality, 
	$$
		\mu_1(\beta) := \inf_{\|u\|_2 = 1} \int_{B_1(0)} \biggl( 
			\bigl|\partial_r |u|\bigr|^2 
			+ \inf_{k \in \Z}|k - \beta r^2/2|^2 \frac{1}{r^2} |u|^2 
			\biggr)\, d\br, 
	$$
	is bounded from above by $g(1/2, 0)^2 = (j_{1/2}')^2$ independent of $\beta$.
\end{proposition}
\begin{proof}
	The first estimate follows e.g.\ from~\cite[Theorem~5.3.1]{FouHel-10}, 
	while the second from  bounding the infimum by $1/4$ and taking as a trial state 
	the Bessel function $u(\br) = J_{1/2}(j_{1/2}'r)$.
\end{proof}


\subsection{A one-dimensional projection bound}\label{sec:ProjArgInterval}

Our strategy in order to find a uniform bound for the scalar interaction 
of Theorem~\ref{thm:MainRadialBound}
will be to borrow a bit of the radial kinetic energy to smear 
$\rho$ over intervals whenever it has critical oscillations.
As a preliminary to the proceeding analysis we therefore
study the localized effective quadratic form 
\begin{equation}
	h_{I, \rho}(u) := \int_{I} \bigl(\kappa |u'|^2 + \frac{\rho}{r^2}|u|^2\bigr)r\, dr, \quad \kappa \in [0, 1],
\end{equation}
on an interval $I=(r_1, r_2) \subseteq \R_\limplus$, 
and our goal is to find a bound of the form
\begin{equation}\label{eq:LocalprojHardyBound}
	h_{I, \rho}(u) \gtrsim \int_I \frac{|u|^2}{r}\, dr, 
\end{equation}
i.e.\ corresponding to $\rho$ being constant.
\begin{lemma}\label{lem:1Dprojection}
	Let $I$ be an interval $(r_1, r_2)$, such that $r_1\ge R$ and 
	$|r_2-r_1|\le 2 R$, and let $\rho\in L^\infty(I)$ be non-negative with 
	$\|\rho\|_\infty\le1$. Then for any $\kappa \in [0, 1]$ we have that
		\begin{equation}
			\int_{I} \bigl(\kappa|u'|^2+\frac{\rho}{r^2}|u|^2\bigr)r\, dr\ge \frac{\kappa \bar \rho}{\beta(\kappa)} \int_I \frac{|u|^2}{r}\, dr, 
		\end{equation}
	where $\bar\rho$ denotes the weighted mean on $I$, 
	\begin{equation}
		\bar \rho := \int_I \frac{\rho}{r}\, dr \Bigm/ \int_I \frac{dr}{r}, 
	\end{equation}
	and $\beta(\kappa)$ is an explicit function satisfying 
	$\kappa < \beta(\kappa) < \kappa+1/4$.
\end{lemma}

\begin{remark}
	This lemma can be proven under more general conditions;
	the only condition on $I$ needed for our proof is that $r_2/r_1$ 
	is sufficiently small.
	The current setting is simply what we require later.
\end{remark}

\begin{proof}[Proof of Lemma~\ref{lem:1Dprojection}]
	By the change of variables $r=e^t$ and with $\tilde u(t)=u(e^t)$ we find that
	\begin{equation}
		h_{I, \rho}(u)= \tilde h(\tilde u):= \int_{\ln(I)}\bigl(\kappa|\tilde u'|^2+\tilde\rho|\tilde u|^2\bigr)\, dt.
	\end{equation}
	For this quadratic form we can perform a projection-type argument to bound 
	the first eigenvalue of the associated operator 
	$\widetilde{H}:= -\kappa \frac{d^2}{dt^2}+\tilde\rho$ 
	(with Neumann boundary conditions), 
	which in turn will imply a bound of the desired form.

	Let $P$ denote the orthogonal projection onto the ground state 
	$\psi_0\equiv 1/{\sqrt{|{\ln(I)}|}}$ of $-d^2/dt^2$, where 
	$|{\ln(I)}|=\ln(r_2/r_1)$, and let $P^\perp=\1-P$. 
	Then $\bigl(-\frac{d^2}{dt^2} \bigr)P=0$ and 
	$\bigl(-\frac{d^2}{dt^2}\bigr)P^\perp \ge \pi^2/|{\ln(I)}|^2 P^\perp$
	(the first non-zero Neumann eigenvalue of $-d^2/dt^2$). 

	Since $\tilde{\rho}\ge 0$, an application of Cauchy--Schwarz' and Young's inequalities 
	yields, for any $u\in L^2(\ln(I))$ and $\mu>0$, that
	\begin{align}
	\bigl| \langle u, (P\tilde{\rho}P^\perp+P^\perp \tilde{\rho} P)u\rangle \bigr|
	&= \bigl| \langle \tilde{\rho}^{1/2}Pu, \tilde{\rho}^{1/2}P^\perp u \rangle 
		+ \langle \tilde{\rho}^{1/2}P^\perp u, \tilde{\rho}^{1/2}P u \rangle \bigr|\\
	&\le \mu \|\tilde{\rho}^{1/2}Pu\|_2^2 + \mu^{-1}\|\tilde{\rho}^{1/2}P^\perp u\|_2^2 
	= \langle u, (\mu P \tilde{\rho} P + \mu^{-1}P^\perp \tilde{\rho} P^\perp) u \rangle.
	\end{align}
	Hence we see that
	\begin{equation}\label{eq:Proj_Split_Potential_1D}
	\tilde{\rho} = (P+P^\perp)\tilde{\rho}(P+P^\perp) \ge (1-\mu)P\tilde{\rho}P+(1-\mu^{-1})P^\perp \tilde{\rho}P^\perp.
	\end{equation}
	The 
	operator $P\tilde \rho P$ is equal to 
	$\|\tilde\rho\|_1/{|{\ln(I)}|} P$, 
	where
	\begin{align*}
		\|\tilde\rho\|_1=\int_{\ln(I)} \tilde\rho\, dt = \int_I \rho(r)r^{-1}\, dr, 
	\end{align*}
	and  
	$P^\perp \tilde \rho P^\perp$
	we bound from above by 
	$\|\tilde{\rho}\|_{\infty} P^\perp$. 

	We find that for any 
	$\mu \in (0, 1)$ the operator 
	$\widetilde{H}$ satisfies
	\begin{align}\label{eq:ProjectionBound}
		\widetilde{H} 
		&\ge
		\kappa\Bigl(-\frac{d^2}{dt^2}\Bigr)P+\kappa\Bigl(-\frac{d^2}{dt^2}\Bigr)P^\perp + (1-\mu)P\tilde\rho P +(1-\mu^{-1})P^\perp\tilde\rho P^\perp\\
		&\ge 
		\frac{(1-\mu)}{|{\ln(I)}|}\|\tilde\rho\|_1 P + \biggl(\frac{\kappa\pi^2}{|{\ln(I)}|^2} +(1-\mu^{-1})\|\tilde\rho\|_\infty\biggr)P^\perp\\
		&\ge
		\min\biggl\{\frac{(1-\mu)\|\tilde \rho\|_1}{|{\ln(I)}|}, \frac{\kappa\pi^2}{|{\ln(I)}|^2} +(1-\mu^{-1})\|\tilde\rho\|_\infty\biggr\}.
	\end{align}
	With $|r_2-r_1|=|I|\le 2R$ and $r_1\ge R$ we find that
	\begin{equation}
		|{\ln(I)}| = \ln \Bigl(\frac{r_2}{r_1}\Bigr) = \ln \Bigl(1+\frac{|I|}{r_1}\Bigr) \le \ln \Bigl(1+\frac{2R}{R}\Bigr)= \ln (3).
	\end{equation}
	Hence, writing $\mu=1-\kappa/\beta$, $\beta> \kappa$, and using that $\|\tilde\rho\|_\infty=\|\rho\|_\infty\le 1$, $\|\tilde\rho\|_1/|{\ln(I)}|\le 1$ we have that
	\begin{align}
		\min\biggl\{\frac{(1-\mu)\|\tilde \rho\|_1}{|{\ln(I)}|}, \frac{\kappa\pi^2}{|{\ln(I)}|^2} +(1-\mu^{-1})\|\tilde\rho\|_\infty\biggr\} 
		& \ge
		\kappa \frac{\|\tilde\rho\|_1}{|{\ln(I)}|}\min\biggl\{\frac{1}{\beta}, \frac{\pi^2}{\ln(3)^2} -\frac{1}{\beta-\kappa}\biggr\}, 
	\end{align}
	where we assumed the positivity of the second argument (this will be clear by the choice of $\beta$ below).
	Note that the first argument of the minimum is decreasing in $\beta>\kappa$ while the second one is increasing. Thus to find the maximizing $\beta$ we only need to solve the equation $1/\beta = \pi^2/\ln(3)^2-1/(\beta-\kappa)$. Plugging the solution, given by
	\begin{equation}
		\beta(\kappa) = \frac{\pi^2\kappa+\sqrt{\pi^4\kappa^2+4\ln(3)^4}+2\ln(3)^2}{2\pi^2} > \kappa, 
	\end{equation}
	into the above yields
	\begin{align}
		\tilde{H}\ge \frac{\kappa}{\beta(\kappa)}\frac{\|\tilde\rho\|_1}{|{\ln(I)|}} = \frac{\kappa\bar\rho}{\beta(\kappa)}.
	\end{align}
	
	Finally, since $\beta(\kappa)$ is a convex function for $\kappa \in [0, 1]$ we can simplify this expression using 
	\begin{equation}
	 	\beta(\kappa)\le \beta(0)+(\beta(1)-\beta(0))\kappa=:L_\beta(\kappa), 
	\end{equation} 
	and by simple numerical estimates one finds that $L_\beta(\kappa)< \kappa+1/4$.
\end{proof}


\subsection{Number-theoretic structure of the effective scalar potential}\label{sec:potential-structure}

To proceed with the analysis we will need a more precise understanding of how $\rho$ 
depends on the positions of the other particles. Note first that we may assume that 
$\alpha>0$ using the reflection-conjugation symmetry. We then begin by writing for the normalized flux
$$
	\hat\Phi(r) = \alpha(1 + 2\N(r)), \quad r \ge R/2, 
$$
where we introduce the \emph{particle counting function}
\begin{equation}\label{eq:Particle_Counting_Function}
	\N(r) := \sum_{l=1}^{N-2} \int_{B_r(0)} \frac{\1_{B_R(\by_l)}}{\pi R^2}. 
\end{equation}
Recall that in the expression \eqref{eq:def_anyonPhi} for the flux $\hat\Phi$,
all particles are treated relative to the
fixed center of mass $\bX$ of the considered particle pair, and have also been renumbered for convenience:
${\by_l := \bx_l-\bX \in \R^2}$, with $l \in \{1, \ldots, N-2\}$.

In terms of the function $\N$ we have that
\begin{equation}\label{eq:rho_q_Identity}
	\rho(r) = \min_{q \in \Z} \bigl( \alpha(1+2\N(r))-2q \bigr)^2, 
	\qquad \N(r) = \frac{1}{2\alpha}\hat{\Phi}(r) - \frac{1}{2}, 
\end{equation}
and we may cover the interval $[R/2, L]$ by smaller intervals
$J_q$ labeled by the minimizer $q \in \NN$ 
(note the monotonicity of the function $\N(r)$, and that we might already have 
$q \gg 1$ on the first such interval at $r=R/2$).
Each $J_q$ contains, 
except possibly for the first and last such interval, 
exactly one zero of $\rho$ which we denote by $r_q$:
$$
	\rho(r_q) = ( \alpha(1+2\N(r_q))-2q )^2 = 0 
	\quad \Leftrightarrow \quad 
	\N(r_q) = \frac{q}{\alpha} - \frac{1}{2}, 
$$
so that
\begin{equation}\label{eq:NAwayFromZ}
	|\N(r_q) - p| = \frac{1}{2\alpha}| (2p+1)\alpha - 2q | 
	\ge \frac{\alpha_N}{2\alpha}
	\quad \forall p \in \{ 0, 1, \ldots, N-2 \}.
\end{equation}
We then also have the very useful identity 
\begin{align}
	\rho(r) &= |\alpha(1+2\N(r))-2q|^2
	= |\alpha(1 + 2\N(r))-\alpha(1+2\N(r_q))|^2 \nonumber\\
	&= 4\alpha^2|\N(r)-\N(r_q)|^2, \label{eq:rhoInTermsOfN}
\end{align}
whenever $r \in J_q$.
Let us denote by $e^\limminus_q$ and 
$e^\limplus_q$ the nearest points 
to the left resp.\ right of $r_q$ where $\rho(r) = 1$\footnote{Typically we have that
$e^\limplus_q = e^\limminus_{q+1}$ and $J_q=[e^\limminus_q, e^\limplus_q]$ unless
$\rho$ stabilizes at $1$ on some interval between $r_q$ and $r_{q+1}$, 
in which case $e^\limplus_q < e^\limminus_{q+1}$ and the intervals $J_q$ and $J_{q+1}$ overlap.}, then
$$
	\rho(e^\limpm_q) = 1, \quad \mbox{and}\quad
	\rho(r) = 4\alpha^2 (\N(r)-\N(r_q))^2 < 1
	\quad \forall r \in (e^\limminus_q, e^\limplus_q) \subseteq J_q. 
$$
Finally, we also denote by $z^\limminus_q$ and $z^\limplus_q$ 
the nearest points to the left resp.\ right of $r_q$ where 
$\N(z^\limminus_q), \N(z^\limplus_q) \in \Z$, 
and hence $\N(z^\limplus_q) - \N(z^\limminus_q) = 1$, 
and we observe due to \eqref{eq:NAwayFromZ}, \eqref{eq:rhoInTermsOfN} 
and monotonicity that
\begin{equation}\label{eq:rhoOutsideIk}
	\rho(r) \ge \alpha_N^2 \qquad 
	\forall r \in J_q \setminus (z^\limminus_q, z^\limplus_q).
\end{equation}
Recall that this constant depends in a non-trivial way on number-theoretic
aspects of the parameter $\alpha$, and that it remains bounded away from
zero for all $N$ if and only if $\alpha$ is an odd-numerator rational number
(see~\cite[Proposition~5]{LunSol-13a}). To clarify the above definitions, 
two sets of points $r_q$, $e^\limpm_q$, $z^\limpm_q$ are illustrated in Figure~\ref{fig:N_plot}
for a particular particle configuration.
\begin{figure}[ht]
	\centering
	\begin{tikzpicture}
		\node [above right] at (0,0) {\includegraphics[scale=0.9]{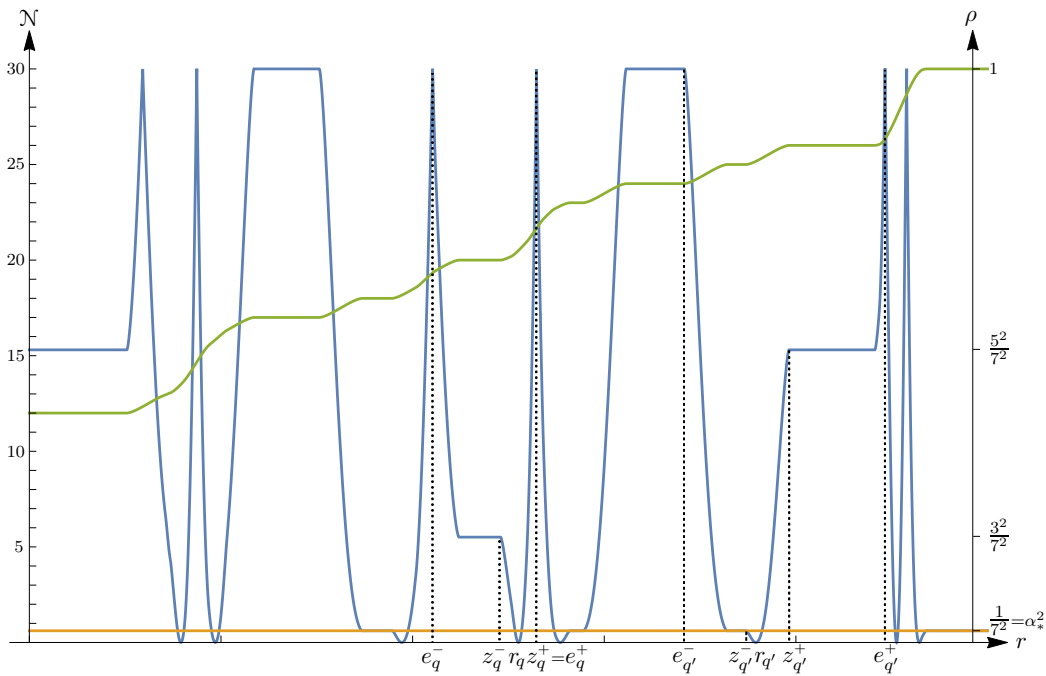}};
		\node [above right] at (0.17, 8.255) {\scalebox{0.8}{$\N$}};
		\node [above right] at (13.27, 0) {\scalebox{0.8}{$r$}};
		\node [above right] at (5.45,-0.31) {\scalebox{0.8}{$e_{q}^\limminus$}};
		\node [above right] at (6.25,-0.31) {\scalebox{0.8}{$z_{q}^\limminus$}};
		\node [above right] at (6.6,-0.28) {\scalebox{0.8}{$r_{\!q}$}};
		\node [above right] at (6.83,-0.31) {\scalebox{0.8}{$z_{q}^\limplus{\scriptstyle =} e_{q}^\limplus$}};
		\node [above right] at (8.75,-0.33) {\scalebox{0.8}{$e_{q'}^\limminus$}};
		\node [above right] at (9.486,-0.33) {\scalebox{0.8}{$z_{q'}^\limminus$}};
		\node [above right] at (9.836,-0.28) {\scalebox{0.8}{$r_{\!q'}$}};
		\node [above right] at (10.2,-0.33) {\scalebox{0.8}{$z_{q'}^\limplus$}};
		\node [above right] at (11.4,-0.33) {\scalebox{0.8}{$e_{q'}^\limplus$}};
		\node [above right] at (12.9,0.18) {\scalebox{0.8}{$\frac{1}{7^2}{\scriptstyle=\alpha_*^2}$}};
		\node [above right] at (12.9,1.29) {\scalebox{0.8}{$\frac{3^2}{7^2}$}};
		\node [above right] at (12.9,3.77) {\scalebox{0.8}{$\frac{5^2}{7^2}$}};
		\node [above right] at (12.92,7.62) {\scalebox{0.8}{$\scriptstyle{1}$}};
		\node [above right] at (12.6,8.255) {\scalebox{0.8}{$\rho$}};
		\node [above right] at (0.1,1.27) {\scalebox{0.8}{$\scriptstyle{5}$}};
		\node [above right] at (0,2.54) {\scalebox{0.8}{$\scriptstyle{10}$}};
		\node [above right] at (0,3.81) {\scalebox{0.8}{$\scriptstyle{15}$}};
		\node [above right] at (0,5.08) {\scalebox{0.8}{$\scriptstyle{20}$}};
		\node [above right] at (0,6.35) {\scalebox{0.8}{$\scriptstyle{25}$}};
		\node [above right] at (0,7.62) {\scalebox{0.8}{$\scriptstyle{30}$}};
	\end{tikzpicture}
	\caption{The function $\N(r)$ (green) together with $\rho(r)$ (blue) and 
	$\alpha_*^2$ (yellow), for $\alpha=3/7$, over an interval where the enclosed 
	number of particles increases from $12$ to $30$. Two separate zeros 
	$r_q$ and $r_{q'}$ of $\rho$, with $q'=q+2$, are indicated together with the 
	corresponding points $z_q^\limpm, e_q^\limpm$ and $z_{q'}^\limpm, e_{q'}^\limpm$.}
	\label{fig:N_plot}
\end{figure}


Hence, we can reduce our problem to studying precisely those smaller intervals around
each zero of $\rho$ not covered by \eqref{eq:rhoOutsideIk}.  
To this end we let $I_q$ denote the interval $(z_q^\limminus, z_q^\limplus)$ around the zero ${r_q\in [R/2, L]}$. When considering a fixed $I_q$ we may for notational simplicity drop the subscripts $q$ when referring to its endpoints. 
Observe by the size of each particle that $|I_q| \le 2R$, 
and furthermore that there is always at least one particle covering 
the entire interval:

\begin{lemma}\label{lem:OneParticleCoverIk}
	If $r_q \ge R/2$ is a zero of $\rho$ then with $I_q$ constructed as above there 
	exists a particle centered at $\by_l$, at a distance $d=|\by_l|=|\bx_l-\bX|$, 
	such that $I_q \subseteq [d-R, d+R]$. 
	In other words, the angular projection of some particle completely 
	covers $I_q$.
\end{lemma}

\begin{proof}
	Let $I_q=(z^\limminus, z^\limplus)$ and let $\widetilde\N(r)$ be the 
	particle counting function corresponding to our particle configuration 
	but where we remove all particles (seen as closed disks $\bar B_R(\by_l)$) 
	that have empty intersection with the closed disk $\bar B_{z^\limminus}(0)$, 
	i.e.\ we remove all particles that are centered at a distance strictly 
	larger than $z^\limminus+R$ from the origin. 
	By the construction of $I_q$, there is at least one particle that has 
	non-empty intersection with $\partial B_{z^\limminus}(0)$
	(not counting any fully enclosed ones), since 
	otherwise $\N(r)$ would be constant here which 
	contradicts the choice of $z^\limminus$. 
	Let now $r'$ be the radius such that all the particles that intersected 
	$\partial B_{z^\limminus}(0)$ are completely contained in the closed 
	disk $\bar B_{r'}(0)$. 
	By the construction of $\widetilde\N(r)$, 
	its value at $r'$ is an integer which, 
	since there were particles intersecting 
	$\partial B_{z^\limminus}(0)$, is strictly larger than 
	$\widetilde\N(z^\limminus)=\N(z^{\limminus})$. 
	But then, since $\widetilde\N(r)\le \N(r)$, the function $\N(r)$ must 
	take at least one integer value on $(z^\limminus, r']$. 
	Thus, by the definition of $z^\limplus$ we conclude that 
	$z^\limplus\le r'$, which completes the proof.
\end{proof}


\subsection{Geometric structure of the particle counting function}\label{sec:GeometricCalculationsForN}

To proceed we will need more information on the local behavior of 
the particle counting function $\N(r)$. We note that 
\begin{equation}
	\N(r)=\sum_{l=1}^{N-2}\frac{|B_r(0) \cap B_R(\by_l)|}{\pi R^2}, 
\end{equation}
where $\by_l$ are the centers (in relative coordinates) of the $N-2$ particles 
not in our presently studied pair.

To analyze $\N(r)$ we thus need to work with the area of the intersection of pairs of disks. An elementary, although slightly tedious, calculation yields the following expression.  
\begin{proposition}\label{prop:the_area_of_intersecting_disks}
Let $B_1=B_{r_1}(\bx_1)$ and $B_2=B_{r_2}(\bx_2)$ be disks of radii 
$r_1, r_2$, with $r_1\le r_2$, centered at the points $\bx_1$ and $\bx_2$. 
Then with $d=|\bx_1-\bx_2|$ we have for the area of intersection, 
in the non-trivial regime $d\le r_1+r_2$ and $d+r_1\ge r_2$, that
\begin{align*}
	|B_1\cap B_2| 
	=\ & 
	r_1^2 \arccos\Bigl(\frac{d^2+r_1^2-r_2^2}{2 d r_1}\Bigr)+ r_2^2\arccos\Bigl(\frac{d^2+r_2^2-r_1^2}{2 d r_2}\Bigr)\\
	&-
	\frac{1}{2}\sqrt{(-d+r_1+r_2)(d+r_1-r_2)(d-r_1+r_2)(d+r_1+r_2)}.
\end{align*}
If $d>r_1+r_2$ the area is zero and if $d+r_1 < r_2$ the area is $\pi r_1^2$.
\end{proposition}

Differentiating the flux contribution from a single particle located at 
$\by_l \in \R^2$, given by 
\begin{equation}
	F(|\by_l|, r) := |B_r(0)\cap B_R(\by_l)|/(\pi R^2), 
\end{equation}
we find for arbitrary $d, r \ge 0$ that 
\begin{equation}\label{eq:oneParticleProfile}
	f(d, r):=\frac{\partial}{\partial r}F(d, r) = \left\{\begin{array}{ll}
	2r/ R^2, & \textrm{if }r\le R-d, \\[3pt]
	0, & \textrm{if }r>R+d \textrm{ or }r<d-R, \\
 \frac{2 r}{\pi R^2}\arccos\Bigl(\frac{d^2+r^2-R^2}{2 d r}\Bigr), & \textrm{otherwise.}
\end{array}\right.
\end{equation}

In what follows we will frequently use that $f(d, \cdot\,)$ is essentially concave on its support (compare Figure~\ref{fig:particleProfile}); 
the precise statement and its proof is found in Appendix~\ref{app:concavity}.
Furthermore, it satisfies some simple bounds:

\begin{lemma}\label{lem:ProfileBounds}
	With $f(d, \cdot\, )$ denoting the one-particle profile~\eqref{eq:oneParticleProfile} we have for 
	$d\ge 0$ and $r\ge R$ the following bounds:
	\begin{align}\label{eq:DerivativeBounds}
		f(d, r) &\le f_\sqcap(d, r):= \frac{2}{R} \1_{(d-R, d+R)}(r), \\
		f(d, r) &\ge f_{\wedge}(d, r):= \frac{2(R-d+r)}{\pi R^2}\1_{(d-R, d)}(r)+ \frac{2(d+R-r)}{\pi R^2}\1_{[d, d+R)}(r).
	\end{align}
\end{lemma}

\begin{proof}[Proof of Lemma~\ref{lem:ProfileBounds}]
	The upper bound for $f$ given by the lemma is clear from the geometric 
	construction of $f$ and $F$. 
	The value of $f$ is equal to the length of the circle segment 
	$\partial B_r(0)\cap B_R(\bx)$ where $|\bx|=d$, divided by $\pi R^2$, 
	and clearly this cannot exceed $2/R$. For the lower bound we use concavity. 

	For $d\ge R$ the function $f(d, \cdot\, )$ 
  	is concave on its support $[d-R, d+R]$ (see Appendix~\ref{app:concavity}). 
	Moreover, $f_\wedge(d, \cdot\,)$ is continuous, piecewise linear and has the same support as $f(d, \cdot\,)$.
	By the construction of $f_\wedge$ and the concavity of $f(d, \cdot\,)$ it suffices to prove that the inequality
	holds at the maximum of $f_\wedge(d, \cdot\,)$, i.e.\ that
	$f(d, d)\ge f_\wedge(d, d)$, 
	which is clear:
	for $d\ge R$ we have that $f(d, d)$ is a decreasing function and that 
	$\lim_{d\to\infty}f(d, d)=\frac{2}{\pi R}=f_\wedge(d, d)$. 

	For $d<R$ we have that $f(d, \cdot\, )$ and $f_\wedge(d, \cdot\, )$ are 
	concave on $[R, d+R]$ and zero otherwise (see Appendix~\ref{app:concavity}). 
	By the linearity of $f_\wedge(d, \cdot\,)$ on this interval it is
  	sufficient to prove that $f(d, R) \ge f_\wedge(d, R)$, which follows since 
	$f(d, R)=\frac{2}{\pi R}\arccos\bigl(\frac{d}{2 R}\bigr) \ge \frac{2d}{\pi R^2}=f_\wedge(d, R)$.
\end{proof}

\begin{figure}[ht]
	\centering
	\begin{tikzpicture}
		\node [above right] at (0,0) {\includegraphics[scale=0.86]{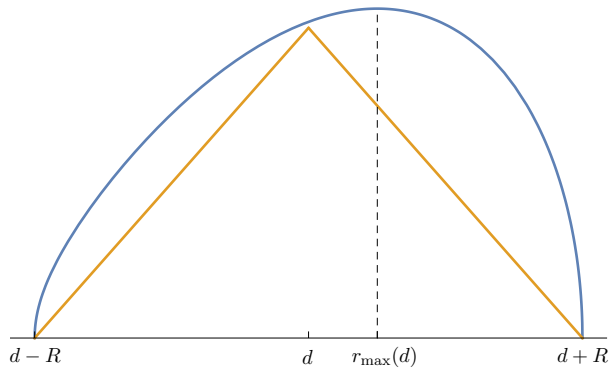}};
		\node [above right] at (0,0.05) {\scalebox{0.7}{$d-R$}};
		\node [above right] at (3.85,0.05) {\scalebox{0.7}{$d$}};
		\node [above right] at (4.5,0) {\scalebox{0.7}{$r_{\textrm{max}}(d)$}};
		\node [above right] at (7.2,0.05) {\scalebox{0.7}{$d+R$}};
	\end{tikzpicture}
	\caption{The one-particle profile $f(d, \cdot\, )$ and its lower bound 
		$f_\wedge(d, \cdot\, )$ plotted over the support of $f$. 
		The profile depicted is for $d = 3R/2$, while as $d$ increases 
		this profile more and more resembles the upper half of a disk.}
	\label{fig:particleProfile}
\end{figure}


The following lemma captures in a convenient form essential aspects of the shape
of the particle profile, and will play an important role in the analysis 
on intervals of oscillation below.

\begin{lemma}[Shape lemma]\label{lem:ShapeLemma}
	If $r\in [r_1, r_2]$ with $r_1\ge R$ and $r_2-r_1 \le R/2$, we have that
	\begin{equation}
		\N\hspace{0.5pt}'(r_1)+\N\hspace{0.5pt}'(r_2) \ge \N\hspace{0.5pt}'(r).
	\end{equation}
\end{lemma}

\begin{remark} 
	The assumption $r_2-r_1\le R/2$ can be relaxed slightly by instead 
	requiring that $r_1$ is sufficiently large. 
	In particular, in the limit $r_1\to \infty$ the one-particle profile 
	approaches a half disk and it is then geometrically clear that the 
	statement holds whenever $r_2-r_1\le R$. 
\end{remark}

\begin{proof}[Proof of Lemma~\ref{lem:ShapeLemma}]
By linearity it is sufficient to prove that the inequality holds with 
$\N\hspace{0.5pt}'$ replaced by the one-particle profile 
$f(d, \cdot\, )$ for any $d\ge 0$.

The proof utilizes that the profile $f(d, \cdot\, )$ is concave on its support 
intersected with $[R, \infty)$, which is shown in Appendix~\ref{app:concavity}. 
If in addition $d\ge R$ the profile is concave on its full support 
$(d-R, d+R)$, also shown in Appendix~\ref{app:concavity}. 
Thus, whenever $(r_1, r_2)$ does not contain the maximum of 
$f(d, \cdot\, )$ the 
statement is clear, since if this is the case $f(d, \cdot\, )$ is monotone here 
and thus has its maximum value in either $r_1$ or $r_2$.

Thus we may assume that the unique maximum of $f(d, \cdot\, )$ is attained at a point 
$r_{\textrm{max}}(d)$ in $(r_1, r_2)$. Moreover, by the concavity of 
$f(d, \cdot\, )$ it suffices to consider the case when $|r_2-r_1|=R/2$. 
The inequality we wish to prove can now be written as
\begin{equation}\label{eq:reducedShapeLemma}
	f(d, r_{\textrm{max}}(d))\le f(d, r_1)+f(d, r_1+R/2), 
\end{equation}
which should hold for all $r_1\ge R$ such that 
$r_{\textrm{max}}(d)\in (r_1, r_1+R/2)$. 

{\bf Case 1: $d\ge R$.} In this case it holds that 
$(r_1, r_1+R/2) \subseteq (d-R, d+R)$, since 
$r_{\textrm{max}}(d) \in (d-R/2, d+R/2)$.
This can be verified by considering $\frac{\partial}{\partial r}f(d, r)|_{r=d+R/2}$, 
which can be shown by straightforward computation to be decreasing in $d$ and
moreover it is negative at $d=R$.
Similarly, $\frac{\partial}{\partial r}f(d, r)|_{r=d-R/2}$
can be verified to be positive, and hence $d-R/2 < r_{\textrm{max}}(d) < d+R/2$.

This implies that the right-hand side of~\eqref{eq:reducedShapeLemma} is a 
concave function of $r_1$, and hence its minimum value is attained at one of 
the extremal points of the allowed intervals. But this is precisely when either 
$r_1$ or $r_2$ is equal to $r_{\textrm{max}}(d)$, in which case the statement 
is trivial by the non-negativity of $f$.

{\bf Case 2: $d\le 2R/3$.} 
By similar calculations as in Case 1, we have that
$\frac{\partial}{\partial r}f(d, r)|_{r=R} < 0$ for $d \le 2R/3$.
Then by concavity $f(d, \cdot\,)$ is a monotonically decreasing function on $[R, d+R]$.
Thus $f(d, r_1)\ge f(d, r)$ and the statement follows.

{\bf Case 3: $2R/3 < d < R$.} Again the function $f(d, \cdot\, )$ is concave on 
$(R, d+R)$. Thus we again only need to consider the extremal cases of the 
intervals $(r_1, r_2)$ containing the maximum of $f(d, \cdot\, )$ on this 
interval. This reduces to three different options. Either $r_1=R$, 
or $r_2=d+R$, or one of the endpoints of the interval is located at the maximum. 
In the last case the statement is trivially true.

If we were in the second option then $(r_1, r_2)=(d+R/2, d+R)$. 
Through a similar computation as above one checks that on this interval $f(d, \cdot\, )$ 
is monotone, and hence the statement follows. 

If however $(r_1, r_2)=(R, 3R/2)$ the inequality is reduced to
\begin{equation}
	f(d, r)\le f(d, R)+f(d, 3R/2).
\end{equation}
By scaling we may without loss of generality assume that $R=1$. Using the explicit expression of $f$ we need to show that
\begin{equation}
	\frac{2r}{\pi}\arccos\biggl(\frac{d^2+r^2-1}{2dr}\biggr)\le \frac{2}{\pi}\arccos\biggl(\frac{d}{2}\biggr)+\frac{3}{\pi}\arccos\biggl(\frac{d^2+5/4}{3d}\biggr).
\end{equation}
Since for
$d\le R=1$ we have that $r_{\textrm{max}}(d)\le 3/2$, it follows that
\begin{equation}
	f(d, r_{\textrm{max}}(d)) \le \frac{3}{\pi}\arccos\biggl(\frac{d^2+r_{\textrm{max}}(d)^2-1}{2d r_{\textrm{max}}(d)}\biggr).
\end{equation}
But the function $\frac{3}{\pi}\arccos\bigl(\frac{d^2+r^2-1}{2d r}\bigr)$ is decreasing in $r$, for $1\le r\le d+1$, and thus we only need to verify the inequality
\begin{equation}
	\frac{3}{\pi}\arccos\biggl(\frac{d}{2}\biggr)\le \frac{2}{\pi}\arccos\biggl(\frac{d}{2}\biggr)+\frac{3}{\pi}\arccos\biggl(\frac{d^2+5/4}{3d}\biggr).
\end{equation}
This is equivalent to $\arccos\bigl(\frac{d}{2}\bigr)\le 3\arccos\bigl(\frac{d^2+5/4}{3d}\bigr)$. 
We observe that the left-hand side of this inequality is decreasing whilst the right is increasing. Thus it suffices to check the validity at $d=2/3$, which is a simple numerical evaluation.
\end{proof}


\subsection{Local bounds for the mean potential}

In this subsection we use the explicit form of $\N(r)$ uncovered above
for $r\in(R, L)$ 
and the projection argument of Lemma~\ref{lem:1Dprojection}
to locally replace the effective 1-dimensional potential $\rho(r)/r$ 
with some constant times $\alpha_N^2/r$. 
By Lemma~\ref{lem:1Dprojection} it suffices to prove that given an interval
$I\subseteq (R, L)$ of small enough measure we have a suitable 
bound for the weighted mean $\bar\rho$ on $I$. 
On intervals \eqref{eq:rhoOutsideIk} 
where $\rho$ is already larger than $\alpha_N^2$
we need not perform any detailed analysis.
Thus the only intervals that remain are those 
of the form $I_q = (z_q^\limminus, z_q^\limplus) \subseteq J_q$
close to the zeros of $\rho$. 
The analysis is split into several parts depending on the behavior of $\rho$ 
near a specific zero. Our first bound provides a general estimate for $\bar\rho$ 
on any subinterval of the $J_q$ constructed above (Section~\ref{sec:potential-structure}) 
which contains the unique zero of $\rho$ on this interval.


\newcommand{\rzero}{{r_0}}
\begin{lemma}\label{lem:GeneralIntegralBound}
	Let $(r_1, r_2)$, with $r_1\ge R/2$, be such that on this interval 
	$\rho(r)=|\hat\Phi(r)-2q|^2$ for some fixed $q\in \Z$ 
	and such that there exists some 
	$\rzero \in (r_1, r_2)$ with $\rho(\rzero)=0$. 
	Then, with $\delta(r):=\min\{r-r_1, r_2-r\}$, we have that
	\begin{equation}
	 	\int_{r_1}^{r_2}\frac{\rho(r)}{r}\, dr \ge \frac{2\alpha^2}{r_2(r_2-r_1)}\biggl(\int_{r_1}^{r_2}\N\hspace{0.5pt}'(r)\delta(r)\, dr\biggr)^2, 
	\end{equation} 
	where as before $\N(r)$ denotes the particle counting function~\eqref{eq:Particle_Counting_Function}.
\end{lemma}

\begin{proof}[Proof of Lemma~\ref{lem:GeneralIntegralBound}]
	On such an interval $(r_1, r_2)$ we can, 
	according to \eqref{eq:rhoInTermsOfN}, express $\rho$ in terms of $\N$ as
	\begin{align*}
		\rho(r) &= |\alpha(1+2\N(r))-\alpha(1+2\N(\rzero))|^2
		=
		4\alpha^2|\N(r)-\N(\rzero)|^2.
	\end{align*}
	Inserting this into the integral we wish to bound and using the trivial 
	estimate $1/r\ge 1/r_2$, we have that
	\begin{equation}
		\int_{r_1}^{r_2}\frac{\rho(r)}{r}\, dr \ge \frac{4\alpha^2}{r_2}\int_{r_1}^{r_2}|\N(r)-\N(\rzero)|^2\, dr.
	\end{equation}
	We split the above integral into two parts, 
	\begin{align*}
		\int_{r_1}^{r_2}|\N(r)-\N(\rzero)|^2\, dr 
		&=
		\int_{r_1}^{\rzero}(\N(r)-\N(\rzero))^2\, dr +\int_{\rzero}^{r_2}(\N(\rzero)-\N(r))^2\, dr \\
		&=
		\int_{r_1}^{\rzero}\biggl(\int_r^\rzero \N\hspace{0.5pt}'(t)\, dt\biggr)^2dr +
		\int_{\rzero}^{r_2}\biggl(\int_\rzero^r \N\hspace{0.5pt}'(t)\, dt\biggr)^2dr.
	\end{align*}
	Using the Cauchy--Schwarz inequality and changing the order of integration 
	one finds that
	\begin{align*}
		\int_{r_1}^{r_2}|\N(r)-\N(\rzero)|^2\, dr
		&\ge
		\frac{1}{r_2-r_1}\biggl(\biggl(\int_{r_1}^\rzero\int_r^\rzero \N\hspace{0.5pt}'(t)\, dt dr\biggr)^2+\biggl(\int_{\rzero}^{r_2}\int_\rzero^r \N\hspace{0.5pt}'(t)\, dt dr\biggr)^2\biggr)\\
		&=
		\frac{1}{r_2-r_1}\biggl(\biggl(\int_{r_1}^\rzero\int_{r_1}^t \N\hspace{0.5pt}'(t)\, dr dt\biggr)^2+\biggl(\int_{\rzero}^{r_2}\int_t^{r_2} \N\hspace{0.5pt}'(t)\, dr dt\biggr)^2\biggr)\\
		&=
		\frac{1}{r_2-r_1}\biggl(\biggl(\int_{r_1}^\rzero \N\hspace{0.5pt}'(t)(t-r_1)\, dt\biggr)^2+\biggl(\int_{\rzero}^{r_2}\N\hspace{0.5pt}'(t)(r_2-t)\, dt\biggr)^2\biggr).
	\end{align*}
	To obtain the desired estimate we combine the above with the observation 
	that both $t-r_1$ and $r_2-t$ are larger than $\delta(t)$, and the elementary 
	inequality $2(a^2+b^2)\ge (a+b)^2$, 
	\begin{align*}
		\int_{r_1}^{r_2}\frac{\rho(r)}{r}\, dr 
		&\ge
		\frac{4\alpha^2}{r_2(r_2-r_1)}\biggl(\biggl(\int_{r_1}^\rzero \N\hspace{0.5pt}'(r)\delta(r)\, dr\biggr)^2+\biggl(\int_{\rzero}^{r_2}\N\hspace{0.5pt}'(r)\delta(r)\, dr\biggr)^2\biggr)\\
		&\ge
		\frac{2\alpha^2}{r_2(r_2-r_1)}\biggl(\int_{r_1}^{r_2} \N\hspace{0.5pt}'(r)\delta(r)\, dr\biggr)^2.\qedhere
	\end{align*}
\end{proof}


We now study $\bar\rho$ on the intervals $I_q = (z^\limminus_q, z^\limplus_q)$
constructed earlier around zeros of $\rho$, with $\N(z^\limpm_q) \in \Z$.
We begin with a lemma providing a bound for the local weighted mean on a 
certain subclass of these intervals where the potential is in some sense 
well behaved.

\begin{lemma}[Good intervals]\label{lem:GoodIntervals}
	Let $I_q=(z^\limminus, z^\limplus)$ be one of the intervals constructed 
	above which satisfies $z^\limminus\ge R$. Then if either
	\begin{equation}
		|I_q|\ge C R \quad\mbox{or}\quad \frac{\inf_{I_q}\N\hspace{0.5pt}'}{\sup_{I_q}\N\hspace{0.5pt}'}\ge \frac{C^2}{\pi} 
	\end{equation}
	for some $0<C\le 1$, we have that
	\begin{equation}
		\bar\rho_{I_q} := {\int_{I_q}\frac{\rho(r)}{r}\, dr} \biggm/{\int_{I_q}\frac{dr}{r}} \ge \frac{\alpha^2C^4}{24\pi^2}.
	\end{equation}
\end{lemma} 

\begin{remark}
We will later see that for our treatment of intervals $I_q$ that are not 
covered by this lemma we will need to choose $C$ rather small, 
approximately $C \approx 1/10$.
\end{remark}

\begin{proof}[Proof of Lemma~\ref{lem:GoodIntervals}] By Lemma~\ref{lem:GeneralIntegralBound} we may estimate the integral of the potential by
\begin{equation}
	\int_{I_q}\frac{\rho(r)}{r}\, dr \ge \frac{2\alpha^2}{z^\limplus(z^\limplus-z^\limminus)}\biggl(\int_{I_q}\N\hspace{0.5pt}'(r)\delta(r)\, dr\biggr)^2.
\end{equation}

By Lemma~\ref{lem:OneParticleCoverIk} the interval $I_q$ is covered by at least one particle. Thus for $r\in I_q$ we can bound $\N\hspace{0.5pt}'(r)$ from below by using our lower bound for the one-particle profile $f(d, r)$ and minimizing over particle positions $d$ such that $I_q \subseteq (d-R, d +R)$. Let as before $f_\wedge(d, r)$ denote the lower bound for $f$ given by Lemma~\ref{lem:ProfileBounds}. We conclude that
\begin{equation}
	\int_{I_q} \N\hspace{0.5pt}'(r) \delta(r)\, dr \ge \inf_{d\in (z^\limminus-2R, z^\limplus+2R)} \int_{I_q} f_\wedge(d, r)\delta(r)\, dr.
\end{equation}
As this integrand is piecewise linear in $d$ we must have that the integral is minimized in one of the extremal points: a particle starting at $z^\limminus$, a particle ending at $z^\limplus$ or a particle centered at $(z^\limplus-z^\limminus)/2$. By symmetry the last alternative maximizes the integral and thus we can discard this option. Moreover, the same symmetry implies that the first two alternatives are equal. Through a straightforward calculation we find that
\begin{equation}
	\int_{I_q} \N\hspace{0.5pt}'(r) \delta(r)\, dr \ge 
	\frac{1}{4\pi}\left\{\begin{array}{ll}
		|I_q|^3/R^2, & \textrm{if } |I_q|\le R, \\[3pt]
		|I_q|, & \textrm{if } |I_q| >R.
	\end{array}\right.
\end{equation}
Thus if $|I_q|\ge CR$, $0<C\le1$, the above yields
\begin{align*}
	 \int_{I_q}\frac{\rho(r)}{r} \, dr &\ge \frac{\alpha^2 C^4}{8\pi^2 z^\limplus} |I_q|.
\end{align*}

If instead of $|I_q|\ge CR$ we have that
\begin{equation}\label{eq:InfSupRatioLarge}
	\frac{\inf_{I_q}\N\hspace{0.5pt}'}{\sup_{I_q}\N\hspace{0.5pt}'}\ge \frac{C^2}{\pi}
\end{equation}
we can obtain the same bound. Namely, if we again consider the bound given by Lemma~\ref{lem:GeneralIntegralBound}, 
\begin{equation}
	\int_{I_q} \frac{\rho(r)}{r} \, dr \ge
	\frac{2\alpha^2}{z^\limplus(z^\limplus-z^\limminus)}
	\biggl(\int_{I_q} \N\hspace{0.5pt}'(r) \delta(r)\, dr \biggr)^2, 
\end{equation}
we find, using
$\int_{I_q} \delta(r) \, dr = |I_q|^2/4$, 
that
\begin{equation}
	\int_{I_q} \frac{\rho(r)}{r} \, dr \ge
	\frac{\alpha^2}{8 z^\limplus}(\inf_{I_q}\N\hspace{0.5pt}')^2|I_q|^3 \ge \frac{\alpha^2 (\inf_{I_q}\N\hspace{0.5pt}')^2}{8 z^\limplus(\sup_{I_q}\N\hspace{0.5pt}')^2}|I_q| 
	\ge
	\frac{\alpha^2 C^4}{8\pi^2z^\limplus}|I_q|, 
\end{equation}
where we also used that $(\sup_{I_q}\N\hspace{0.5pt}') |I_q| \ge  \int_{I_q} \N\hspace{0.5pt}' = 1$ 
for each $q$.

For the weighted mean we now find that
\begin{align*}
	\bar \rho_{I_q} &= \frac{\int_{I_q}\rho(r)/r \, dr}{\int_{I_q}1/r\, dr} \ge \frac{z^\limminus}{|I_q|}\int_{I_q}\frac{\rho(r)}{r}\, dr \ge 
	\frac{z^\limminus \alpha^2 C^4}{z^\limplus 8\pi^2}
	\ge \frac{\alpha^2 C^4}{24 \pi^2}, 
\end{align*} 
where we used that $|I_q|\le 2R$ and $z^\limminus\ge R$ implies that $z^\limminus/z^\limplus \ge 1/3$. 
\end{proof}

The previous lemma does not cover the scenario where $\N(r)$ increases rapidly, resulting in rapid oscillations on many short intervals $I_q$. In the next lemma we consider the remaining intervals $I_q$ and use our 
geometric knowledge of $\N(r)$ to show that these intervals cannot cover too much 
of our large-scale interval $[R, L]$. To achieve this we first cover the 
remaining collection of intervals $I_q$ with a collection of intervals $J_l$ 
such that $|J_l|=R/2$ for all $l$.


\begin{lemma}[Bad intervals]\label{lem:BadIkHasSmallMeasure}
	Let $J\subseteq (R, L]$ be an interval of length $R/2$. Then the fraction of $J$ covered by intervals $I_q$ satisfying both
	\begin{equation}\label{eq:AssumptionsOnRemainingZeros}
		|I_q|<CR \quad\mbox{and}\quad \frac{\inf_{I_q}\N\hspace{0.5pt}'}{\sup_{I_q}\N\hspace{0.5pt}'}<\frac{C^2}{\pi}, 
	\end{equation}
	with $C < \sqrt{\pi/2}$, is less than
	\begin{equation}
		\frac{8C(\pi-C^2)}{\pi-2C^2}.
	\end{equation}
\end{lemma}

\begin{proof} 
	Let $\{I_k\}_{k=1}^m$ denote the subset of the intervals $I_q$ for 
	which~\eqref{eq:AssumptionsOnRemainingZeros} is satisfied and 
	$J\cap I_k\neq \emptyset$ for each $k=1, \ldots, m$, and ordered from left 
	to right (note in particular that throughout this proof the labeling of
	the intervals differs from that described below~\eqref{eq:rho_q_Identity}). 
	For further notational convenience we will let $\inf_k$ and $\sup_k$ 
	denote $\inf_{I_k}\N\hspace{0.5pt}'$ and $\sup_{I_k}\N\hspace{0.5pt}'$, respectively. 
	We will also denote by $i_k$ and $s_k$ a (fixed) choice of points in 
	each $I_k$ such that $\N\hspace{0.5pt}'(i_k)=\inf_k$ and $\N\hspace{0.5pt}'(s_k)=\sup_k$.

	We begin by showing that we may assume that the distance between any
	two points in two consecutive intervals is less than $R/2$, 
	allowing us to apply Lemma~\ref{lem:ShapeLemma}. 
	If, for some $k\in\{1, \ldots, m\}$, $I_k=(z^\limplus_k, z^\limminus_k)$ 
	and $I_{k+1}=(z^\limplus_{k+1}, z^\limminus_{k+1})$ are such that 
	$z^\limplus_{k+1}-z^\limminus_k>R/2$, then since both intervals have 
	non-empty intersection with $J$ we must have that $m=2$. 
	But this implies that ${|J\cap (\cup_{k=1}^m I_k)| \le 2CR}$ and the 
	statement follows. Similarly the statement is true if $m=1$.

	Suppose that there exists a $j$ such that $i_j < s_j <s_{j+1}<i_{j+1}$. 
	Then, since we by the above may assume that $i_{j+1}-i_j<R/2$, 
	Lemma~\ref{lem:ShapeLemma} implies that 
	\begin{equation}
		\max\{\sup\nolimits_j, \sup\nolimits_{j+1}\}
		\le \inf\nolimits_j+\inf\nolimits_{j+1}.
	\end{equation}
	But combined with~\eqref{eq:AssumptionsOnRemainingZeros} this leads to a contradiction:
	\begin{align}
	 	\max\{\sup\nolimits_j, \sup\nolimits_{j+1}\} 
	 	\le \inf\nolimits_j+\inf\nolimits_{j+1} 
	 	\le 2\max\{\inf\nolimits_j, \inf\nolimits_{j+1}\} 
	 	< \frac{2C^2}{\pi}\max\{\sup\nolimits_j, \sup\nolimits_{j+1}\}, 
	\end{align} 
	which is impossible since $\frac{2C^2}{\pi} < 1$.

	Let us say that an interval $I_k$ where $s_k<i_k$ is of type A, and one 
	where instead $i_k<s_k$ is of type B 
	(note that $i_k \neq s_k$ by the assumption on $I_k$). 
	We let $\mathcal{A}$ and $\mathcal{B}$ denote the collections of intervals 
	of type A and type B respectively. 

	The above contradiction argument yields that an interval of type A cannot follow one of type B, 
	i.e.\ if we for some $j$ have that $I_j\in \mathcal{A}$ then $I_k\in\mathcal{A}$ for all $k<j$, 
	and similarly, if $I_j\in\mathcal{B}$ then $I_k\in\mathcal{B}$ for all $k>j$. 
	We conclude that there is at most one $k$ such that $I_k$ and $I_{k+1}$ are of different type, and $I_k$ must then be of type A.

	As we will now show, it turns out that the sequence of lengths $|I_k|$ of 
	consecutive intervals starting at any interval of type A and going to the 
	left, resp.\ type B and going to the right, is monotonically decreasing and 
	bounded from above by a geometric sequence. 
	By assumption~\eqref{eq:AssumptionsOnRemainingZeros}, all $|I_k|<CR$, 
	and in particular this holds for the first interval in any such sequence. 
	Using these observations we will be able to bound the total measure of 
	$\cup_k I_k$. 

	We begin by studying a sequence starting at an interval of type A and going to the left (note that such a sequence may not exist if all $I_k \in \mathcal{B}$). We wish to prove that $|I_k|$ decreases along this sequence.

	Let $j$ be such that $I_j\in\mathcal{A}$. Then 
  $i_{j-1}<s_{j}<i_j$, and by Lemma~\ref{lem:ShapeLemma} we have that $\sup_{j}\le \inf_{j-1}+\inf_j$. Since we assume that $\inf_{j}< C^2/\pi \sup_{j}$ this implies that
	\begin{equation}
		\frac{\pi-C^2}{C^2}\inf\nolimits_j
		< \Bigl(1-\frac{C^2}{\pi}\Bigr)\sup\nolimits_j 
		< \inf\nolimits_{j-1}.
	\end{equation}
	The only thing we used above was that $I_j \in \mathcal{A}$. Since this implies that also
 	$I_{j-1}\in \mathcal{A}$, we can iterate this argument until we reach $I_1$. 
	This yields for $k<j$ that
	\begin{align}\label{eq:inf_sup_ineq}
		\Bigl(\frac{\pi-C^2}{C^2}\Bigr)^{j-k} \inf\nolimits_j 
		< \Bigl(\frac{\pi-C^2}{C^2}\Bigr)^{j-k}\frac{C^2}{\pi}\sup\nolimits_j 
		< \inf\nolimits_{k}.
	\end{align}

	Using that $|I_k|\inf_k \le 1 \le |I_k|\sup_k$ (for any $k$) we, for $k<j$, find that~\eqref{eq:inf_sup_ineq} implies 
	\begin{equation}
		|I_{j}|\ge \frac{1}{\sup_j} \ge \Bigl(\frac{\pi-C^2}{C^2}\Bigr)^{j-k}\frac{C^2}{\pi} \frac{1}{\inf_{k}} \ge \Bigl(\frac{\pi-C^2}{C^2}\Bigr)^{j-k}\frac{C^2}{\pi} |I_{k}|, 
	\end{equation}
	where we used that, for $k\le j$, $\inf_{k}>0$ since otherwise $\sup_j$ would be zero which cannot happen by the construction of the $I_k$'s. Since $C$ is small this proves the claim in the case of type A intervals. 

	For the case of type B intervals the proof is almost identical and one finds instead that, if $I_j \in \mathcal{B}$, 
	\begin{equation}
		|I_{j}| \ge \Bigl(\frac{\pi-C^2}{C^2}\Bigr)^{k-j} \frac{C^2}{\pi}|I_{k}|, 
		\qquad k > j.
	\end{equation}

	We are now ready to complete the proof of the lemma. Begin by finding $j$ such that $I_j\in\mathcal{A}$ and $I_{j+1}\in\mathcal{B}$ (if $\mathcal{A}$, alt.\ $\mathcal{B}$, is the empty set we take $j=0$, alt.\ $j=m$). Then using the above estimates we obtain that
	\begin{align}
		\biggl|J \cap \bigcup_k I_k \biggr| &\le \sum_k |I_k| = \sum_{k\le j}|I_k|+ \sum_{k>j}|I_k| 
		\\
		&\le 
		|I_j|\Bigl(1+
		\frac{\pi}{C^2}\sum_{k=1}^{j-1}\Bigl(\frac{C^2}{\pi-C^2}\Bigr)^{j-k} \Bigr)
		+
		|I_{j+1}|\Bigl(1+ \frac{\pi}{C^2}\sum_{k=j+2}^{m}\Bigl(\frac{C^2}{\pi-C^2}\Bigr)^{k-j-1}\Bigr)\\
		&< 
		CR\Bigl(1+
		\frac{\pi}{C^2}\sum_{l=1}^{\infty}\Bigl(\frac{C^2}{\pi-C^2}\Bigr)^l \Bigr)
		+
		CR\Bigl(1+ \frac{\pi}{C^2}\sum_{l=1}^{\infty}\Bigl(\frac{C^2}{\pi-C^2}\Bigr)^{l}\Bigr)\\
		&=
		\frac{4C(\pi-C^2)}{\pi-2C^2}R, 
	\end{align}
	and dividing this quantity by $|J|=R/2$ completes the proof.
\end{proof}


\subsection{Proof of Theorem~\ref{thm:MainRadialBound}}\label{sec:ProofLongRange}

What we have found is that the Lebesgue measure of the subset of $J$ 
where $\rho$ is already large, or 
can be averaged to be large,  
is at least  
\begin{equation}\label{eq:FractionOfJ}
	\Bigl(\frac{1}{2}-\frac{4C(\pi-C^2)}{\pi-2C^2}\Bigr)R.
\end{equation}
Using this we can find a non-trivial uniform lower bound on $\bar{\rho}_J$ and
therefore, using the local projection argument, we finally obtain that there 
exists a constant $c(\kappa)>0$ such that
\begin{equation}
	\int_{R}^{L} \Bigl(|u'|^2+\frac{\rho}{r^2}|u|^2\Bigr)r\, dr \ge \int_{R}^{L}\Bigl((1-\kappa)|u'|^2+c(\kappa)^2\frac{\alpha_N^2}{r^2}\1_{[3R, L-3R]}|u|^2\Bigr)r\, dr.
\end{equation}

We proceed as follows:
\begin{align}\label{eq:SmearingCalculations}
	\int_{R}^{L}\Bigl(|u'|^2+\frac{\rho}{r^2}|u|^2\Bigr)r\, dr 
	&=
	\int_{R}^{L}(1-\kappa)|u'|^2r\, dr + \int_{R}^{L}\Bigl(\kappa|u'|^2+\frac{\rho}{r^2}|u|^2\Bigr)r\, dr\\
	&\ge
	\int_{R}^{L}(1-\kappa)|u'|^2r\, dr + \int_{R}^{L}\Bigl(\frac{\kappa}{2}|u'|^2+ \frac{\hat \rho}{r^2}|u|^2\Bigr)r\, dr, 
\end{align}
where $\hat\rho$ denotes a new weight obtained by replacing $\rho(r)$ with 
$\frac{2\kappa}{1+2\kappa}\frac{\alpha^2C^4}{24\pi^2}$ on all $I_q$ covered 
by Lemma~\ref{lem:GoodIntervals} that intersect $(3R, L-2R)$, by using 
Lemma~\ref{lem:1Dprojection} with $\kappa/2$. 
Thus the only remaining zeros of $\hat\rho$ on $(3R, L-2R)$
are those contained in intervals $I_q$ which satisfy the assumptions of 
Lemma~\ref{lem:BadIkHasSmallMeasure}. Let $\mathcal{Q}\subset \NN$ 
denote the set of integers $q$ for which $I_q$ is such an interval. 
We now cover $(3R, L-3R)$ by a collection of disjoint intervals 
$J \subset (3R, L-2R)$, each of length $|J|=R/2$.
Specifically, we take the intervals
$\bigl(3R+\frac{(n-1)R}{2}, 3R+\frac{nR}{2}\bigr)$ 
where $n$ runs from $1$ to $\bigl\lfloor\frac{2(L-5R)}{R}\bigr\rfloor$.
On each such $J=(r_1, r_2)$ we then have that
\begin{align}
	\int_J \frac{\hat\rho}{r}\, dr &\ge \frac{1}{r_2}\int_J \hat\rho \, dr\ge 
	\frac{1}{r_2}\int_{J\cap (\cup_{q\in \mathcal{Q}} I_q)^c} \hat\rho \, dr \ge \frac{2\kappa}{1+2\kappa}\frac{\alpha_N^2 C^4}{r_2 24\pi^2}\Bigl|J \cap \Bigl(\bigcup_{q\in \mathcal{Q}} I_q\Bigr)^c\Bigr|.
\end{align}
By Lemma~\ref{lem:BadIkHasSmallMeasure} we then obtain for the weighted mean of 
$\hat\rho$ that
\begin{align}
	\int_J \frac{\hat\rho}{r}\, dr \Bigm/ \int_J \frac{dr}{r} &\ge \frac{r_1}{r_2}\frac{2\kappa}{1+2\kappa}\frac{\alpha_N^2 C^4}{12\pi^2} \Bigl(\frac{1}{2}- \frac{4C(\pi - C^2)}{\pi-2C^2}\Bigr), 
	\quad \text{with} \quad \frac{r_1}{r_2} \ge \frac{6}{7}.
\end{align}
Thus for each $J$ we can again apply Lemma~\ref{lem:1Dprojection} and obtain
\begin{equation}
	\int_J\Bigl(\frac{\kappa}{2}|u'|^2+ \frac{\hat \rho}{r^2}|u|^2\Bigr)r\, dr 
	\ge
	\Bigl(\frac{2\kappa}{1+2\kappa}\Bigr)^2 \frac{C^4}{14\pi^2}\Bigl(\frac{1}{2}-\frac{4C(\pi - C^2)}{\pi-2C^2}\Bigr)\int_J\frac{\alpha_N^2}{r}|u|^2\, dr. 
\end{equation}
Applying this for each $J$ we obtain the desired estimate with
\begin{equation}
	c(\kappa)^2 = \Bigl(\frac{2\kappa}{1+2\kappa}\Bigr)^2 \frac{C^4}{14\pi^2}\Bigl(\frac{1}{2}-\frac{4C(\pi - C^2)}{\pi-2C^2}\Bigr).
\end{equation}
Maximizing this in $C\in (0, 1)$ we obtain for $C\approx 0.0996$ the extremely 
small (but positive) constant
\begin{equation}
	c(\kappa) \ge 5.3 \cdot 10^{-4} \frac{\kappa}{1+2\kappa}.
\end{equation}
This concludes the proof of Theorem~\ref{thm:MainRadialBound}
and hence the treatment of the long-range interaction of
Theorem~\ref{thm:long-range}.

We note that with this choice of $C$ we allow for approximately $80\%$ of any 
(and all) $R/2$ long interval contained in $(R, L]$ to be covered by the 
intervals $I_q$ satisfying~\eqref{eq:AssumptionsOnRemainingZeros}. 
As we expect that this is rather far from the actual situation for most 
particle configurations there seems to be room for improvement in the above 
considerations. One such improvement could be to use that the 
effective potential must between every two $I_q$ intervals go up to one and 
then back down again. Our current method does not take this into account and 
is blind to the fact that  
there must be helpful gaps between the $I_q$'s.

Another way of improving this constant would be to refine the bounds in 
Lemma~\ref{lem:GoodIntervals} by using the precise shape of the 
one-particle profile instead of the simpler lower bound provided by $f_\wedge$. 
One could also take into account that all intervals cannot be at the edge of a 
particle, i.e.\ make use of the observation that a large number of the particles 
are likely to cover more than one interval $I_q$.


\section{Local exclusion}\label{sec:exclusion}

We now formulate the obtained energy bounds for anyons in terms of local exclusion
principles, following~\cite{LunSol-13a, LunSol-13b, LunSol-14, LunPorSol-15, LunNamPor-16},
with some refinements to take both the short- and the long-range magnetic interactions
into account.

With a weight partition $\bkappa=(\kappa_1, \kappa_2, \kappa_3) \in [0, 1]^3$, 
$\kappa_1+\kappa_2+\kappa_3=1$, 
we can write for the total kinetic energy 
for $N$ anyons in a normalized state $\Psi \in \domD{\alpha, R}$
\begin{align} \label{eq:T-split}
	\langle \Psi, \hat{T}_\alpha \Psi \rangle
	&= \kappa_1 \sum_{j=1}^N \int_{\R^{2N}} |D_j\Psi|^2\, d\sx
	+ \kappa_2 \sum_{j=1}^N \int_{\R^{2N}} |D_j\Psi|^2\, d\sx
	+ \kappa_3 \sum_{j=1}^N \int_{\R^{2N}} |D_j\Psi|^2\, d\sx \\
	&\ge \int_{\R^{2N}} \sum_{j=1}^N \biggl(
		\kappa_1 \bigl|\nabla_j|\Psi|\bigr|^2
		+ \kappa_2 \sum_{\substack{k=1\\k \neq j}}^N 2\pi|\alpha| 
		\frac{\1_{B_R(0)}}{\pi R^2}(\bx_j - \bx_k) \, |\Psi|^2
		+ \kappa_3 |D_j\Psi|^2
		\biggr) d\sx, 
\end{align}
where we used Lemma~\ref{lem:diamagnetic} and Lemma~\ref{lem:short-range}.
We then make a partitioning of the plane $\R^2$ into disjoint squares 
$Q$'s:
\begin{equation}
	\langle \Psi, \hat{T}_\alpha \Psi \rangle
	\ge \sum_Q T^\bkappa_Q[\Psi], 
\end{equation}
where the expected local energy on each square $Q$ is given by
(the definitions extend to all $\bkappa \in \R^3$)
\begin{align} \label{eq:localEnergyT}
	T^\bkappa_Q[\Psi] &:= \sum_{j=1}^N \int_{\R^{2N}} \biggl(
		\kappa_1 \bigl|\nabla_j|\Psi|\bigr|^2
		+ \kappa_2 \sum_{\substack{k=1\\k \neq j}}^N 2\pi|\alpha| 
		\frac{\1_{B_R(0)}}{\pi R^2}(\bx_j - \bx_k) \, |\Psi|^2
		+ \kappa_3 |D_j\Psi|^2
		\biggr) \1_Q(\bx_j) \, d\sx \\[-20pt]
	&\phantom{:}\ge \sum_{n=0}^N E^\bkappa_n(|Q|) p_n(\Psi;Q).
\end{align}
Here the local $n$-particle energy 
(translation invariant and with Neumann b.c.) is given by 
\begin{equation} \label{eq:localEnergyE}
	E^\bkappa_n(|Q|) := \inf_{\int_{Q^n}|\psi|^2=1}  
		\sum_{j=1}^n \int_{Q^n} \biggl(
		\kappa_1 \bigl|\nabla_j|\psi|\bigr|^2
		+ \kappa_2 \sum_{\substack{k=1\\k \neq j}}^n 2\pi|\alpha| 
		\frac{\1_{B_R(0)}}{\pi R^2}(\bx_j - \bx_k) \, |\psi|^2
		+ \kappa_3 |D_j\psi|^2
		\biggr) d\sx, 
\end{equation}
and $p_n(\Psi;Q)$ denotes the $n$-particle probability distribution induced from $\Psi$, 
\begin{equation}
	p_n(\Psi;Q) := \sum_{A \subseteq \{1, \ldots, N\}, |A|=n} 
		\int_{(Q^c)^{N-n}} \int_{Q^n} |\Psi|^2
		\prod_{k \in A} d\bx_k \prod_{l \notin A} d\bx_l, 
\end{equation}
having the normalizations $\sum_{n=0}^N p_n(\Psi;Q) = 1$ 
and $\sum_{n=0}^N np_n(\Psi;Q) = \int_Q \varrho_\Psi$, 
the expected number of particles on $Q$.
In~\eqref{eq:localEnergyE} the operators $D_j$ still depend on all $N$ particles, 
with the first $n$ on $Q$, 
and we take the infimum over the remaining $N-n$ positions in $\R^2 \setminus Q$.

The inequality \eqref{eq:localEnergyT} is obtained by simply partitioning 
the configuration space $\R^{2N}$, e.g.\ by inserting into the integrand the partition
of unity $\1 = \prod_{k=1}^N(\1_Q(\bx_k) + \1_{Q^c}(\bx_k))$ and expanding.
This approach to bound the energy goes all the way back to Dyson and Lenard~\cite{DysLen-67}.

\subsection{Short-range exclusion}\label{sec:exclusionSR}

We consider first the contribution to the local energy coming solely from the 
short-range part of the magnetic interaction.

\begin{lemma}[Local exclusion --- short range] \label{lem:LocalExclusionSR}
	For any $\alpha \in \R$, $R > 0$ and $Q \subseteq \R^2$ a square, 
	and with $\gamma(Q) := R|Q|^{-\frac{1}{2}}$, we have that
	\begin{equation}
		E_n^{(1, 1, 0)}(|Q|)
		\ge \frac{\eSR(\alpha, \gamma(Q), n)}{|Q|} (n-1)_+, 
	\end{equation}
	and 
	\begin{equation}
		T^{(1, 1, 0)}_Q[\Psi] \ge 
			\frac{\eSR(\alpha, \gamma(Q), \int_Q\varrho_\Psi)}{|Q|} 
			\biggl( \int_Q\varrho_\Psi \ -1 \biggr)_{\!\!+}, 
	\end{equation}
	where 
	\begin{equation}
		\eSR(\alpha, \gamma, n) := \left\{ \begin{array}{ll} \displaystyle
			\frac{ |\alpha|\min\bigl\{(1-\gamma^2/2)_+^{-1}, K_\alpha/2\bigr\} }{ 
			K_\alpha 
			+ 2|\alpha|\bigl( -\ln (\gamma/{\sqrt{2}}) \bigr)_+ }
			& \text{for $\gamma < \sqrt{2}$, } \\[1cm]
			2|\alpha| \gamma^{-2} n
			& \text{for $\gamma \ge \sqrt{2}$. }
			\end{array}\right.
	\end{equation}
	Here
	\begin{equation}
		K_\alpha := \sqrt{2|\alpha|} \frac{I_0(\sqrt{2|\alpha|})}{I_1(\sqrt{2|\alpha|})}
		 \ge 2, \qquad K_0:=2,
	\end{equation}
	and $I_\nu$ denotes the modified Bessel function of order $\nu$.
\end{lemma}

\begin{proof}[Proof of Lemma~\ref{lem:LocalExclusionSR}]
	We consider the local energy form in \eqref{eq:localEnergyE}.
	In the case that $\gamma(Q) \ge \sqrt{2}$, the short-range potential 
	in the second term covers the full domain $Q$ for every particle, and hence
	$$
		E^{(1, 1, 0)}_n(|Q|) \ge \frac{2\pi|\alpha|}{\pi R^2} n(n-1)_+
		= \frac{2|\alpha|}{|Q|} \gamma(Q)^{-2} n(n-1)_+.
	$$
	By convexity we then also have that
	$$
		\sum_{n=0}^N E^{(1, 1, 0)}_n(|Q|)p_n(\Psi;Q) 
		\ge \frac{2|\alpha|}{|Q|} \gamma(Q)^{-2} 
			\biggl({\int_Q\varrho_\Psi}\biggr)\biggl({\int_Q\varrho_\Psi} \ -1\biggr)_{\!\!+}.
	$$
	In the case that $\gamma(Q) < \sqrt{2}$, we use Dyson's 
	lemma~\cite{Dyson-57} in two dimensions
	(see~\cite{LieYng-01, LieSeiSolYng-05, LunPorSol-15})
	to smear the potential to the full domain as done 
	in~\cite[Proposition~19]{LunPorSol-15}, 
	keeping part of the potential intact and smearing the rest.
	For $n > 1$ and any $\kappa \in [0, 1]$ we can bound the energy form in 
	$E^{(1, 1, 0)}_n(Q)$ from below by
	\begin{align}
		&n \int_{Q^2} \biggl(
			(1-\kappa) \Bigl( \bigl|\nabla_1|\psi|\bigr|^2
			+ \frac{2\pi|\alpha|}{\pi R^2} \1_{B_R(\bx_2)}(\bx_1) \, |\psi|^2 \Bigr)
			+ \kappa \frac{2\pi|\alpha|}{\pi R^2} \1_{B_R(\bx_2)}(\bx_1) \, |\psi|^2 
			\biggr) d\sx \\
		&\ge (n-1)_+ \int_{Q^2} \biggl(
			(1-\kappa) U(|\bx_1-\bx_2|) \1_{B_R(\bx_2)^c}(\bx_1) 
			+ \kappa \frac{2\pi|\alpha|}{\pi R^2} \1_{B_R(\bx_2)}(\bx_1) 
			\biggr) |\psi|^2 \, d\sx, 
	\end{align}
	with
	$$
		U(r) := |Q|^{-1} \biggl(1 - \frac{R^2}{2|Q|}\biggr)^{-1} \biggl( 
			\frac{K_\alpha}{2|\alpha|}
			+ \ln \frac{\sqrt{2}|Q|^{1/2}}{R} \biggr)^{-1} 
			\1_{[R, \sqrt{2}|Q|^{1/2}]}(r).
	$$
	This expression arises from the application of Dyson's 
	lemma~\cite[Lemma~3.1]{LieSeiSolYng-05}	on the star-shaped
	domain $Q-\bx_2$ with the requirement that
	$$
		\int_R^{\sqrt{2}|Q|^{1/2}} U(r) \ln (r/a_R) \, rdr \le 1, 
		\qquad U(r) = 0 \ \text{for} \ r < R, 
	$$
	and where the considered pair potential is
	$$
		W(\bx) := \frac{W_0}{R^2} \1_{B_R(0)}(\bx), 
		\qquad W_0 = 4|\alpha|, 
	$$
	with scattering length (see e.g.~\cite[Appendix~A.2.4]{LunPorSol-15})
	\begin{equation} \label{eq:soft-disk-scattering-length}
		a_R = R\exp\biggl( -\frac{1}{\sqrt{W_0/2}} \frac{I_0(\sqrt{W_0/2})}{I_1(\sqrt{W_0/2})} \biggr)
		= R\exp\biggl( -\frac{K_\alpha}{2|\alpha|} \biggr).
	\end{equation}
	We now demand that $\kappa$ be chosen such that the potentials match:
	$$
		(1-\kappa) U(r) = \kappa \frac{2|\alpha|}{R^2}, 
	$$
	that is, 
	$$
		\frac{\kappa}{1-\kappa} = \gamma(Q)^2 \bigl(1-\gamma(Q)^2/2\bigr)^{-1} \bigl( 
			K_\alpha 
			+ 2|\alpha|\bigl(-\ln (\gamma(Q)/{\sqrt{2}})\bigr) \bigr)^{-1}.
	$$
	However, note that the factor $(1-\gamma(Q)^2/2)^{-1}$ in $U$ diverges as 
	$\gamma(Q) \to \sqrt{2}$ while the other potential term stays bounded, 
	implying $\kappa \to 1$. Hence, 
	in order to be able to bound $1-\kappa$ uniformly we instead truncate the
	potential $U$ by replacing the unbounded factor with
	$$
		\min\bigl\{ \bigl(1-\gamma(Q)^2/2\bigr)^{-1}, K_\alpha/2 \bigr\} 
		\in [1, K_\alpha/2], 
	$$
	also using that $K_\alpha \ge 2$ (see~\cite[Eqn.~10.33.1]{OlvMax-10}).
	With this replacement in the above we then find that
	$$
		\frac{\kappa}{1-\kappa} 
		= \gamma(Q)^2 \frac{ \min\bigl\{ \bigl(1-\gamma(Q)^2/2\bigr)^{-1}, K_\alpha/2 \bigr\} }{
			K_\alpha + 2|\alpha|(-\ln (\gamma(Q)/{\sqrt{2}})) }
		\le \frac{\gamma(Q)^2}{2} \le 1, 
	$$
	and hence $\kappa \le 1/2$ and $1-\kappa \ge 1/2$.
	Summing up, we find for all $n \ge 0$ that
	$$
		E_n(Q) \ge \frac{(n-1)_+}{|Q|} (1-\kappa) 
			2|\alpha| \frac{ \min\bigl\{ \bigl(1-\gamma(Q)^2/2\bigr)_+^{-1}, K_\alpha/2 \bigr\} }{
			K_\alpha + 2|\alpha|\bigl(-\ln(\gamma(Q)/{\sqrt{2}})\bigr)_+ }, 
	$$
	and may again use convexity in $n$ to obtain the corresponding bound for
	$T_Q[\Psi]$.
\end{proof}


Although not aiming to provide the sharpest possible bound, the above lemma
has the advantage of being relatively simple and it captures the overall 
dependence of the pure short-range interaction on the parameters.
In a certain regime however, referred to below as the \emph{soft-core regime}, 
the following version 
(which could in some sense be viewed as a mix between the two and three-dimensional cases 
studied in~\cite{LieYng-98, LieYng-00, LieYng-01, LieSeiSolYng-05})
will yield a comparatively good bound.

\begin{lemma}[Soft-core exclusion] \label{lem:TempleSmearBound}
	For any $R \ge 0$ and $Q \subseteq \R^2$ a square, 
	and with $\gamma(Q) := R|Q|^{-\frac{1}{2}}$, we have that
	\begin{multline}
		E_n^{(\kappa, 1-\kappa, 0)}(Q)
		\ge 2\pi|\alpha|(1-\kappa)\bigl(1-2\gamma(Q)\bigr)_+^2 \frac{n(n-1)}{|Q|} \biggl(
			1 - \frac{2|\alpha| \gamma(Q)^{-2} n(n-1)}{\pi^2\kappa/(1-\kappa) - 2\pi|\alpha|n(n-1)}
			\biggr)_{\!\!+}, 
	\end{multline}
	for any $\kappa \in (0, 1)$, $\alpha \in \R$ and $n \ge 2$ such that
	$\pi^2\kappa/(1-\kappa) > 2\pi|\alpha|n(n-1)$.
\end{lemma}

\begin{proof}
	Following~\cite{LieSeiSolYng-05} we write for the operator of the left-hand side
	$$
		H = \kappa\sum_{j=1}^n (-\Delta_{\bx_j}) + (1-\kappa)W, 
	$$
	with (assuming $\alpha>0$ for notational simplicity)
	$$
		W = 2\pi\alpha \sum_{j \neq k} \frac{\1_{B_R(0)}}{\pi R^2}(\bx_j-\bx_k). 
	$$
	We apply the following result due to Temple~\cite{Temple-28, LieSeiSolYng-05}: 
	If $H=H_0+V$, for some Schr\"odinger operator $H_0\ge 0$ and scalar potential 
	$V\ge 0$, then the ground-state energy of $H$ is bounded from below by
	\begin{equation}
		\lambda_0(H_0)+\langle V \rangle_{\psi_0}- \frac{ \langle V^2 
		\rangle_{\psi_0} - \langle V \rangle_{\psi_0}^2}{\lambda_1(H_0)-\langle V \rangle_{\psi_0}}, 
	\end{equation}
	as long as $\lambda_1(H_0)-\langle V \rangle_{\psi_0}$ is positive. 
	Here $\psi_0$ denotes the normalized ground state of $H_0$, 
	$\langle V \rangle_{\psi_0} := \int V |\psi_0|^2$ 
	is the expectation of $V$ in the state $\psi_0$, 
	and $\lambda_0(H_0)$ resp.\ $\lambda_1(H_0)$ is the first resp.\ second eigenvalue of $H_0$.

	In our case, 
	$H_0 = -\kappa\Delta_{Q^n}^\mathcal{N}$ (the Neumann Laplacian)
	and $\psi_0 \equiv |Q|^{-n/2}$, we have that 
	\begin{align}
	 	2\pi\alpha \frac{n(n-1)}{|Q|} \ge \langle W \rangle_{\psi_0}
	 	\ge 2\pi\alpha (1-2\gamma(Q))^2 \frac{n(n-1)}{|Q|}, 
	\end{align}
	where for the lower bound one integrates the first particle of each pair 
	on a smaller domain with margin $R$ away from the boundary. 
	Moreover, by Cauchy--Schwarz
	\begin{align}
		\langle W^2 \rangle_{\psi_0} 
		\le \frac{2\alpha}{R^2} n(n-1) \langle W \rangle_{\psi_0}.
	\end{align}
	Thus Temple's inequality yields that
	\begin{align}
		H &\ge \langle(1-\kappa)W\rangle_{\psi_0} 
			- \frac{\langle(1-\kappa)^2 W^2\rangle_{\psi_0} - \langle(1-\kappa)W\rangle_{\psi_0}^2}{
				\lambda_1(\kappa\sum_j(-\Delta_j)) - \langle(1-\kappa)W\rangle_{\psi_0}} \\
		&\ge (1-\kappa)\langle W \rangle_{\psi_0} \biggl( 1 - 
			\frac{(1-\kappa)2\alpha R^{-2} n(n-1)}{\kappa \pi^2/|Q| - (1-\kappa)\langle W \rangle_{\psi_0}}
			\biggr) \\
		&\ge 2\pi\alpha(1-\kappa)(1-2\gamma(Q))^2 \frac{n(n-1)}{|Q|} \biggl( 1 - 
			\frac{2\alpha \gamma(Q)^{-2} n(n-1)}{\pi^2\kappa/(1-\kappa) - 2\pi\alpha n(n-1)}
			\biggr), 
	\end{align}
	as claimed.
\end{proof}


\subsection{Long-range exclusion}\label{sec:exclusionLR}

We now turn to local energy bounds for the pure long-range part of the magnetic
interaction.

\begin{lemma}[Local exclusion --- long range] \label{lem:LocalExclusionLR}
	For any $\alpha \in \R$, $R \ge 0$ and $Q \subseteq \R^2$ a square, 
	and with $\gamma(Q) := R|Q|^{-\frac{1}{2}}$, we have that
	\begin{equation}
		E_n^{(0, 0, 1)}(Q)
		\ge \frac{\eLR(\alpha, \gamma(Q))}{|Q|} (n-1)_+, 
	\end{equation}
	and 
	\begin{equation}
		T^{(0, 0, 1)}_Q[\Psi] \ge 
			\frac{\eLR(\alpha, \gamma(Q))}{|Q|} \biggl( \int_Q\varrho_\Psi \ -1 \biggr)_{\!\!+}, 
	\end{equation}
	with
	\begin{equation} \label{eq:eLRextended}
		\eLR(\alpha, \gamma) := 
		\frac{\pi}{24} g\bigl(c\alpha_N, 12\gamma\bigr)^2 (1-12\gamma)_+^3, 
	\end{equation}
	where $c = 5.3/{\sqrt{8}} \cdot 10^{-4}$.
	
	For $R=0$, the above bounds are valid with $\eLR(\alpha, 0) = f((j_{\alpha_N}')^2)$ 
	for all $\alpha \in \R$, where $f\colon [0, (j_1')^2] \to \R$ is a function 
	defined below satisfying
	\begin{equation} \label{eq:eLRideal}
		t/6 \le f(t) \le 2\pi t 
		\quad \text{and} \quad
		f(t) = 2\pi t \bigl(1 - O(t^{1/3})\bigr)
	\end{equation}
	(see Figure~\ref{fig:IdealPlots} for both lower and upper bounds for $f$).
\end{lemma}

\begin{figure}[ht]
	\centering
	\begin{tikzpicture}
		\node [above right] at (0,0) {\includegraphics[scale=0.77,clip,
		trim = 0pt 0.4pt 0pt 0pt]{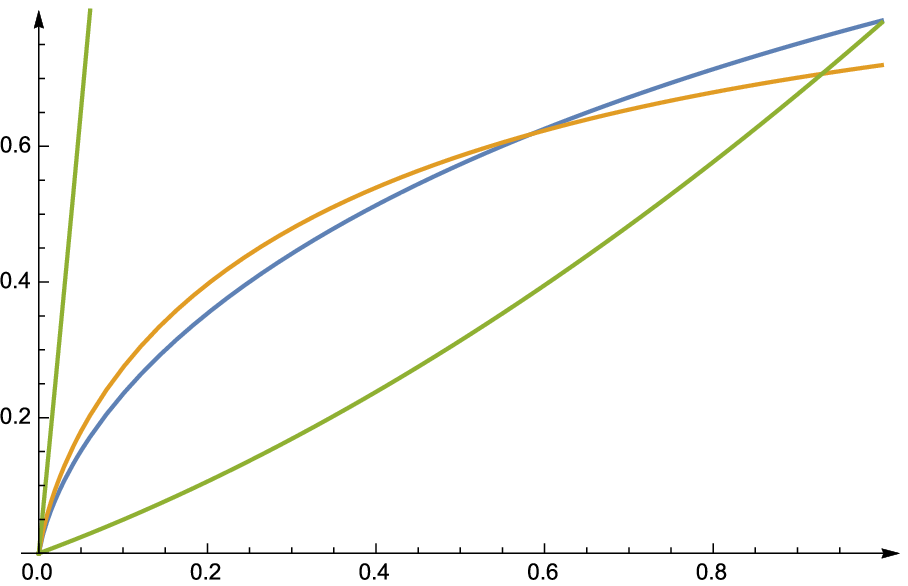}};
		\node [above right] at (7.12,.18) {\scalebox{0.8}{$\nu$\hspace{-2pt}}};
	\end{tikzpicture}
	\begin{tikzpicture}
		\node [above right] at (0,0) {\includegraphics[scale=0.79]{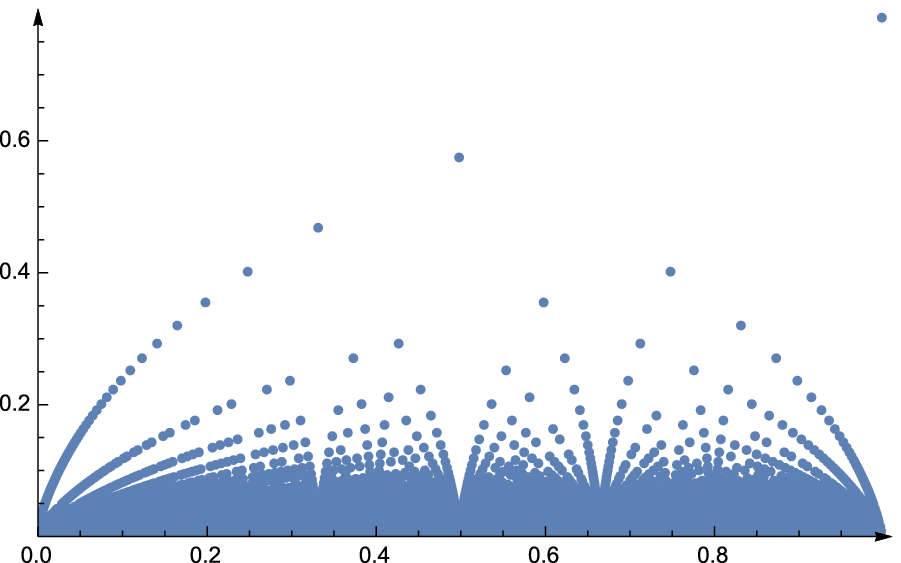}};
		\node [above right] at (7.23,.18) {\scalebox{0.8}{$\alpha$\hspace{-2pt}}};
	\end{tikzpicture}
	\caption{Left: A comparison between the optimized energy bounds for 
	$f((j_\nu')^2)$ on the unit square as a function of $\nu \in [0, 1]$, 
	obtained by means of the projection method (blue) and 
	Temple (yellow), as well as the upper and lower bounds given in 
	(\ref{eq:eLRideal})
	(green). 
	Right: A numerical lower bound to the energy 
	$\eLR(\alpha, 0) = f((j_{\alpha_*}')^2)$ 
	on the unit square as a function of $\alpha$. 
	The bound uses the projection method and the erratic behavior is due to 
	the function $\alpha \mapsto \alpha_*$ being discontinuous at all 
	odd-numerator rationals.}
	\label{fig:IdealPlots}
\end{figure} 

The tiny constant $c$ stems from Theorem~\ref{thm:long-range} and again we 
expect that it could be replaced with $c=1/{\sqrt{3}}$ or just slightly smaller
(recall Remark~\ref{rem:MainRadialBound}).
Accordingly we have not aimed for the sharpest possible bounds in our proof for $R>0$.
Note however that for $R=0$ and in the limit $\alpha \to 0$,
the two-particle energy per particle  
is exactly the expected one from average-field theory, 
$\pi (j_{\alpha_*}')^2 \sim 2\pi \alpha_* \sim 2\pi|\alpha|$ for suitable $\alpha$,
however the bound is only linear (and not quadratic) in
$n$ and hence only good for small enough boxes $Q$, resulting in a worse 
constant (by a factor $1/2$) when applied below in the thermodynamic limit.
Also note that the bounds involve $\alpha_N$ and not $\alpha_n$
or $\alpha_{\lceil \int_Q \varrho_\Psi \rceil}$ because there is a probability that
more particles (in fact all the way up to $N$) can be found on $Q$.

\begin{proof}[Proof of Lemma~\ref{lem:LocalExclusionLR}]
\mbox{} 

{\bf Ideal case.} 
We begin with the more transparent case of $R=0$, 
and note that we may set $|Q|=1$ by scaling. 
Our starting point is the long-range magnetic interaction bound provided by 
Theorem~\ref{thm:long-range}. For the ideal case the 
theorem states that
\begin{equation}
	\sum_{j=1}^n \int_{Q^n} |D_j\Psi|^2\, d\sx \ge \frac{1}{n}\sum_{j<k}\int_{Q^n}(j'_{\alpha_N})^2 
	\frac{\1_{B_{\delta(\bX_{jk})}}(\br_{jk})}{\delta(\bX_{jk})^2}|\Psi|^2\, d\sx.
\end{equation}
In order to convert this non-uniform potential to a uniform bound for the energy
we take part of the kinetic energy and then apply 
either Temple's inequality as in Lemma~\ref{lem:TempleSmearBound} 
or a projection argument as in~\cite[Lemma~7]{LunSol-13a}
or Lemma~\ref{lem:1Dprojection}.
To this end we take a fraction $\kappa\in [0, 1]$ of the original kinetic
energy for which we use the diamagnetic inequality and the identity
\begin{equation}
	\sum_{j=1}^n |\bz_j|^2 = \frac{1}{n-1}\sum_{j<k}(|\bz_j|^2+|\bz_k|^2), 
	\qquad \bz_j\in \C, 
\end{equation}
and on the remaining fraction $1-\kappa$
we use Theorem~\ref{thm:long-range}.
We then obtain that\footnote{It also turns out that we do not gain much by keeping the
$n$-dependence in the first term if we are aiming for a bound which is 
convex in $n$.} 
\begin{align}\label{eq:Diamagnetic+Longrange}
	&\sum_{j=1}^n \int_{Q^n} |D_j\Psi|^2\, d\sx \\
	&\ge \frac{1}{n}\sum_{j<k}\int_{Q^n}\biggl(\frac{\kappa n}{n-1}
		\bigl(\bigl|{\nabla_{j}|\Psi|}\bigr|^2 + \bigl|{\nabla_{k}|\Psi|}\bigr|^2\bigr)
		+ (1-\kappa) (j'_{\alpha_N})^2 \frac{\1_{B_{\delta(\bX_{jk})}}(\br_{jk})}{\delta(\bX_{jk})^2}|\Psi|^2\biggr)d\sx\\
	&\ge \frac{1}{n}\sum_{j<k}\int_{Q^{n-2}}\int_{Q^2}\biggl(\kappa\bigl(\bigl|{\nabla_{j}|\Psi|}\bigr|^2+\bigl|{\nabla_{k}|\Psi|}\bigr|^2\bigr)+ 
		(1-\kappa)(j'_{\alpha_N})^2 \frac{\1_{B_{\delta(\bX_{jk})}}(\br_{jk})}{\delta(\bX_{jk})^2}|\Psi|^2\biggr)d\bx_jd\bx_k d\sx' \\[-7pt]
	&\ge (n-1)_+ \, \eLR(\alpha, 0), 
\end{align}
where $\eLR(\alpha, 0) := f((j^{\prime}_{\alpha_N})^2)$ and
\begin{equation}\label{eq:def_of_ideal_pair_f}
	f(t) := \frac{1}{2} \sup_{\kappa \in (0, 1)} \inf_{\int_{Q^2}|\psi|^2=1} \int_{Q^2}
		\biggl(\kappa\bigl(\bigl|{\nabla_1|\psi|}\bigr|^2 
			+ \bigl|{\nabla_2|\psi|}\bigr|^2\bigr) 
		+ (1-\kappa) t \frac{\1_{B_{\delta(\bX)}}(\br)}{\delta(\bX)^2}|\psi|^2\biggr) d\bx_1 d\bx_2.
\end{equation}
We then use the convexity in $n$ to obtain the corresponding bound for 
$T_Q[\Psi]$ in terms of $\eLR(\alpha, 0)$.
The upper bound $f(t)\le 2\pi t$ is found simply by taking the trial state 
$\psi = \psi_0 \equiv 1$ and then $\kappa=0$, carrying out the integration as below
(with $\hat\delta = 0$).

We now wish to find a lower bound for the integral in $f(t)$, 
which then is to be maximized in $\kappa$.
This is equivalent to finding a lower bound for the ground-state energy of the 
Schr\"odinger operator 
\begin{equation}
	H:= -\kappa\Delta_{Q^2}^\mathcal{N}+t(1-\kappa)V, 
	\qquad V(\bx_1, \bx_2):=V(\br, \bX)=\frac{\1_{B_{\delta(\bX)}}(\br)}{\delta(\bX)^2}.
\end{equation}
However, to apply a projection bound or use Temple's inequality requires that 
${V\in L^\infty(Q^2)}$ and $V\in L^2(Q^2)$, respectively. 
As neither of these conditions are satisfied for our $V$ we use the fact that 
$V\ge0$ and thus truncating our potential will only lower the energy. 
Therefore we instead study the eigenvalue problem with $V$ replaced by the 
truncated potential $\hat{V}$ defined in relative coordinates by
\begin{equation}
	\hat{V}(\br, \bX):= \left\{\begin{matrix}
		\displaystyle{\frac{\1_{B_{\delta(\bX)}}(\br)}{\delta(\bX)^2}}, & \delta(\bX)\ge \hat\delta\\[9pt]
		\displaystyle{\frac{\1_{B_{\delta(\bX)}}(\br)}{\hat\delta^2}}, & \delta(\bX)< \hat\delta
	\end{matrix}\right.
\end{equation}
(in slightly more compact notation, 
$\hat{V}=\min\{V, 1/\hat\delta^{2}\}$). 
As $\hat{V} \in L^\infty(Q^2)$, $\|\hat{V}\|_{\infty}=1/\hat\delta^{2}$, 
it follows that also $\hat{V}\in L^2(Q^2)$. 

We proceed by calculating the expectation of $\hat{V}$ and $\hat{V}^2$ 
in the ground state $\psi_0 \equiv 1$ of $-\Delta^\mathcal{N}_{Q^2}$, 
as needed for the bounds. Through a straightforward calculation one finds that
\begin{align}
	\langle \hat{V} \rangle_{\psi_0}
	&= 4\int_{Q}\int_{Q_\bX} \hat{V}(\br, \bX)\, d\br d\bX\\
	&= 4\biggl(\int_{[\hat\delta, 1-\hat\delta]^2}\int_{Q_\bX}\frac{1}{\delta(\bX)^2}\, d\br d\bX+\int_{Q\setminus[\hat\delta, 1-\hat\delta]^2}\int_{Q_\bX}\frac{1}{\hat\delta^2}\, d\br d\bX\biggr)\\
	&= 4\pi \Bigl(1+2\hat\delta^2- \frac{8\hat\delta}{3}\Bigr), 
\end{align}
and correspondingly for $\hat{V}^2$ we obtain
\begin{align}
	\langle \hat{V}^2 \rangle_{\psi_0}
	&= 4\int_{Q}\int_{Q_\bX} \hat{V}(\br, \bX)^2\, d\br d\bX\\
	&= 4\biggl(\int_{[\hat\delta, 1-\hat\delta]^2}\int_{Q_\bX}\frac{1}{\delta(\bX)^4}\, d\br d\bX+\int_{Q\setminus[\hat\delta, 1-\hat\delta]^2}\int_{Q_\bX}\frac{1}{\hat\delta^4}\, d\br d\bX\biggr)\\
	&= 8\pi \Bigl(\frac{8}{3\hat\delta}+4\ln(2\hat\delta)-5\Bigr).
\end{align}
Choosing $\hat\delta=\eta/2$ for some $\eta \in [0, 1]$ 
(this normalization is convenient) results in
\begin{align}
	\langle \hat{V} \rangle_{\psi_0} = 4\pi\Bigl(1+\frac{\eta^2}{2}-\frac{4\eta}{3}\Bigr), \quad \mbox{and}\quad
	\langle \hat{V}^2 \rangle_{\psi_0} = 8\pi \Bigl(\frac{16}{3}\eta^{-1}+4\ln\eta-5\Bigr).
\end{align}
Our considerations here have been for $\Omega=Q$ the unit square but also 
other domains $\Omega$ could be of interest.
Similar calculations when $\Omega$ is the unit disk and $\hat\delta = \eta$ 
give instead
\begin{align}
	\langle \hat{V} \rangle_{\psi_0} = 4 \Bigl(1+\frac{\eta^2}{2}-\frac{4\eta}{3}\Bigr), \quad\mbox{and}\quad
	\langle \hat{V}^2 \rangle_{\psi_0} = 2\Bigl(\frac{16}{3}\eta^{-1}+4\ln\eta-5\Bigr).
\end{align}

Let $P$ denote the orthogonal projection onto the ground state $\psi_0\equiv 1$, 
and let $P^\perp=1-P$. Then $(-\Delta_{Q^2}^\mathcal{N})P=0$, and with 
$\lambda_1(-\Delta_{Q^2}^\mathcal{N})$ the first non-zero Neumann eigenvalue, 
\begin{equation}
	(-\Delta_{Q^2}^\mathcal{N})P^\perp \ge  
	\lambda_1(-\Delta_{Q}^\mathcal{N})P^\perp = \pi^2 P^\perp.
\end{equation}
Arguing as in Lemma~\ref{lem:1Dprojection}, we for any $\mu\in (0, 1)$ obtain that 
\begin{equation}\label{eq:Proj_Split_Potential}
	\hat{V}  
	\ge (1-\mu)P\hat{V}P+(1-\mu^{-1})P^\perp \hat{V}P^\perp, 
\end{equation}
the first of these operators is equal to 
$\langle \hat{V} \rangle_{\psi_0}P$, 
and we can control the second term by using that 
$\|P^\perp \hat{V}P^\perp\| \le \|\hat{V}\|_{\infty}
=4/\eta^2$.

Thus, for any $\mu, \kappa, \eta \in (0, 1)$ we find that
\begin{align}
	H 
	&\ge 
	(1-\mu)4\pi t(1-\kappa)\Bigl(1+\frac{\eta^2}{2}-\frac{4\eta}{3}\Bigr) P + \biggl(\kappa\pi^2+(1-\mu^{-1})\frac{4t(1-\kappa)}{\eta^2}\biggr)P^\perp\\
	&\ge 
	\min \biggl\{(1-\mu)4\pi t (1-\kappa)\Bigl(1+\frac{\eta^2}{2}-\frac{4\eta}{3}\Bigr), 
	\kappa\pi^2+(1-\mu^{-1})\frac{4t(1-\kappa) }{\eta^2}\biggr\}(P+P^\perp).
\end{align}
The last expression, seen as a function in $t$, is piecewise linear and concave. 
Thus to obtain the largest linear minorant of this function it suffices to find 
the largest value attained at the right endpoint of our range of values $t$, 
that is at $t=(j'_1)^2\approx 3.8996$.

By the $\mu$ dependence of each of the two terms in the minimum this quantity 
is seen to be maximal when the two terms are equal. Solving this quadratic 
equation in $\mu$ and choosing $\eta=\kappa = 0.68$  
we find that
\begin{equation}\label{eq:Unif.bound ideal pair q-form}
	H \ge t/3 \quad \mbox{and hence}\quad f(t)\ge t/6.
\end{equation}

To obtain that $f(t)= 2\pi t(1-O(t^{1/3}))$
 we apply Temple's inequality
(as in Lemma~\ref{lem:TempleSmearBound}). In our current setting it yields that
\begin{align}
	H 
	&\ge 
	\langle t(1-\kappa)\hat{V} \rangle_{\psi_0}
		- \frac{\langle t^2(1-\kappa)^2\hat{V}^2 \rangle_{\psi_0} 
			- \langle t(1-\kappa) \hat{V} \rangle_{\psi_0}^2
			}{\kappa\lambda_1(-\Delta_{Q^2}^\mathcal{N})
			-\langle t(1-\kappa)\hat{V} \rangle_{\psi_0}}\\
			&= 
	4\pi t(1-\kappa)\Biggl(1+\frac{\eta^2}{2}-\frac{4\eta}{3}
	-
	\frac{2t(1-\kappa)}{\pi}\frac{
		\frac{16}{3}\eta^{-1}+4\ln\eta-5 -
		2\pi\bigl(1+\frac{\eta^2}{2}-\frac{4\eta}{3}\bigr)^2
	}
	{
	\kappa\pi- 4 t(1-\kappa)\bigl(1+\frac{\eta^2}{2}-\frac{4\eta}{3}\bigr)
	}\Biggr),
\end{align}
provided that 
$\kappa\pi- 4 t(1-\kappa)\bigl(1+\frac{\eta^2}{2}-\frac{4\eta}{3}\bigr)>0$.
We decrease the above quantity by throwing away positive terms and increasing 
the denominator of the last term yielding
\begin{align}
	H 
	&\ge
	  4\pi t (1-\kappa) \biggl(1-\frac{4\eta}{3}
	 -
	 \frac{32}{3\pi}\frac{
	 	(1-\kappa) t\eta^{-1}
	 }
	 {
	 \kappa\pi- 4 (1-\kappa) t
	 }\biggr).
\end{align}
The positivity of denominator is then ensured if $\kappa \ge \frac{4 t}{\pi}$. 
We can thus, for $t$ sufficiently small, choose $\kappa= t^\beta$ for some 
$0<\beta < 1$ to be fixed later. Inserting this into our expression we find that
\begin{align}
	H &\ge
	 4\pi t(1-t^\beta)  \biggl(1-\frac{4\eta}{3}
	-
	\frac{32}{3\pi}\frac{
		(1-t^\beta) t\eta^{-1}
	}
	{
	t^\beta\pi- 4(1-t^\beta) t
	}\biggr).
\end{align}
Setting $\eta= t^\gamma$, $\gamma>0$, we obtain that
\begin{equation}
	H \ge 4\pi t(1-O(t^\beta)-O(t^\gamma)-O(t^{1-\beta-\gamma})),
\end{equation}
and choosing $\beta=\gamma=1/3$ yields
\begin{equation}
	H \ge 4\pi t(1-O(t^{1/3})).
\end{equation}
Inserting this into~\eqref{eq:def_of_ideal_pair_f} we have
\begin{equation}
	f(t)=2\pi t(1-O(t^{1/3})), 
\end{equation}
which completes the proof.

\medskip

{\bf Extended case.} Let, in the case that $R \ge 0$, $\gamma$ denote
the relative length scale of the
interaction, $\gamma = \gamma(Q) = R|Q|^{-1/2}$, 
and note that we may again rescale everything so that $|Q|=1$.
We then proceed as above using projection, 
where the bound from Theorem~\ref{thm:long-range} is replaced by
\begin{equation}
	\sum_{j=1}^n \int_{Q^n} |D_j\Psi|^2\, d\sx 
	\ge (1-\kappa')\frac{1}{n}\sum_{j<k}\int_{Q^n} 
		g\Bigl(\nu, \frac{3\gamma}{\delta(\bX_{jk})
			-3\gamma}\Bigr)^2 
		\frac{\1_A(\bx_j, \bx_k)}{\delta(\bX_{jk})^2} |\Psi|^2 \, d\sx, 
\end{equation}
where $\nu = c(\kappa')\alpha_N/{\sqrt{1-\kappa'}}$ and 
$\kappa' \in (0, 1)$ is an additional parameter that we may optimize over, 
however we will in order to simplify the analysis take $\kappa' = 1/2$.
Since $\delta(\bX_{jk})$ maximally takes the value $1/2$, 
the above expression is zero for $\gamma \ge 1/12$. 
For $0 \le \gamma < 1/12$ we can proceed by truncating to 
the, in $\gamma$, uniformly bounded potential
\begin{equation}
	\hat{V}(\bX, \br) := \frac{1}{2}
		g\Bigl(\nu, \frac{3\gamma}{\delta(\bX)-3\gamma}\Bigr)^2 
		\frac{\1_{\hat{A}}(\bx_1, \bx_2)}{\delta(\bX)^2}, 
\end{equation}
with the support (consisting of truncated relative annuli)
\begin{equation}
	\hat{A} := \{ (\bx_1, \bx_2) \in Q^2
		: 3\gamma+1/4 \le \delta(\bX) \le 1/2 \ \text{and} \ 
		3\gamma \le |\br| \le \delta(\bX) - 3\gamma \}, 
\end{equation}
and therefore,
since $g(\nu, \gamma)$ 
is monotonically decreasing in $\gamma$,
\begin{equation}
	\|\hat{V}\|_\infty \le \frac{1}{2(3\gamma+1/4)^2} 
		g(\nu, 0)^2
		\le 8 (j_\nu')^2.
\end{equation}
Also, using the coarea formula and that $|\nabla \delta|=1$ almost everywhere, 
we obtain that
\begin{align}
	\langle\hat{V}\rangle_{\psi_0}
	&= \frac{1}{2} \int_Q  \int_{Q_{\bX}} 
		g\Bigl(\nu, \frac{3\gamma}{\delta(\bX)-3\gamma} \Bigr)^2 
		\frac{\1_{\hat{A}}(\bx_1, \bx_2)}{\delta(\bX)^2} \, 4d\br d\bX \\
	&= 2\pi \int_Q g\Bigl(\nu, \frac{3\gamma}{\delta(\bX)-3\gamma} \Bigr)^2 
		\frac{\bigl((\delta(\bX)-3\gamma)^2 - (3\gamma)^2\bigr)_+}{\delta(\bX)^2} \, d\bX \\
	&= 8\pi \int_{3\gamma+1/4}^{1/2} g\Bigl(\nu, \frac{3\gamma}{t-3\gamma} \Bigr)^2 
		(1-6\gamma/t)(1-2t) \, dt \\
	&\ge \frac{\pi}{3} g(\nu, 12\gamma)^2 (1-12\gamma)^3,
\end{align}
where in the last step we again used the monotonicity of $g$, and
\begin{equation}
	\int_{3\gamma+1/4}^{1/2} (1-6\gamma/t)(1-2t) \, dt
	= \biggl( \frac{1}{16} 
		+ \Bigl(\frac{3}{2} - 6\ln\frac{2}{1+12\gamma} \Bigr)\gamma 
		- 27\gamma^2 \biggr)
	\ge \frac{1}{24} (1-12\gamma)^3, 
\end{equation}
where the lower bound is found by Taylor expansion around $\gamma = 1/12$.

Thus, the corresponding projection bound for the operator 
$H = -\kappa\Delta_{Q^2}^\mathcal{N} + (1-\kappa)\hat{V}$ reads
\begin{align}
	H \ge \min \Bigl\{
		(1-\mu)(1-\kappa) \frac{\pi}{3} g(\nu, 12\gamma)^2 (1-12\gamma)_+^3, 
		\ \kappa\pi^2 - (\mu^{-1}-1)8(1-\kappa) (j_\nu')^2
		\Bigr\}.
\end{align}
We take, for simplicity,
$\mu = 1/2$ and $\kappa = 1/2$, 
and use that $g(\nu, 12\gamma) \le j_\nu' \ll \pi$, 
to obtain the claimed bound
\begin{equation}
	\sum_{j=1}^n \int_{Q^n} |D_j\Psi|^2\, d\sx 
		\ge (n-1)_+ \, \eLR(\alpha, \gamma), 
	\quad
	\eLR(\alpha, \gamma) = \frac{\pi}{24} g(\nu, 12\gamma)^2 (1-12\gamma)_+^3, 
\end{equation}
with $\nu = c \alpha_N$ and
$c = c(\kappa')/{\sqrt{1-\kappa'}} = 5.3/{\sqrt{8}} \cdot 10^{-4}$.
Again we may use the convexity in $n$ to obtain the corresponding bound 
for $T_Q[\Psi]$.
\end{proof}


\section{Application to the homogeneous anyon gas}\label{sec:gas}

Let us finally consider the homogeneous gas in the thermodynamic limit, i.e.\ $N$
particles confined to a large box (square) $Q_0 \subseteq \R^2$, 
where we shall take simultaneously $N \to \infty$ and $|Q_0| \to \infty$ 
while keeping the density $\bar\varrho := N/|Q_0|$ fixed. 
The only dimensionless parameters are then the magnetic interaction strength
$\alpha \in \R$ and the relative interaction length scale
(magnetic filling ratio) $\bar\gamma := R\bar\varrho^\frac{1}{2}$, 
also held fixed, so that in the limit the ground-state energy, 
\begin{equation} \label{eq:gsEnergy}
	E_0(N, Q_0, \alpha, R) 
	:= \inf \Bigl\{ \langle\Psi, \hat{T}_\alpha \Psi\rangle \,:\, 
		\Psi \in \domD{\alpha, R} \cap C^\infty_c(Q_0^N) \,,\, \|\Psi\|_2 = 1 \Bigr\}, 
\end{equation}
per particle must for dimensional reasons be given by
\begin{equation} \label{eq:ThermodynLim}
	\frac{E_0(N, Q_0, \alpha, R)}{N} \to e(\alpha, \bar\gamma) \bar\varrho, 
\end{equation}
where $e(\alpha, \bar\gamma) \ge 0$ is dimensionless.
We have that $e(0, \bar\gamma) = 0$ for all $\bar\gamma \ge 0$, 
corresponding to non-interacting bosons, 
and $e(1, 0) = 2\pi$ for ideal fermions in two dimensions
due to the Weyl asymptotics for the Laplacian eigenvalues.
We also have a reflection-conjugation symmetry 
$e(-\alpha, \bar\gamma) = e(\alpha, \bar\gamma)$ for all $\alpha, \bar\gamma$.
Furthermore, in the dilute limit we should 
see a periodicity in the entire spectrum with respect to 
any shift in $\alpha$ by an even integer, and in particular
\begin{equation} \label{eq:ThermoPeriodicity}
	e(\alpha+2n, 0) = e(\alpha, 0) 
	\qquad \forall \ \alpha \in \R, \ n \in \Z, 
\end{equation}
due to the gauge equivalence~\eqref{eq:gauge-equivalence}.
On the other hand, average-field theory~\eqref{eq:average-field} 
suggests a linear dependence
$e(\alpha, \bar\gamma) = 2\pi|\alpha|$ for arbitrary $\alpha$
and large enough $\bar\gamma$.
Hence there must be some non-trivial interpolation between these 
two regimes of low respectively high density.

Although the existence of the thermodynamic limit~\eqref{eq:ThermodynLim}
might be expected on physical grounds, 
as is indeed the case for bosons and fermions with reasonable scalar 
interactions (see e.g.~\cite{CatBriLio-98, LieSei-09}), 
we are not aware of any proof of it for anyons, 
whose interaction is long-range and magnetic instead of scalar.
Furthermore, there is for anyons also a subtlety in the choice of boundary
conditions, partly since topology plays an important role in the whole problem
and therefore periodic b.c.\ may seem problematic, 
and even in the case of a constant magnetic field we know that
Neumann and Dirichlet b.c.\ differ substantially 
(cf.\ Section~\ref{sec:short-range} and Proposition~\ref{prop:constant-field}).
We shall therefore replace the limit~\eqref{eq:ThermodynLim} with the
$\liminf$ and also stick to Dirichlet b.c.\ (`hard-wall' confined anyons)
in all that follows.

\begin{theorem}[Universal bounds for the homogeneous anyon gas] \label{thm:gas}
	Let $e(\alpha, \bar\gamma)$, where $\bar\gamma=R\bar\varrho^{1/2}$, denote 
	the ground-state energy per particle and unit density of the extended anyon gas
	in the thermodynamic limit at fixed $\alpha \in \R$, $R \ge 0$ and density 
	$\bar\varrho > 0$ where Dirichlet boundary conditions have been imposed, 
	that is
	$$
		e(\alpha, \bar\gamma) 
		:= \liminf_{\substack{N, \, |Q_0| \to \infty \\ \! N/|Q_0| = \bar\varrho}} 
			\frac{E_0(N, Q_0, \alpha, R)}{\bar\varrho N}.
	$$
	Then
	\begin{multline} \label{eq:main-homo-universal}
		e(\alpha, \bar\gamma) \ge C\biggl(
			2\pi\frac{|\alpha| 
				\min\bigl\{2(1-\bar\gamma^2/4)^{-1}, K_\alpha\bigr\} }{ 
				K_\alpha + 2|\alpha|\bigl( -\ln (\bar\gamma/2) \bigr) } 
				\1_{\bar\gamma < 2}
			+ 2\pi|\alpha| \1_{\bar\gamma \ge 2} \\
			+ \pi g\bigl(c\alpha_*, 12\bar\gamma/{\sqrt{2}}\bigr)^2 (1-12\bar\gamma/{\sqrt{2}})_+^3
			\biggr), 
	\end{multline}
	for some universal constant $C > 0$, 
	with $K_\alpha$ given in Lemma~\ref{lem:LocalExclusionSR}, 
	and $c > 0$ in Lemma~\ref{lem:LocalExclusionLR}.
	Furthermore, for any $\alpha \in \R$
	and with $f$ given in Lemma~\ref{lem:LocalExclusionLR}, 
	we have for the ideal gas that
	\begin{equation} \label{eq:main-homo-ideal}
		e(\alpha, 0) \ge \frac{1}{4} f\bigl( (j_{\alpha_*}')^2 \bigr) 
		= \frac{1}{2} 2\pi \alpha_* \bigl(1 - O(\alpha_*^{1/3})\bigr).
	\end{equation}
	Moreover, for any fixed $\alpha \in \R \setminus \{0\}$ 
	we obtain in the dilute limit that
	\begin{equation} \label{eq:main-homo-limits}
		\liminf_{\bar\gamma \to 0} \frac{e(\alpha, \bar\gamma)}{2\pi |{\ln \bar\gamma}|^{-1}} \ge 1, 
		\qquad \text{and} \qquad
		\liminf_{\bar\gamma \to 0} e(\alpha, \bar\gamma) 
		\ge \frac{\pi}{81} (j_{c\alpha_*}')^2
		\ge \frac{c}{81} 2\pi \alpha_*, 
	\end{equation}
	while if\/ $\bar\gamma>0$ is arbitrary but fixed, and
	\begin{equation} \label{eq:soft-alpha-bound}
	     |\alpha| \le \eps^5 \min\bigl\{ \bar\gamma^2, \eps^3\bar\gamma^{-4} \bigr\}, 
	    \quad 0 < \eps < \sqrt{\pi}/8, 
	\end{equation}
	then
	\begin{equation} \label{eq:main-homo-soft}  
	   e(\alpha, \bar\gamma) \ge 2\pi|\alpha| (1 - O(\eps)).
	\end{equation}
\end{theorem}

Note that for the short-range part of the interaction, 
one can view the height of the potential compared to the average density
as a dimensionless interaction strength, and that in the dilute 
limit~\eqref{eq:main-homo-limits} with fixed $\alpha > 0$ we have that
$$
	\frac{\alpha}{R^2} / \bar\varrho = \alpha\bar\gamma^{-2} \to \infty, 
$$
corresponding to a hard-core interaction.
On the other hand, under the conditions in~\eqref{eq:soft-alpha-bound}, 
$$
	\frac{\alpha}{R^2} / \bar\varrho = \alpha\bar\gamma^{-2} 
	\le \eps^5 \ll 1, 
$$
and thus corresponding to a very weak soft-core interaction rather than a hard-core one
in this regime. 

We also note that the average-field description 
with its linear dependence on $\alpha$
has indeed been proved to be correct 
for the trapped anyon gas in a certain
almost-bosonic regime;
see~\cite{LunRou-15}.
In the present context this corresponds 
to taking $Q_0$ fixed, 
$\alpha \sim \beta/N$ and $R \sim N^{-\eta}$ with $0 < \eta < 1/4$, 
in which case we have that $\bar\gamma \sim N^{1/2-\eta} \to \infty$ and
$\alpha\bar\gamma^{-2} \sim N^{2\eta-2} \to 0$ as $N \to \infty$,
i.e.\ a combined high-density and weak soft-core limit.
However, the sense in which 
average-field theory then holds 
is that all the anyons become 
identically distributed subject to a self-consistent magnetic field, 
and it should
be remarked that the constant $2\pi$ that is predicted by the usual
(constant-field)
average-field approximation and which appears above
does not take such self-interactions fully
into account and may ultimately be replaced by a larger  
effective constant,
at least in a particular limit~\cite{CorLunRou-16}.

\begin{proof}[Proof of Theorem~\ref{thm:gas}]
Let us begin with the universal bound \eqref{eq:main-homo-universal} 
for all $\alpha, \bar\gamma$.
We have a sequence of $N \ge 1$ and squares $Q_0 \subseteq \R^2$ 
with $N/|Q_0| = \bar\varrho$, 
and consider in each case an arbitrary function $\Psi \in \domD{\alpha, R}$ 
supported on $Q_0^N$.
Let us again write
\begin{align} \label{eq:T-split-homo}
	T[\Psi] &:= \langle \Psi, \hat{T}_\alpha \Psi \rangle
	= \kappa_1 \sum_{j=1}^N \int_{\R^{2N}} |D_j\Psi|^2\, d\sx
	+ \kappa_2 \sum_{j=1}^N \int_{\R^{2N}} |D_j\Psi|^2\, d\sx
	+ \kappa_3 \sum_{j=1}^N \int_{\R^{2N}} |D_j\Psi|^2\, d\sx \\
	&\ge \int_{\R^{2N}} \sum_{j=1}^N \biggl(
		\kappa_1 \bigl|\nabla_j|\Psi|\bigr|^2
		+ \kappa_2 \sum_{k \neq j} 2\pi|\alpha| 
		\frac{\1_{B_R(0)}}{\pi R^2}(\bx_j - \bx_k) \, |\Psi|^2
		+ \kappa_3 |D_j\Psi|^2
		\biggr) d\sx.
\end{align}
Take $\kappa_1 = \kappa_2 = \kappa/2$ and $\kappa_3 = 1-\kappa$, 
and a partition of $Q_0$ into $M^2$ squares $Q$ of equal size. 
Then, by the local exclusion principles of
Lemma~\ref{lem:LocalExclusionSR} and Lemma~\ref{lem:LocalExclusionLR}, 
\begin{align} \label{eq:T-bound-homo}
	N^{-1} & T[\Psi] 
	\ge N^{-1} \sum_Q T_Q^{(\kappa/2, \kappa/2, 1-\kappa)}[\Psi] \\
	&\ge N^{-1} \sum_Q |Q|^{-1} \Bigl(
		\frac{\kappa}{2} \eSR\bigl(\alpha, \gamma(Q), {\textstyle \int_Q \varrho_\Psi}\bigr)
		+ (1-\kappa) \eLR\bigl(\alpha, \gamma(Q)\bigr) \Bigr)
		\biggl( \int_Q\varrho_\Psi \ -1 \biggr)_{\!\!+} \\
	&\ge 
		N^{-1}|Q_0|^{-1} M^2 \sum_Q
		\biggl( \int_Q\varrho_\Psi \ -1 \biggr)_{\!\!+} \times \\
		&\quad\times
		\left\{\begin{array}{ll}
			\frac{\kappa}{2} |\alpha| 
			\min\bigl\{(1-\gamma(Q)^2/2)^{-1}, K_\alpha/2\bigr\}
			\bigl( 
				K_\alpha + 2|\alpha|\bigl( -\ln (\gamma(Q)/
			&\hspace{-10pt}
				{\sqrt{2}}) \bigr)
			\bigr)^{-1} \\[8pt]
			\qquad +\ (1-\kappa) \frac{\pi}{24} g\bigl(c\alpha_N, 12\gamma(Q)\bigr)^2 (1-12\gamma(Q))_+^3, 
			&\hspace{10pt}\text{for}\ \gamma(Q) < \sqrt{2} \\[9pt]
			\kappa |\alpha|\gamma(Q)^{-2} \int_Q \varrho_\Psi, 
			&\hspace{10pt}\text{for}\ \gamma(Q) \ge \sqrt{2}.
		\end{array}\right.
\end{align}
Note that $\gamma(Q) = \bar\gamma MN^{-1/2}$ and we are free to choose
$\kappa \in [0, 1]$ and the integer $M \ge 1$ as we like.
We choose $M := \mu N^{1/2}$ for suitable $\mu > 0$, so that
$\gamma(Q) = \mu\bar\gamma$.
Then for $\mu < \min\{ \sqrt{2}/\bar\gamma, 1 \}$ we have, 
using $\sum_Q (\int_Q \varrho_\Psi \ - 1)_+ \ge (N - M^2)_+$, that
\begin{align} \label{eq:T-bound-homo-smallmu}
	N^{-1} T[\Psi] 
	\ge \bar\varrho \mu^2 (1 - \mu^2)_+ \biggl(
	&		\frac{\kappa}{2} |\alpha| \frac{
				\min\bigl\{(1-\mu^2\bar\gamma^2/2)_+^{-1}, K_\alpha/2\bigr\} }{
				K_\alpha + 2|\alpha|\bigl( -\ln (\mu\bar\gamma/{\sqrt{2}}) \bigr)
			} \\
	&		+ (1-\kappa) \frac{\pi}{24} g(c\alpha_N, 12\mu\bar\gamma)^2 (1-12\mu\bar\gamma)_+^3
		\biggr).
\end{align}
On the other hand for $\sqrt{2}/\bar\gamma \le \mu \le 1$, 
we may use 
$$
	\frac{1}{M^2} \sum_Q \int_Q \varrho_\Psi \biggl(\int_Q \varrho_\Psi \ - 1\biggr)_{\!\!+} 
	\ge \frac{N}{M^2} \biggl(\frac{N}{M^2} - 1\biggr)_{\!\!+}, 
$$
which follows from convexity, to obtain that
\begin{align} \label{eq:T-bound-homo-bigmu}
	N^{-1} T[\Psi] 
	\ge \kappa|\alpha|\bar\varrho \bar\gamma^{-2} (\mu^{-2} - 1)_+.
\end{align}
Hence, in the case $\bar\gamma \ge 2 > \sqrt{2}$ we can in 
the thermodynamic limit choose $\kappa = 1$ and $\mu = \sqrt{2}/\bar\gamma$ 
in order to obtain that
\begin{equation} \label{eq:T-bound-homo-bigmu-final}
	e(\alpha, \bar\gamma)
	\ge \frac{1}{2}|\alpha| (1-2/\bar\gamma^2) 
	\ge \frac{1}{4}|\alpha|, 
\end{equation}
while for $\bar\gamma < 2$ we choose, for simplicity, 
$\kappa = 2/3$ and $\mu = 1/{\sqrt{2}}$ 
obtaining that
\begin{equation} \label{eq:T-bound-homo-smallmu-final}
	e(\alpha, \bar\gamma) \ge 
	\frac{1}{288} \biggl(
		12 |\alpha| \frac{
			\min\bigl\{2(1-\bar\gamma^2/4)_+^{-1}, K_\alpha\bigr\} }{
			K_\alpha + 2|\alpha|( -\ln (\bar\gamma/2) )
		}
		+ \pi g(c\alpha_N, 12\bar\gamma/{\sqrt{2}})^2 (1-12\bar\gamma/{\sqrt{2}})_+^3
		\biggr).
\end{equation}
This proves the first part of the theorem with $C = 1/288$. 

In the ideal case $R=0$, and hence $\bar\gamma = 0$, we take $\kappa=0$
and $M \sim \sqrt{N/2}$ in~\eqref{eq:T-bound-homo} 
(which means approximately 2 particles in each box)
to obtain \eqref{eq:main-homo-ideal} from \eqref{eq:eLRideal} of 
Lemma~\ref{lem:LocalExclusionLR}.

\medskip

The second bound in \eqref{eq:main-homo-limits} follows immediately from
\eqref{eq:T-bound-homo-smallmu} and the properties of $g$, by 
setting $\kappa=0$ and $\mu=1/{\sqrt{2}}$.
For the first bound we set
$\kappa_1 = 1-\kappa$, $\kappa_2 = \kappa$ and $\kappa_3=0$
in~\eqref{eq:T-split-homo} and use the result~\cite{LieYng-01}
of Lieb and Yngvason for the dilute repulsive Bose gas in two dimensions.
We find for the (bosonic, and therefore positive; see~\cite[Corollary~3.1]{LieSei-09}) 
ground state $\Psi_0$ of this expression, 
with fixed $\kappa \in (0, 1)$ and $\alpha>0$, that
\begin{align} 
	\frac{T[\Psi]}{N\bar\varrho} &\ge 
	\frac{1-\kappa}{N\bar\varrho} \int_{\R^{2N}} \biggl(
		\sum_{j=1}^N \bigl|\nabla_j\Psi_0\bigr|^2
		+ \sum_{j < k} W(\bx_j - \bx_k) \, |\Psi_0|^2
		\biggr) d\sx \\
	&= \frac{4\pi(1-\kappa)}{|{\ln a_R^2\bar\varrho}|} \bigl( 1+O\bigl(|{\ln a_R^2\bar\varrho}|^{-1/5}\bigr) \bigr)
	= \frac{2\pi(1-\kappa)}{K'_{\alpha, \kappa} - \ln\bar\gamma} 
		\bigl( 1+O\bigl((K'_{\alpha, \kappa} - \ln\bar\gamma)^{-1/5}\bigr) \bigr),  
\end{align}
where we used that the pair potential
$$
	W(\bx) := \frac{W_0}{R^2} \1_{B_R(0)}(\bx), 
	\qquad W_0 = 4\alpha \kappa/(1-\kappa), 
$$
has scattering length (cf.\ \eqref{eq:soft-disk-scattering-length})
$$
	a_R = R\exp\biggl( -\frac{1}{\sqrt{W_0/2}} \frac{I_0(\sqrt{W_0/2})}{I_1(\sqrt{W_0/2})} \biggr)
	= R\exp( -K'_{\alpha, \kappa}), 
$$
with
$$
	K'_{\alpha, \kappa}
	:= \frac{1}{\sqrt{2\alpha\kappa/(1-\kappa)}} \frac{I_0(\sqrt{2\alpha\kappa/(1-\kappa)})}{I_1(\sqrt{2\alpha\kappa/(1-\kappa)})}
	= \frac{ K_{\alpha\kappa/(1-\kappa)} }{2\alpha\kappa/(1-\kappa)}.
$$
Hence for any $\alpha>0$ and $0<\eps \ll 1$ we, by setting $\kappa = \eps$
and then taking the limit $\bar\gamma \to 0$, obtain that
$$
	\frac{|{\ln \bar\gamma}|}{2\pi} e(\alpha, \bar\gamma) 
	\ge (1-\eps) \bigl( 1 + K'_{\alpha, \eps} |{\ln \bar\gamma}|^{-1} \bigr)^{-1} 
		\bigl( 1 + O\bigl( (K'_{\alpha, \eps} + |{\ln \bar\gamma}|)^{-1/5} \bigr) \bigr)
	\to 1-\eps.
$$
So for each fixed $\alpha \in \R \setminus \{0\}$
$$
	\liminf_{\bar\gamma \to 0} \frac{e(\alpha, \bar\gamma)}{2\pi |{\ln \bar\gamma}|^{-1}} \ge 1.
$$

\medskip

To obtain the bound \eqref{eq:main-homo-soft} for the soft-core regime we 
follow~\cite{LieYng-98, LieYng-00, LieYng-01, LieSeiSolYng-05}. 
Again we partition $Q_0$ into $M^2$ squares $Q$ of equal size, 
and let $\ell = |Q|^{1/2}$. With $\kappa\in [0, 1]$ we then have that
\begin{align}
	N^{-1}T[\Psi] 
		&\ge 
	N^{-1}\sum_Q T_Q^{(\kappa, 1-\kappa, 0)}[\Psi]
		\ge
	N^{-1}\sum_Q\sum_{n\ge0} E_n^{(\kappa, 1-\kappa, 0)}(|Q|)p_n(\Psi;Q).
\end{align}
Set $c_n = \sum_Q p_n(\Psi;Q) |Q|/|Q_0|$, i.e.\ $c_n$ is the fraction of cells $Q$ containing precisely $n$ particles, then
\begin{equation}
	\sum_{n\ge 0} c_n =1 \quad \mbox{and} \quad \sum_{n\ge 0} c_n n = \bar\varrho \ell^2.
\end{equation}
Rearranging the sum and from now on suppressing the weight 
$\bkappa = (\kappa, 1-\kappa, 0)$ we find that
\begin{align}\label{eq:soft_core_hom_sum}
	N^{-1}T[\Psi] 
		&\ge 
	\frac{1}{\bar\varrho\ell^2} \sum_{n\ge 0}E_n(|Q|) c_n, 
\end{align}
which is precisely the starting point of the argument in~\cite{LieYng-98, LieYng-00, LieYng-01}.

Fix $p\in \NN$. Since the energy is superadditive, $E_{n+n'} \ge E_{n}+E_{n'}$, we for all $n\ge p$ have that 
\begin{equation}
	E_n(|Q|) \ge \lfloor n/p\rfloor E_p(|Q|) \ge \frac{n}{2p}E_p(|Q|).
\end{equation}
Applying Lemma~\ref{lem:TempleSmearBound} yields
\begin{equation}
	E_n(|Q|)\ge \pi |\alpha|\frac{n(p-1)}{\ell^2}K(p, \ell), 
\end{equation}
where 
\begin{equation}
	K(n, \ell) := (1-\kappa)\biggl(1-\frac{2R}{\ell}\biggr)^{\!2}_{\!\!+} 
	\biggl(1- \frac{2|\alpha| \ell^2 R^{-2} n(n-1)}{\pi^2 \kappa/(1-\kappa)-2\pi |\alpha| n(n-1)}\biggr)_{\!\!+}, 
\end{equation}
if the expression in the last denominator is positive and $K(n, \ell):=0$
otherwise.

If instead $n< p$ we use that $K(n, \ell)$ is decreasing in $n$ to find
\begin{equation}
	E_n(|Q|)\ge 2\pi |\alpha|\frac{n(n-1)}{\ell^2}K(p, \ell).
\end{equation}

Splitting the sum~\eqref{eq:soft_core_hom_sum} into two we thus find that
\begin{align}
	\sum_{n\ge 0}E_n(|Q|)c_n 
		&= 
	\sum_{n<p} E_n(|Q|)c_n + \sum_{n\ge p} E_n(|Q|)c_n\\
		&\ge
	\frac{2\pi |\alpha|}{\ell^2}K(p, \ell)\biggl( \sum_{n<p}n(n-1)c_n + \frac{1}{2}\sum_{n\ge p}n(p-1)c_n \biggr).
\end{align}
We wish to minimize
\begin{equation}\label{eq:soft_core_sum_minimize}
	\sum_{n<p}n(n-1)c_n + \frac{1}{2}\sum_{n\ge p}n(p-1)c_n.
\end{equation}
Set
\begin{equation}
	k:=\bar\varrho\ell^2 \quad \mbox{and} \quad t:= \sum_{n<p} c_n n\le k, 
\end{equation}
by convexity~\eqref{eq:soft_core_sum_minimize} is then larger than
\begin{equation}
	t(t-1)+\frac12(k-t)(p-1).
\end{equation}
If $p\ge 4k-1$ and $t\le k$ this is minimized at $t=k$, where it is equal to $k(k-1)$. 
Thus by choosing $p=\lfloor 4\bar\varrho\ell^2 \rfloor$ we have shown that
\begin{equation}
	N^{-1}T[\Psi] 
	\ge 
	2\pi |\alpha|\bar\varrho\Bigl(1-\frac{1}{\bar\varrho\ell^2}\Bigr)_{\!+} K(4\bar\varrho\ell^2, \ell), 
\end{equation}
and hence, upon taking the thermodynamic limit $N, |Q_0| \to \infty$ 
with all the other parameters kept fixed, 
\begin{equation}\label{eq:soft_core_bound_ell}
	e(\alpha, \bar\gamma) \ge 2\pi|\alpha| (1-\kappa)
		\Bigl(1-\frac{1}{\bar\varrho\ell^2}\Bigr)_{\!+}
		\Bigl(1-2\frac{\bar\gamma}{\bar\varrho^{1/2}\ell}\Bigr)_{\!+}^2
		\biggl(1-\frac{32|\alpha|\bar\gamma^{-2} \bar\varrho^3\ell^6}{\pi^2\kappa/(1-\kappa) - 32\pi|\alpha|\bar\varrho^2\ell^4}\biggr)_{\!\!+}, 
\end{equation}
as long as $32|\alpha|\bar\varrho^2\ell^4 < \pi\kappa/(1-\kappa)$.

Given $\eps>0$, let us choose $\kappa = \eps$  
and also demand that $(\bar\varrho\ell^2)^{-1} \le \eps$, 
$\bar\gamma(\bar\varrho^{1/2}\ell)^{-1} \le \eps$, 
$|\alpha|\bar\gamma^{-2} \bar\varrho^3\ell^6 \le \eps^2$, 
and $|\alpha|\bar\varrho^2\ell^4 \le \eps\pi/64$.
We therefore choose
$$
	\ell = (\eps\bar\varrho)^{-1/2} \max\bigl\{1, \eps^{-1/2}\bar\gamma\bigr\}
$$
and then find that, together with the requirement
\eqref{eq:soft-alpha-bound} on $\alpha$ and $\eps$
which implies $|\alpha|\bar\varrho^2\ell^4 \le \eps^3 < \eps\pi/64$, 
all conditions above are satisfied, and the error terms in
\eqref{eq:soft_core_bound_ell} are of order $\eps$ or higher.
\end{proof}

\FloatBarrier

\appendix
\section{Some properties of Bessel functions} \label{app:bessel}

\begin{proposition}\label{prop:BesselPrimeZeroBounds}
	For $\nu>0$ we let $j_\nu'$ denote the first positive zero of the derivative 
	of the Bessel function $J_\nu$. Then we have that
	\begin{equation}
		\sqrt{2\nu} \le j_\nu' \le \sqrt{2\nu(1+\nu)}.
	\end{equation}
\end{proposition}
A proof of the above proposition and much more refined bounds for the zeros of Bessel 
functions and their derivatives can be found in~\cite{IsmMul-95}. 
For completeness we provide an elementary proof which covers our needs.

\begin{proof}
By a standard variational argument it can be shown that
\begin{equation}
	\inf_{u}\frac{\int_0^1 \bigl(|u'|^2+\nu^2r^{-2}|u|^2\bigr)r\, dr}{\int_0^1 |u|^2r\, dr}=(j_\nu')^2,
\end{equation}
where the infimum is taken over all $u\in W^{1, 2}([0, 1], rdr)$ and is attained by $u(r)=J_\nu (j_\nu' r)$.

For $\nu>0$ and $u\in W^{1, 2}([0, 1], rdr)$ with $u(0)=0$ we obtain using H\"older's inequality that
\begin{align}
	|u(t)|^2 
	&= 
	2 \Re\biggl[ \int_0^t \bar u (r) u'(r) \, dr\biggr]\\
	&\le 
	2 \biggl(\int_0^t|u'(r)|^2r\, dr \biggr)^{1/2}\biggl(\int_0^t|u(r)|^2r^{-1}\, dr\biggr)^{1/2}\\
	&= 
	\frac{2}{\nu} \biggl(\int_0^t|u'(r)|^2r\, dr \biggr)^{1/2}\biggl(\nu^2 \int_0^t|u(r)|^2r^{-1}\, dr\biggr)^{1/2}.
\end{align}
Through an application of Young's inequality we then find
\begin{align}
	|u(t)|^2 
	\le
	\frac{1}{\nu} \int_0^t \Bigl( |u'(r)|^2+\frac{\nu^2}{r^2}|u(r)|^2\Bigr)r\, dr
	\le
	\frac{1}{\nu}\int_0^{1} \Bigl( |u'(r)|^2+\frac{\nu^2}{r^2}|u(r)|^2\Bigr)r\, dr,
\end{align}
and integrating both sides in $t$ over $(0, 1)$ against $t\, dt$ yields
\begin{align}
	\int_0^{1} |u(t)|^2t\, dt 
	&\le
	\frac{1}{\nu}\biggl(\int_0^{1} t\, dt\biggr) \biggl( \int_0^1 \Bigl( |u'(r)|^2+\frac{\nu^2}{r^2}|u(r)|^2\Bigr)r\, dr\biggr)\\
	&=
	\frac{1}{2\nu} \int_0^{1} \Bigl( |u'(r)|^2+\frac{\nu^2}{r^2}|u(r)|^2\Bigr)r\, dr,
\end{align}
which implies that
\begin{equation}
	\frac{\int_0^1 \bigl(|u'|^2+\nu^2r^{-2}|u|^2\bigr)r\, dr}{\int_0^1 |u|^2r\, dr} \ge 2\nu.
\end{equation}
Taking the infimum over all functions $u\in W^{1,2}([0, 1], rdr)$ such that $u(0)=0$, in particular this includes $J_\nu$, we see that
\begin{equation}
	(j_\nu')^2\ge 2 \nu,
\end{equation}
which completes the proof of the lower bound. 
To obtain the upper bound, simply take $u(r)=r^\nu$ in the variational quotient above.
\end{proof}

In the case of $R$-extended anyons our bounds result in studying the behavior 
of solutions to a Bessel-type eigenvalue equation of order $\nu$ with Neumann 
boundary conditions on the interval $(\gamma, 1)$, for some $0<\gamma<1$. 
Thus it is of interest for us to understand the behavior of the lowest 
eigenvalue of such an equation in both parameters $\gamma$ and $\nu$.

\begin{proposition}\label{prop:BesselEigenvalueBounds}
	Given $\nu > 0$ and $0 < \gamma < 1$, let 
	$g(\nu,\gamma) := \sqrt{\lambda}$, where $\lambda$
	denotes the first positive solution to the eigenvalue equation
	\begin{equation}\label{eq:BesselEquationInterval}
			-u''(r)-\frac{u'(r)}{r}+\Bigl(\frac{\nu^2}{r^2}-\lambda\Bigr)u(r)=0,
	\end{equation}
	with the Neumann boundary conditions $u'(\gamma)=u'(1)=0$. 
	Then, for fixed $\gamma$, $g(\nu,\gamma)$ is a monotonically increasing
	function in $\nu$.
	Also, for fixed $\nu$, $g(\nu,\gamma)$ is a monotonically decreasing 
	function of $\gamma$, and satisfies
	\begin{equation}
		\nu < g(\nu,\gamma) < \min\{j_\nu', \nu/\gamma\}.
	\end{equation}
	Moreover, we have that\/ $\lim_{\gamma\to 0} g(\nu,\gamma) = j_\nu'$ 
	and\/ $\lim_{\gamma\to 1} g(\nu,\gamma) = \nu$.
\end{proposition}

\begin{proof}
	That $g(\nu, \gamma)$ is monotonically increasing in $\nu$ is clear from the variational characterization of $\lambda$,
	\begin{equation}
		\lambda = \inf_u \frac{\int_{\gamma}^1 \bigl(|u'|^2+\nu^2r^{-2}|u|^2\bigr)r\,dr}{\int_\gamma^1 |u|^2r\,dr}. 	
	\end{equation}

	It is well known that the solution of the above differential equation is 
	given by a linear combination of the Bessel functions 
	$J_\nu(\sqrt{\lambda}r)$ and $Y_\nu(\sqrt{\lambda}r)$. 
	Only if $\gamma$ were zero could we
	exclude the Bessel function of the second kind since it fails to be in $W^{1,2}([0,1],r dr)$
	and thus cannot be a solution. Thus the problem reduces to finding the 
	smallest $\lambda>0$ such that the system
	\begin{align}
		\alpha J_\nu'(\sqrt{\lambda}\gamma)+\beta Y_\nu'(\sqrt{\lambda}\gamma)&=0\\
		\alpha J_\nu'(\sqrt{\lambda})+\beta Y_\nu'(\sqrt{\lambda})&=0
	\end{align}
	admits a non-trivial solution, which is equivalent to the determinant equation
	\begin{equation}
		J_\nu'(\sqrt{\lambda}\gamma)Y_\nu'(\sqrt{\lambda})-Y_\nu'(\sqrt{\lambda}\gamma)J_\nu'(\sqrt{\lambda})=0.
	\end{equation}
	Assuming that $\sqrt{\lambda}$ is smaller than the first zero of $Y_\nu'$ (this will be seen to be true once we find our solution) we can equivalently solve the equation
	\begin{equation}
		\frac{J_\nu'(\sqrt{\lambda})}{Y_\nu'(\sqrt{\lambda})} = \frac{J_\nu'(\sqrt{\lambda}\gamma)}{Y_\nu'(\sqrt{\lambda}\gamma)}.
	\end{equation}
	Letting $G_\nu(x) := J'_\nu(x)/Y_\nu'(x)$ we find that
	\begin{align}
		G'_\nu(x) = \frac{2(\nu^2-x^2)}{\pi x^3 Y_\nu'(x)^2},
	\end{align}
	where we used that $J_\nu$ and $Y_\nu$ satisfy the Bessel equation~\eqref{eq:BesselEquationInterval}
	and the well-known identity 
	$ J_\nu(x) Y'_\nu(x)-J'_\nu(x) Y_\nu(x) = 2/(\pi x)$; see e.g.~\cite[Eqn.~10.5.2]{OlvMax-10}.

	Thus $G_\nu(x)$ is strictly increasing on $(0, \nu)$ and decreasing after 
	that. We also know that $G_\nu(0)=G_\nu(j_\nu')=0$. But then it is 
	clear that the graph of $G_\nu(x)$ and that of its dilation 
	$G_\nu(\gamma x)$ must intersect between $x=\nu$ and the minimum of
	$x=\nu/\gamma$ and $x=j_\nu'$ 
	(compare Figure~\ref{fig:BesselDilationEquationPlot}), and as this solution is 
	less than the first zero of $Y_\nu'$ the assumption above is seen to be true. 
	Moreover, as $\gamma \to 0$ we see that the solution $x=\sqrt{\lambda}$ 
	tends to the zero $j_\nu'$ and if instead $\gamma \to 1$ it tends to 
	the maximum point $\nu$.

\begin{figure}[ht]
	\centering
	\begin{tikzpicture}
		\node [above right] at (0,0) {\includegraphics[scale=0.86,clip,trim=-0ptcm 10pt -0pt 0pt]{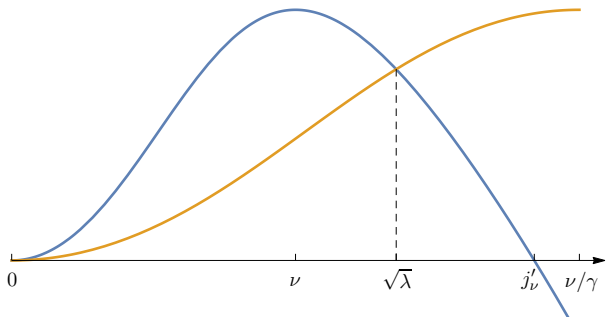}};
		\node [above right] at (-0.035,0.4) {\scalebox{0.7}{$0$}};
		\node [above right] at (3.68,0.4) {\scalebox{0.7}{$\nu$}};
		\node [above right] at (4.9,0.35) {\scalebox{0.7}{$\sqrt{\lambda}$}};
		\node [above right] at (6.73,0.35) {\scalebox{0.7}{$j_\nu'$}};
		\node [above right] at (7.3,0.32) {\scalebox{0.7}{$\nu/\gamma$}};
	\end{tikzpicture}
	\caption{The function $G_\nu(x)$ (blue) and its dilation $G_\nu(\gamma x)$ (yellow) plotted for $\nu=1$ and $\gamma=1/2$.}
	\label{fig:BesselDilationEquationPlot}
\end{figure}

	By the above geometric considerations we can conclude that for $0<\gamma<1$ and $\nu>0$ we have that $\lambda$, the smallest positive eigenvalue of~\eqref{eq:BesselEquationInterval}, satisfies
	\begin{equation}
		\lambda \in [\nu^2, \min\{j_\nu', \nu/\gamma\}^2],
	\end{equation}
	and is monotonically decreasing in $\gamma$.
\end{proof}

\FloatBarrier

\section{Concavity of the one-particle profile} \label{app:concavity}

	We have several times used concavity properties of the one-particle profile $f(d, \cdot\,)$, which however may fail if $d$ is small. More precisely we have the following: 
	\begin{proposition}\label{prop:Concavity}
	For any $d\ge 0$ the function $f(d, \cdot\,)$ given by~\eqref{eq:oneParticleProfile} is concave on its support intersected with $[R, \infty)$. If in addition $d\ge R$ the function is concave on its full support $[d-R, d+R]$. 
	\end{proposition}

	\begin{proof}
	Without loss of generality we may, and do, assume that $R=1$. 
	The proof is then a straightforward computation. We begin with assuming 
	that $d<R=1$. For such $d$ the function $f(d, \cdot\,)$ is $C^2$ on $[1, d+1]$ 
	(and zero on $(d+1,\infty)$)
	which reduces the statement to proving that $\partial_{r}^2f(d, r)\le 0$ 
	in this region. Calculating this derivative one finds
	\begin{align*}
	 	\partial_{r}^2f(d, r)= -\frac{2((d^2-1)^3-3(d^2-1)^2r^2+(5+3d^2)r^4-r^6)}{\pi r ((r+1-d)(1+d-r)(d+r-1)(1+d+r))^{3/2}},
	\end{align*} 
	and clearly the overall sign is determined by that of the polynomial in the denominator
	\begin{equation}
	 	p(d, r):= (d^2-1)^3-3(d^2-1)^2r^2+(5+3d^2)r^4-r^6.
	 \end{equation}
	We need to prove that $p\ge 0$ for $(r, d)$ in the triangular region given by $1\le r\le d+1$ where $0\le d\le 1$.

	We first check the statement on the boundary of the region:
	\begin{align*}
	 	p(1, r)&=r^4(8-r^2)>0\\
	 	p(d, 1)&=d^2(12-6d^2+d^4)>0\\
	 	p(d, d+1)&=(1+r)(r-1)(4r^2+1-r^4)>0.
	\end{align*}
	Thus all that remains is to check that we have no stationary points for $p$ in the interior of the region. Calculating the derivative in $r$ one finds that
	\begin{equation}
		\partial_r p(d, r)=6d(d^2-1)^2-12d(d^2-1)r^2+6dr^4.
	\end{equation}
	As this is a quadratic polynomial in $r^2$ we can solve the equation $p_r(d, r)=0$ and find that there are no solutions in our region. 
	This completes the proof of the claim in the case $d<R=1$.

	In the case $d\ge R=1$ we wish to prove that $f(d, \cdot\,)$ is concave on 
	$[d-1, d+1]$.
	It is here convenient to study the problem in the variables $d$ and $\eta=r-d$, and letting
	\begin{equation}
	 	g(d, \eta) := f(d, d+\eta) 
	 	= \frac{2(d+\eta)}{\pi}\arccos\biggl(\frac{d^2+(d+\eta)^2-1}{2d(d+\eta)}\biggr), 
	 	\quad d\ge 1, \ \eta \in [-1, 1].  
	\end{equation}
	Differentiating twice in $\eta$ we find that
	\begin{equation}\label{eq:second_derivative_relative_profile}
		\partial_\eta^2g(d,\eta)=\frac{2P(d, \eta)}{\pi  (d+\eta) ((1-\eta^2) (2 d+\eta-1) (2 d+\eta+1))^{3/2}},
	\end{equation}
	where
	\begin{equation}
		P(d, \eta):= -8 d^4+8 d^3(\eta^3-4\eta)+12 d^2(\eta^4-3\eta)+d (6 \eta^5-20 \eta^3+6\eta)+\eta^6-5 \eta^4+3\eta^2+1.
	\end{equation}
	As before the sign of~\eqref{eq:second_derivative_relative_profile} is determined by that of the polynomial $P$. If we can prove that $P(d, \eta)\le 0$ for all $d\ge 1$ and $-1\le \eta \le 1$ the claim follows. To this end we proceed as above. The values of $P$ on the boundaries of this region are (in the same manner as before) readily checked to be negative: 
	\begin{align}
		P(d, 1)&= -8 d (d+1)^3<0,\\
		P(d, -1)&= -8d(d-1)^3\le 0,\\
		P(1, \eta) &= (1+\eta)^4(\eta^2+2\eta-7)\le 0.
	\end{align}

	What remains is to check stationary points in the interior. 
	For this polynomial solving either of the equations $\partial_\eta P(d, \eta)=0$ or $\partial_d P(d, \eta)=0$ is slightly harder. 
	However, since certain terms cancel one can instead solve the equation
	\begin{equation}
		\partial_\eta P(d, \eta)=\partial_d P(d, \eta),
	\end{equation}
	and the solutions are $d=0, \eta=-d-\sqrt{d^2-1}$ and $\eta=-d+\sqrt{d^2-1}$. The third solution is the only one contained within our region. Evaluating the derivative at this solution we obtain
	\begin{equation}
		\partial_\eta P(d, -d+\sqrt{d^2-1})= -32(d^2-1)^{3/2},
	\end{equation}
	and since this is non-zero in the interior of our domain the proof is complete.
\end{proof}

\medskip


\bibliographystyle{siam}

\end{document}